\documentclass{lmcs} 
\pdfoutput=1

\usepackage{lastpage}
\lmcsdoi{18}{4}{1}
\lmcsheading{}{\pageref{LastPage}}{}{}%
{Dec.~04,~2020}{Oct.~21,~2022}{}


\usepackage[utf8]{inputenc}
\usepackage{hyperref}
\usepackage{doi}
\usepackage{amsthm}
\usepackage{amssymb}
\usepackage{amsmath}
\usepackage[noxy]{virginialake}
\odbackgroundtrue
\odframefalse
\odbackgroundfirstfalse
\odframefirstfalse
\vlsmallleftlabels
\vlsmallbrackets

\usepackage{amsfonts}
\usepackage{xypic}
\usepackage{tikz}
\usetikzlibrary{shapes}

\makeatletter
\def\pgfsys@hboxsynced#1{%
	{%
		\pgfsys@beginscope%
		\setbox\pgf@hbox=\hbox{%
			\hskip\pgf@pt@x%
			\raise\pgf@pt@y\hbox{%
				\pgf@pt@x=0pt%
				\pgf@pt@y=0pt%
				\pgflowlevelsynccm%
				\pgfsys@hbox#1}%
			\hss%
		}%
		\wd\pgf@hbox=0pt%
		\ht\pgf@hbox=0pt%
		\dp\pgf@hbox=0pt%
		\box\pgf@hbox%
		\pgfsys@endscope%
	}%
}
\makeatother

\usepackage{url}
\usepackage{cmll}
\usepackage{verbatim}
\usepackage{adjustbox}
\usepackage{mathtools}
\usepackage{cleveref}
\usepackage{standalone}
\usepackage{enumitem}
\usepackage{xspace}

\usepackage{thmtools,thm-restate}

\usepackage{ulem}
\newcommand{\isom}[1][]{\simeq_{#1}}

\def\L{\mathsf{L}}
\def\sysS{\mathsf{S}}

\newcommand{\proves}[2][]{\mathord{\vdash_{#1}\,}{#2}}

\def\lun{\mathord{\circ}}
\def\lbox{\Box}
\def\limp{\multimap}
\def\lcoimp{\multimapinv}

\newcommand{\rcomp}[1]{\mathop{\circ_{#1}}}
\newcommand{\fcomp}{\mathop{\circ}}
\renewcommand{\emptyset}{\varnothing}
\def\unit{\emptyset}
\def\GS{\mathsf{GS}}
\def\SGS{\mathsf{SGS}}
\def\SGSd{\mathsf{SGS}\mathord{\downarrow}}
\def\SGSu{\mathsf{SGS}\mathord{\uparrow}}

\def\defn#1{\emph{\bfseries\boldmath #1}}

\newcommand{\fequiv}{\equiv}

\newcommand{\connn}[1]{\llparenthesis#1\rrparenthesis}
\newcommand{\Connn}[1]{\begin{pmatrix}\!\begin{vmatrix}#1\end{vmatrix}\!\end{pmatrix}}  
\newcommand{\coons}[2]{[#1]_{#2}}														
\newcommand{\Coons}[2]{\!\begin{bmatrix}#1\end{bmatrix}_{\!#2}}  						
\newcommand{\coonso}[1]{[\cdot]_{#1}}

\newcommand{\quand}{\quad\mbox{and}\quad}
\newcommand{\qquand}{\qquad\mbox{and}\qquad}

\newcommand{\quor}{\quad\mbox{or}\quad}
\newcommand{\qquor}{\qquad\mbox{or}\qquad}

\newcommand{\quiff}{\quad\iff\quad}

\newcommand{\qomma}{\;,\;}
\newcommand{\quadfs}{\rlap{\quad.}}
\def\idr{\dirule}
\def\iur{\uirule}
\def\aidr{\dairule}
\def\aiur{\uairule}
\def\ssdr{\dssrule}
\def\ssur{\ussrule}
\def\pdr{\dprule}
\def\pur{\uprule}
\def\gdr{\dgrule}
\def\gur{\ugrule}
\def\rr{\mathsf{r}}

\def\ods#1#2#3#4{\od{\odd{\odh{#1}}{#2}{#3}{#4}}}
\def\mod#1{\od{#1}}

\newcommand{\nvls}[1]{#1}
\renewcommand{\vlpa}{\vlbin{\lpar}}%

\def\set#1{\{#1\}}

\def\multiset#1{\{\mkern-5mu|#1|\mkern-5mu\}}
\def\NP{{\bf NP}}

\newcommand{\intset}[2]{\set{#1,\dots, #2}}

\def\tuple#1{\langle#1\rangle}

\def\cneg#1{{#1^\bot}}
\def\cnegg#1{{#1^{\bot\bot}}}
\def\cnegl#1{#1\rlap{$^\bot$}}

\def\lpar{\vlbin{\parr}}
\def\ltens{\vlbin{\otimes}}

\def\axrule{\mathsf {ax}}
\def\cutr{\mathsf {cut}}
\def\tensrule{\ltens}
\def\parrule{\lpar}

\def\airule{\mathsf{ai}}
\def\dairule{\airule\mathord{\downarrow}}
\def\uairule{\airule\mathord{\uparrow}}

\def\irule{\mathsf{i}}
\def\dirule{\irule\mathord{\downarrow}}
\def\uirule{\irule\mathord{\uparrow}}

\def\grule{\mathsf{g}}
\def\dgrule{\grule\mathord{\downarrow}}
\def\ugrule{\grule\mathord{\uparrow}}

\def\swir{\mathsf{sw}}

\def\ssrule{\mathsf{s\mkern-1mu.\mkern-1mu sw}}
\def\ussrule{\ssrule\mathord{\uparrow}}
\def\dssrule{\ssrule\mathord{\downarrow}}

\def\prule{\mathsf{p}}
\def\dprule{\prule\mathord{\downarrow}}
\def\uprule{\prule\mathord{\uparrow}}

\def\MLL{\mathsf{MLL}}

\def\mixr{\mathsf{mix}}
\def\MLLm{\mathsf{MLL^\circ}}

\def\gA{A}
\def\gB{B}
\def\gC{C}

\def\gG{G}
\def\gH{H}
\def\gK{K}

\def\gL{L}
\def\gM{M}
\def\gN{N}
\def\gP{P}
\def\gQ{Q}

\def\gX{X}

\def\cN{\mathcal N}
\def\cM{\mathcal M}

\newcommand{\context}[1]{C[#1]}

\newcommand{\vertices}[1][]{V_{#1}}
\newcommand{\gEdges}[1][]{E_{#1}}

\newcommand{\vB}{\vertices[B]}

\newcommand{\vC}{\vertices[C]}

\newcommand{\vG}{\vertices[G]}
\newcommand{\eG}{\gEdges[G]}

\newcommand{\enG}{\gEdges[\cneg G]}
\newcommand{\vH}{\vertices[H]}
\newcommand{\eH}{\gEdges[H]}

\newcommand{\enH}{\gEdges[\cneg H]}
\newcommand{\vHi}{\vertices[H_i]}
\newcommand{\eHi}{\gEdges[H_i]}
\newcommand{\vHj}{\vertices[H_j]}

\newcommand{\vM}{\vertices[M]}
\newcommand{\eM}{\gEdges[M]}
\newcommand{\vN}{\vertices[N]}

\newcommand{\vGp}{\vertices[G']}
\newcommand{\eGp}{\gEdges[G']}
\newcommand{\vP}{\vertices[P]}

\newcommand{\vQ}{\vertices[Q]}

\newcommand{\cA}{\mathcal{A}}
\newcommand{\cV}{\mathcal{V}}

\newcommand{\dD}{\mathcal{D}}
\newcommand{\tdD}{\tilde{\mathcal{D}}}

\newcommand{\graphof}[1]{\llbracket#1\rrbracket}
\newcommand{\sizeof}[1]{|#1|}
\newcommand{\labelof}[2][]{\ell_{#1}(#2)}
\newcommand{\lG}[1]{\labelof[G]{#1}}
\newcommand{\lnG}[1]{\labelof[\cneg G]{#1}}

\def\cna{\cneg a}
\def\cnP{\cneg P}

\def\graphcolor{red}
\def\gdualraphcolor{violet}
\def\interactioncolor{blue}
\def\dircolor{green}

\tikzstyle{graphstyle}=[>=stealth,overlay,remember picture,thin, draw=\graphcolor,fill=\graphcolor,opacity=1]  
\tikzstyle{complgraphstyle}=[densely dotted,>=stealth,overlay,remember picture, draw=\gdualraphcolor,fill=\gdualraphcolor,opacity=1]  
\tikzstyle{interactionstyle}=[>=stealth,overlay,remember picture,very thick, draw=\interactioncolor,opacity=1]  
\tikzstyle{dirgraphstyle}=[>=stealth,overlay,remember picture,thin, draw=\dircolor,fill=\dircolor,opacity=1]

\tikzstyle{edgestyle}=[>=stealth,overlay,remember picture,thin, opacity=1]

\newcommand{\va}[1]{\flowvertex{a}{a#1}}
\newcommand{\vna}[1]{\flowvertex{\cnegl a}{na#1}}
\newcommand{\vb}[1]{\flowvertex{b}{b#1}}
\newcommand{\vnb}[1]{\flowvertex{\cnegl b}{nb#1}}
\newcommand{\vc}[1]{\flowvertex{c}{c#1}}
\newcommand{\vnc}[1]{\flowvertex{\cnegl c}{nc#1}}
\newcommand{\vd}[1]{\flowvertex{d}{d#1}}
\newcommand{\vnd}[1]{\flowvertex{\cnegl d}{nd#1}}
\newcommand{\ve}[1]{\flowvertex{e}{e#1}}
\newcommand{\vne}[1]{\flowvertex{\cnegl e}{ne#1}}
\newcommand{\vf}[1]{\flowvertex{f}{f#1}}
\newcommand{\vnf}[1]{\flowvertex{\cnegl f}{nf#1}}
\newcommand{\vg}[1]{\flowvertex{g}{g#1}}
\newcommand{\vng}[1]{\flowvertex{\cnegl g}{ng#1}}
\newcommand{\vh}[1]{\flowvertex{h}{h#1}}
\newcommand{\vnh}[1]{\flowvertex{\cnegl h}{nh#1}}

\newcommand{\vu}[1]{\flowvertex{\unit}{u#1}}

\newcommand{\vm}[1]{\flowvertex{\bullet}{m#1}}
\newcommand{\vbul}[1]{\flowvertex{\bullet}{n#1}}

\newcommand{\flowvertex}[2]{\mathord{%
    \tikz[remember picture,baseline=(#2\nodecode.base)]%
    \node[inner sep=0pt](#2\nodecode){$#1\strut$};}}

\newcommand{\edge}[2]{%
  \tikz[graphstyle]%
  \draw[-] (#1\nodecode) -- (#2\nodecode);
}

\newcommand{\edges}[1]{
  \foreach \aaa/\bbb in {#1} {\edge{\aaa}{\bbb}}
}

\newcommand{\bentedge}[3]{%
	\tikz[graphstyle]%
	\draw[-] (#1\nodecode) to [bend left = #3]  (#2\nodecode);
}

\newcommand{\bentedges}[1]{
	\foreach \aaa/\bbb/\ccc in {#1} {\bentedge{\aaa}{\bbb}{\ccc}}
}

\newcommand{\blackedge}[2]{%
	\tikz[>=stealth,overlay,remember picture,thick,opacity=1]%
	\draw[-] (#1\nodecode) -- (#2\nodecode);
}

\newcommand{\blackedges}[1]{
	\foreach \aaa/\bbb in {#1} {\blackedge{\aaa}{\bbb}}
}

\newcommand{\blackbendedge}[2]{%
	\tikz[>=stealth,overlay,remember picture,thick,opacity=1]%
	\draw (#1\nodecode) to [bend left] (#2\nodecode);
}

\newcommand{\blackbendedges}[1]{
	\foreach \aaa/\bbb in {#1} {\blackbendedge{\aaa}{\bbb}}
}

\newcommand{\joinedges}[2]{
  \foreach \aaa in {#1} \foreach \bbb in {#2} {\edge{\aaa}{\bbb}}
}

\newcommand{\intedge}[2]{%
  \tikz[interactionstyle]%
  \draw (#1\nodecode) -- (#2\nodecode);
}

\newcommand{\intedges}[1]{
  \foreach \aaa/\bbb in {#1} {\intedge{\aaa}{\bbb}}
}

\newcommand{\modulevertex}[2]{
	\mathord{%
		\tikz[remember picture,baseline=(#2\nodecode.base)]%
		\node[draw,rounded corners,inner sep=2pt](#2\nodecode){\;$#1\strut$\;};
	}
}

\newcommand{\graphvertex}[2]{
		\tikz[baseline=(#2\nodecode.base)]%
		\node[draw,densely dotted, rounded corners,inner sep=2pt](#2\nodecode){$#1\strut$};
}

\newcommand{\nodevertex}[2]{
	\mathord{%
		\tikz[remember picture,baseline=(#2\nodecode.base)]%
		\node[inner sep=1pt](#2\nodecode){$#1\strut$};
	}
}

\newcommand{\vmodule}[2]{\modulevertex{#2}{M#1}} 

\newcommand{\vgraph}[2]{\graphvertex{#2}{M#1}} 
\newcommand{\vnode}[2]{\nodevertex{#2}{M#1}} 

\def\nodecode{}

\def\gsize#1{|\mkern-2.75mu|#1|\mkern-2.75mu|}
\newcommand{\gless}[2]{\gsize{#1} < \gsize{#2}}

\theoremstyle{plain}
\newtheorem{lemma}[thm]{Lemma}

\def\gcon{\mathsf{G_4}}
\def\gconp{\mathsf{G_4'}}
\def\gconpp{\mathsf{G_4''}}
\def\ngcon{\mathsf{\cneg{G_4}}}
\def\Zcon{\mathsf Z}
\def\Ncon{\mathsf N}

\def\Ccon{\mathsf C}

\def\partition#1{\langle #1\rangle}
\def\block#1{[#1]}
\def\pfour{\mathsf{P_4}}
\def\lwith{\mathbin{\&}}
\def\lplus{\mathbin{\oplus}}
\def\partof#1{\mathbb P_{#1}}
\def\partofn{\partof n}
\def\MLLgen{\MLL^\circ_{\mathcal C}}

\def\Pfour{\mathsf P_4}

\begin{document}

\title[An Analytic Propositional Proof System on Graphs]{An Analytic Propositional Proof System on Graphs}
\titlecomment{{\lsuper*}This is an extended version of a paper published at LICS 2020~\cite{{Acclavio2020}}.}

\author[M.~Acclavio]{Matteo Acclavio}
\address{
Dipartimento di Matematica e Fisica,
Università Roma Tre,
Roma,
Italy
}
\urladdr{
	matteoacclavio.com/Math.html
}

\author[R.~Horne]{Ross Horne}	
\address{
Computer Science,
University of Luxembourg,
Esch-sur-Alzette,
Luxembourg
\\ 
}
\email{\texttt{ross.horne@uni.lu}} 

\author[L.~Stra\ss burger]{Lutz Stra\ss burger}
\address{
 Inria, Equipe Partout,
 Ecole Polytechnique,
 LIX UMR~7161,
 France
}
\urladdr{
	www.lix.polytechnique.fr/Labo/Lutz.Strassburger/
}




\begin{abstract}
  In this paper we present a proof system that operates on graphs
  instead of formulas.  Starting from the well-known relationship
  between formulas and cographs, we drop the cograph-conditions and
  look at arbitrary (undirected) graphs. This means that we lose the
  tree structure of the formulas corresponding to the
  cographs, and we can no longer use standard proof theoretical methods
  that depend on that tree structure. In order to overcome this
  difficulty, we use a modular decomposition of graphs and some
  techniques from deep inference where inference rules do not rely on
  the main connective of a formula. For our proof system we show the
  admissibility of cut and a generalisation of the splitting
  property. Finally, we show that our system is a conservative
  extension of multiplicative linear logic with mix, and we argue that our graphs form a notion of generalised connective.
\end{abstract}

\maketitle


\setcounter{tocdepth}{1}

\begingroup
\vspace*{-0.5cm}
\let\clearpage\relax
\tableofcontents
\endgroup

\section{Introduction}
\label{sec:intro}

The notion of formula is central to all applications of logic and
proof theory in computer science, ranging from the formal verification
of software, where a formula describes a property that the program
should satisfy, to logic programming, where a formula represents a
program~\cite{miller:uniform, Kobayashi1993}, and functional programming, where a
formula represents a type~\cite{howard:80}. Proof theoretical methods
are also employed in concurrency theory, where a formula can represent a
process whose behaviours may be extracted from a proof of the formula~%
\cite{miller:pi,bruscoli:02,FAGES200114,deng_simmons_cervesato_2016,
	OLARTE201746,NIGAM201735,horne:19, Horne2019b,Horne2020}.
This \emph{formulas-as-processes} paradigm
is not as well-investigated as the \emph{formulas-as-properties},
\emph{formulas-as-programs} and \emph{formulas-as-types} paradigms
mentioned before. In our opinion, a reason for this is that
the notion of formula reaches its limitations when it comes to
describing processes as they are studied in concurrency theory.

For example, Guglielmi's $\textsf{BV}$~\cite{gug:SIS} and Retor\'e's \emph{pomset
  logic}~\cite{Retore1997}
are proof systems which extend linear logic with a notion of
sequential composition and can model series-parallel orders.\footnote{In his PhD-thesis~\cite{retore:phd}, Retor\'e considers all partial ordered multisets, but in later versions~\cite{retore:21} only series-parallel orders are considered to maintain the correspondence to formulas.}$^,$\footnote{It has long been believed that $\textsf{BV}$ and pomset
  logic are the same, but recently it has been shown that this is not the case~\cite{tito:lutz:csl22,SIS-III}.} 
However,
series-parallel orders cannot express some ubiquitous patterns of
\emph{causal dependencies} such as the logical time constraints on producer-consumer
queues~\cite{Lodaya2000}, which are within the scope of
pomsets~\cite{Pratt1986}, event structures~\cite{Nielsen1985}, and
Petri nets~\cite{Petri1976}.  The essence of this problem is already
visible when we consider \textit{symmetric dependencies}, such as
separation, which happens to be the dual
concept to concurrency in the \emph{formulas-as-processes}
paradigm.

Let us use some simple examples to explain the problem. Suppose we are
in a situation where two processes $A$ and $B$ can communicate with
each other, written as $A\lpar B$, or are separated from each
other, written as $A\ltens B$, such that no communication is
possible. Now assume we have four atomic processes $a$, $b$, $c$, and
$d$, from which we form the two processes $P=(a\ltens b)\lpar(c\ltens
d)$ and $Q=(a\lpar c)\ltens(b\lpar d)$. Both are perfectly fine
formulas of multiplicative linear logic ($\MLL$)~\cite{girard:87}. 
In
$P$, we have that $a$ is separated from $b$ but can communicate with
$c$ and $d$. Similarly, $d$ can communicate with $a$ and $b$ but is
separated from $c$, and so on. On the other hand, in $Q$, $a$ can only
communicate with $c$ and is separated from the other two, and $d$
can only communicate with $b$, and is separated from the other
two. We can visualise this situation via graphs where $a$, $b$, $c$,
and $d$ are the vertices, and we draw an edge between two vertices if
they are separated, and no edge if they can communicate. Then $P$
and $Q$ correspond to the two graphs shown below.
\begin{equation}
  \label{eq:ex1}
  \begin{array}{c@{\qquad}c}
    P=(a\ltens b)\lpar(c\ltens d)
    &
    Q=(a\lpar c)\ltens(b\lpar d)
    \\[1ex]
    \begin{array}{c@{\quad\;\;}c}
    \vb1 & \vd1 \\
    \\[-1ex]
    \va1 &  \vc1
    \end{array}
    \edges{a1/b1,c1/d1}
    &
    \begin{array}{c@{\quad\;\;}c}
    \vb2 & \vd2 \\
    \\[-1ex]
    \va2 &  \vc2
    \end{array}
    \edges{a2/b2,a2/d2,c2/b2,c2/d2}
  \end{array}
\end{equation}
It should also be possible to describe a situation
where $a$ is separated from $b$, and $b$ is separated from $c$,
and $c$ is separated from $d$, but $a$ can communicate with $c$ and
$d$, and $b$ can communicate with $d$, as indicated by the graph
below.
\begin{equation}
  \label{eq:ex2}
  \begin{array}{c@{\quad\;\;}c}
    \vb3 & \vd3 \\
    \\[-1ex]
    \va3 &  \vc3
  \end{array}
  \edges{a3/b3,b3/c3,c3/d3}
\end{equation}
An example of this behaviour could arise in the setting of concurrent processes, where four processes $a$, $b$, $c$, and $d$ 
 satisfy information flow constraints such that
$c$ and $a$ can communicate;
$a$ and $d$ can communicate; and
$d$ and $b$ can communicate,
and no further communications are possible.
This models an intransitive information flow,
since we have two processes, $a$ and $d$, which can communicate with each other, and respectively with $c$ and $b$;
yet the same processes must also ensure that no information flows between $c$ and $b$.
However, the graph~\eqref{eq:ex2} cannot be described by a formula in the way illustrated for 
the two graphs in~\eqref{eq:ex1}.

This means that the tools of
proof theory, which have been developed over the course of the last
century and which were very successful for the
\emph{formulas-as-properties}, \emph{formulas-as-programs}, and
\emph{formulas-as-types} paradigms, cannot be used for the
\emph{formulas-as-processes} paradigm unless we forbid situations as
in~\eqref{eq:ex2} above. This seems to be a very strong
and unnatural restriction (that is, it is an \textit{a posteri} restriction imposed by the use of formulas, with no \textit{a priory} justification stemming from process modelling problems). The purpose of this paper is to propose
a way to change this unsatisfactory situation.

We will present a proof system, called $\GS$ (for \emph{graphical
 proof system}), whose objects of reason are not formulas but graphs,
giving the example in~\eqref{eq:ex2} the same status as the examples in~\eqref{eq:ex1}.
In a less informal way, one could say that standard
proof systems work on cographs (which are the class of graphs that
correspond to formulas as in~\eqref{eq:ex1}~\cite{duffin:65}), and our proof systems works
on arbitrary graphs.  In order for this to make sense, our proof
system should obey the following basic properties:

\begin{enumerate}
\item\label{intro:prop1}
 \emph{Consistency}: There are graphs that are not provable.
 In particular, if only a finite number of graphs is provable (or not provable) then the proof system would not be interesting.
\item\label{intro:prop2}
 \emph{Transitivity}: The proof
  system should come with an implication that is transitive, i.e., if
  we can prove that $A$ implies $B$ and that $B$ implies $C$, then we
  should also be able to prove that $A$ implies $C$.
\item\label{intro:prop3}
 \emph{Analyticity}: As we no longer have formulas, we cannot ask
  that every formula that occurs in a proof is a subformula of its
  conclusion. However, we can investigate a graph theoretical version of this idea, and we can ask that in a proof search situation, there
  is always only a finite number of ways 
  to apply an
  inference rule.
\item\label{intro:prop4}
 \emph{Minimality}: We want to make as few assumptions as
  possible, 
  so that the theory we develop is as general as possible.
\end{enumerate}

Properties~\ref{intro:prop1}-\ref{intro:prop3} are standard for any proof system, and they are usually
proved using cut elimination. In that respect our paper is no
different. We introduce a notion of cut and show its
admissibility for $\GS$. Then Properties 1-3 are immediate consequences.

Property~\ref{intro:prop4} is of a more subjective nature. In our case, we only make
the following two basic assumptions:

\begin{enumerate}[(i)]
	\item 
	For any graph $A$, we should be able to prove that $A$
	  implies $A$. 
	  This assumption is almost impossible to argue against, so can be expected for any logic.
	  
	\item
	If a graph $A$ is provable, then the graph $G=C\coons{A}{}$ 
		is  also provable, provided that
	  $C\coonso{}$ is a provable context.\footnote{Formally, the notation $G=C\coons{A}{}$ means that
	    $A$ is a module of $G$, and $C\coonso{}$ is the graph obtained
	    from $G$ by removing all vertices belonging to $A$. We give the
	    formal definition in Section~\ref{sec:modules}.} This can be compared to the
	  \emph{necessitation rule} of modal logic, which says that if $A$ is
	  provable then so is $\lbox A$, except that in our case the $\lbox$
	  is replaced by the provable graph context $C\coonso{}$.
\end{enumerate}

All other properties of the system $\GS$ follow from the need to
obtain admissibility of cut. This means that this paper does not
present some random system, but follows the underlying principles of
proof theory.
For a more detailed philosophical presentation of these principles, we refer the reader to Appendix~\ref{sec:phil}.

We also target the desirable property of \emph{conservativity}. 
This means that there should be a well-known logic~$\L$
  (based on formulas) such that the restriction of our proof system to
  those graphs that correspond to formulas proves exactly the
  theorems of~$\L$.
This cannot be an assumption used to design a logical system, since it would create circularity (to specify a logic we need a logic); conservativity is more so a cultural sanity condition to check that we have not invented an esoteric logic.
In our case, conservativity will follow from cut admissibility, and 
the logic $\L$ is multiplicative linear logic with mix ($\MLLm$)~\cite{girard:87,bellin:mix,fleury:retore:94}.

\medskip

Let us now summarise how this paper is organised:
In
Section~\ref{sec:formulas}, we give preliminaries on cographs, which
form the class of graphs that correspond to formulas as
in~\eqref{eq:ex1}. Then, in Section~\ref{sec:modules} we give some
preliminaries on modules and prime graphs, which are needed for our
move away from cographs, so that in Section~\ref{sec:system}, we can
present our proof system, which uses the notation of open
deduction~\cite{gug:gun:par:2010} and follows the principles of deep
inference~\cite{gug:str:01,brunnler:tiu:01,gug:SIS}.  To our
knowledge, this is the first proof system that is not tied to
formulas/cographs but handles arbitrary (undirected) graphs instead.
In Section~\ref{sec:properties} we show some properties of our system,
and 
Sections~\ref{sec:splitting} and~\ref{sec:upfrag}  
are dedicated to cut elimination, which is the basis for showing properties \eqref{intro:prop1}, \eqref{intro:prop2}, and \eqref{intro:prop3}, mentioned above.
We also explain the technology we must develop in order to be able
to prove cut elimination for our proof system. 
The interesting point
is that, not only do we go beyond methods developed for the sequent
calculus, but we also go beyond methods developed for deep inference
on formulas.  In particular, we require entirely new statements of the
tools called \textit{splitting} and \textit{context reduction}, and
furthermore their proofs are inter-dependent, whereas normally context
reduction follows from splitting~\cite{SIS-V,gug:tub:split}.

Then, in
Section~\ref{sec:MLL}, we not only show that our system is a conservative
extension of $\MLLm$, we also show a form of analyticity for our system. 
Finally, in Section~\ref{sec:generalised}, we
show how our work is related to the work on generalised
connectives 
introduced in~\cite{girard:87:b,danos:regnier:89}.
We end this paper with a discussion of related work in
Section~\ref{sec:relatedWork}
and a
conclusion in Section~\ref{sec:conclusion}.

\medskip

Compared to the conference version~\cite{Acclavio2020} of this paper, there are the following three major additions:
\begin{itemize}
	
\item We give detailed proofs of the \emph{Splitting Lemma} and the
  \emph{Context Reduction Lemma} (in Section~\ref{sec:splitting} and Appendix~\ref{sec:splittingproofs}),
  which are crucial for the cut elimination proof. In fact, we also
  completely reorganised the proofs with respect to the technical
  appendix of~\cite{Acclavio2020}\footnote{That appendix is available
    at~\url{https://hal.inria.fr/hal-02560105}.}.  For proving these
  lemmas, we could not rely on the general method that has been
  proposed by Aler Tubella in her PhD~\cite{tubella:phd}.
\item We present a notion of analyticity for proof systems on graphs and
  show that our system is analytic in that respect (in
  Section~\ref{sec:MLL}).
\item We show that general graphs with $n$ vertices can be seen as
  generalised $n$-ary connectives (in Section~\ref{sec:generalised}),
  and we compare this notion with the existing notion of generalised
  (multiplicative)
  connective~\cite{girard:87:b,danos:regnier:89,mai:19,acc:mai:20}.

\item We include in Appendix~\ref{sec:phil} a more detailed discussion of
  how our proof system can be considered satisfactory with respect to 
  the Properties~\eqref{intro:prop1}--\eqref{intro:prop4} mentioned above.

\end{itemize}

\medskip

Finally, let us argue that logics are not designed but discovered. They typically follow logical principles where design parameters are limited.
For example, we will see that we do not get to chose whether or not the following implications hold:
\begin{equation}
    \label{eq:introex}
\nvdash
\begin{array}{c@{\quad\;\;}c}
	\vb1 & \vd1 
	\\\\[-1ex]
	\va1 & \vc1 
\end{array}
\edges{a1/b1,b1/c1,c1/d1}
\multimap
\begin{array}{c@{\quad\;\;}c}
	\vb1 & \vd1 
	\\\\[-1ex]
	\va1 & \vc1 
\end{array}
\edges{a1/b1,c1/d1}
\qquad\qquad\quad
\vdash
\begin{array}{c@{\quad\;\;}c}
	\vb1 & \vd1 
	\\\\[-1ex]
	\va1 & \vc1 
\end{array}
\edges{a1/b1,b1/c1,c1/d1}
\multimap
\begin{array}{c@{\quad\;\;}c}
	\vb1 & \vd1 
	\\\\[-1ex]
	\va1 & \vc1 
\end{array}
\edges{b1/c1,c1/d1}
\end{equation}
There is no pre-existing semantics or proof system we can refer to at this point.
Nonetheless, from the above discussed principles we can argue, that in a logic on graphs, the former implication in~\eqref{eq:introex} cannot hold while the latter must hold. 
Over the course of this paper, we explore the design of proof systems on graphs based on logical principles, which enables us to confidently state such facts.


\section{From Formulas to Graphs}\label{sec:formulas}

In this preliminary section we recall the basic textbook definitions for graphs and formulas, and their correspondence via cographs.

\begin{defi}\label{def:graph}
  A \defn{(simple, undirected) graph} $\gG$ is a pair
  $\tuple{\vG,\eG}$ where $\vG$ is a set of vertices and $\eG$ is a
  set of two-element subsets of $\vG$. We omit the index $\gG$ when it
  is clear from the context. For $v,w\in\vG$ we write $vw$ as an abbreviation
  for $\set{v,w}$. A graph $G$ is \defn{finite} if its vertex set
  $\vG$ is finite. Let $L$ be a set and $\gG$ be a graph. We say that
  $\gG$ is \defn{$L$-labelled} (or just \defn{labelled} if
  $L$ is clear from the context) if every vertex in $\vG$ is associated
  with an element of $L$, called its \defn{label}. We write
  $\labelof[\gG]{v}$ to denote the label of the vertex $v$ in $\gG$. A
  graph $\gG'$ is a \defn{subgraph} of a graph $\gG$, denoted as
  $\gG'\subseteq\gG$ iff $\vGp\subseteq\vG$ and $\eGp\subseteq\eG$. We
  say that $\gG'$ is an \defn{induced subgraph} of $\gG$ if $\gG'$ is
  a subgraph of $\gG$ and for all $v,w\in\vGp$, if $vw\in\eG$ then
  $vw\in\eGp$. For a graph $G$ we write $\sizeof{\vG}$ for its number of vertices and $\sizeof{\eG}$ for its number of edges.
\end{defi}

In the following, we will just say \emph{graph} to mean a finite,
undirected, labelled graph, where the labels come from the set $\cA$
of atoms which is the (disjoint) union of a countable set of
propositional variables $\cV=\set{a,b,c,\ldots}$ and their duals
$\cneg\cV=\set{\cneg a,\cneg b,\cneg c,\ldots}$.

Since we are mainly interested in how vertices are labelled, but not so much in the identity of the underlying vertex,
we heavily rely on the notion of graph isomorphism.

\begin{defi}\label{def:iso}
  Two graphs $\gG$ and $\gG'$ are \defn{isomorphic} if there exists a
  bijection $f \colon \vG \rightarrow \vGp$ such that for all
  $v,u\in\vG$ we have $vu \in \eG$ iff $f(v)f(u) \in \eGp$ and
  $\labelof[\gG]{v}=\labelof[\gG']{f(v)}$. We denote this as
  $G\isom[f]G'$, or simply as $G\isom G'$ if $f$ is clear from the context
  or not relevant.
\end{defi}

In the following, we will, in diagrams, forget the identity of the
underlying vertices, showing only the label, as in the examples in the
introduction.

In the rest of this section we recall the characterisation of those
graphs that correspond to formulas. For simplicity, we restrict
ourselves to only two connectives, and for reasons that will become
clear later, we use the $\lpar$ (\defn{par}) and $\ltens$
(\defn{tensor}) of linear logic~\cite{girard:87}. More precisely,
\defn{formulas} are generated by the grammar
\begin{equation}
  \label{eq:gram}
  \phi,\psi\Coloneqq \lun\mid a\mid\cneg a\mid \phi\lpar\psi\mid \phi\ltens\psi
\end{equation}
where $\lun$ is the \defn{unit}, and $a$ can stand for any
propositional variable in~$\cV$. As usual, we can define the negation
of formulas inductively by letting
$\cnegg a=a$ for all $a\in\cV$, and by
using the De Morgan duality between $\lpar$
and $\ltens$: $\cneg{(\phi\lpar\psi)}=\cneg\phi\ltens\cneg\psi$ and
$\cneg{(\phi\ltens\psi)}=\cneg\phi\lpar\cneg\psi$; the unit is
self-dual: $\cneg \lun=\lun$.

On formulas we
define the following structural equivalence relation:
\begin{equation}
  \label{eq:fequiv}
    \begin{array}{r@{~}l@{\hskip4em}r@{~}l}
	\phi\lpar(\psi\lpar\xi)\fequiv&(\phi\lpar\psi)\lpar\xi 	&
	\phi\ltens(\psi\ltens\xi)\fequiv&(\phi\ltens\psi)\ltens\xi \\
    	\phi\lpar\psi\fequiv&\psi\lpar\phi				&
    	\phi\ltens\psi\fequiv&\psi\ltens\phi 			\\
    	\phi\lpar\lun\fequiv&\phi 					&
    	\phi\ltens\lun\fequiv&\phi 					
    \end{array}
\end{equation}

In order to translate formulas to graphs, we define the following two
operations on graphs:

\begin{defi}\label{def:par}
  Let $\gG=\tuple{\vG,\eG}$ and $\gH=\tuple{\vH,\eH}$ be graphs. We define the \defn{par} of $\gG$ and $\gH$ to be their disjoint union and the \defn{tensor} to be their join, i.e.:
  \begin{eqnarray*}
    \gG\lpar\gH&=&\tuple{\vG\uplus\vH,\eG\uplus\eH}\\
    \gG\ltens\gH&=&\tuple{\vG\uplus\vH,\eG\uplus\eH\uplus\set{vw\mid v\in\vG,w\in\vH}}
  \end{eqnarray*}
\end{defi}

These operations can be visualised as follows:
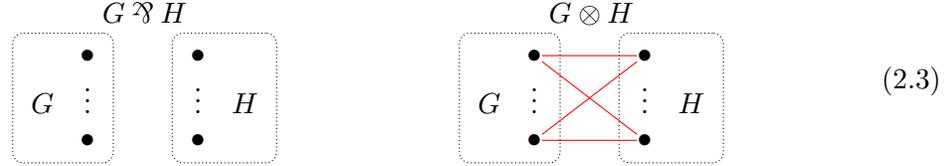
\begin{equation}\label{eq:TensPar}
	\begin{array}{c@{\qquad\qquad\qquad\;}c}
		\gG\lpar\gH &\gG\ltens\gH 
		\\
		\vgraph{1}{\begin{array}{cc}&\vm1 \\ \gG&\vdots \\ &\vm2 \end{array}}
		\qquad
		\vgraph{2}{\begin{array}{cc}\vm3 \\ \vdots&\gH \\ \vm4 \end{array}}
		&
		\vgraph{1}{\begin{array}{cc}&\vm1 \\ \gG&\vdots \\ &\vm2 \end{array}}
		\qquad
		\vgraph{2}{\begin{array}{cc}\vm3 \\ \vdots&\gH \\ \vm4 \end{array}}
		\joinedges{m1,m2}{m4,m3}
	\end{array}
\end{equation}

For a formula $\phi$, we can now define its associated
graph~$\graphof{\phi}$ inductively as follows:
$\graphof\lun=\emptyset$ the empty graph; $\graphof{a}=a$ a 
single-vertex graph whose vertex is labelled by $a$ (by a slight abuse
of notation, we denote that graph also by $a$); similarly
$\graphof{\cneg a}=\cneg a$; finally we define
$\graphof{\phi\lpar\psi}=\graphof{\phi}\lpar\graphof{\psi}$ and
$\graphof{\phi\ltens\psi}=\graphof{\phi}\ltens\graphof{\psi}$.

\begin{thm}\label{thm:cograph:AC} 
	 For any two formulas, $\phi\fequiv\psi$ iff\/ $\graphof{\phi}\isom\graphof{\psi}$.
\end{thm}

\begin{proof}
  By a straightforward induction.
\end{proof}

\begin{defi}
  A graph is \defn{$\pfour$-free} 
  iff it does not have an induced subgraph of the shape
\begin{equation}
  \label{eq:ex3}
  \vm1 \qquad \vm2 \qquad \vm3 \qquad \vm4
  \edges{m1/m2,m2/m3,m3/m4}
\end{equation}
\end{defi}

The following result is classical and its proof can be found, e.g., in~\cite{moh:89} or \cite{gug:SIS}.
\begin{thm}[\cite{duffin:65}]\label{thm:cograph}
  Let $\gG$ be a graph. Then there is a formula $\phi$ with $\graphof\phi\isom G$ iff $\gG$ is $\pfour$-free.
\end{thm}

The graphs characterised by Theorem~\ref{thm:cograph}
are called \defn{cographs}, because they are the smallest class of
graphs containing all single-vertex graphs and being closed under
complement and disjoint union.

Because of Theorem~\ref{thm:cograph}, one can think of standard proof
systems as \emph{cograph proof systems}. 
Since in this paper we want to move from cographs to general graphs, 
we need to investigate how much of the tree structure of formulas 
(which makes cographs so interesting for proof theory~\cite{retore:03,hughes:pws,str:FSCD17})
can be recovered for general graphs.


\section{Modules and Prime Graphs}\label{sec:modules}

One of the consequences of the tree structure of formulas is the notion of \emph{subformula}, which is induced by the notion of subtree. This is lost if we move from cographs to general graphs. However, in a cograph $\gG$, a \emph{module} of $\gG$ corresponds to a subformula of the formula $\phi$ given by Theorem~\ref{thm:cograph}. The notion of module also exists for general graphs, and in our proof systems on graphs, modules will play a similar role as subformulas play in ordinary proof systems on formulas.   

For this reason, we recall here some standard results on graph modules and the modular decomposition of graphs~\cite{moh:rad,ehr:har:roz:theory,cou:del:modular,hab:paul:survey}.

\begin{defi}\label{def:module}
  Let $\gG$ be a graph. A \defn{module} of $\gG$ is an induced
  subgraph $\gM=\tuple{\vM,\eM}$ of $\gG$ such that for all
  $v\in\vG\setminus\vM$ and all $x,y\in \vM$ we have $vx\in\eG$ iff
  $vy\in\eG$.
\end{defi}

\begin{nota}\label{not:mudule}
	Let $\gM$ be a module of a graph $\gG$.
	Since each vertex in $\vM$ has the same relation with all vertices outside $\gM$, that is, with all vertices in $\vG\setminus\vM$, we introduce the following notation when drawing graphs
	\begin{equation}\label{eq:module-dia}
	{\small
		{
			{\begin{array}{ccc}
				\vmodule{2}{
					\gM
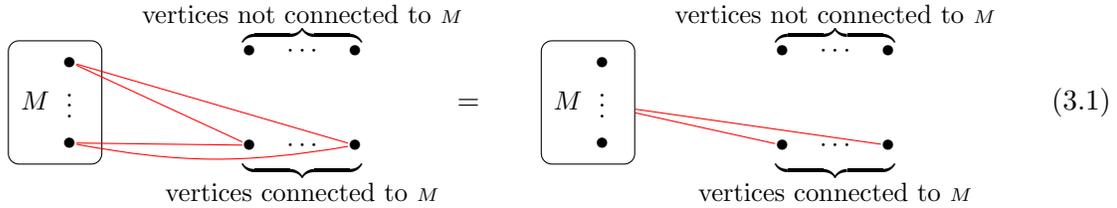
					\begin{array}{c}
						\vm1
						\\
						\vdots
						\\
						\vm2
					\end{array}	
				}
				&
				\begin{array}{c}
					\overbrace{\bullet \quad\cdots\quad \bullet}^{\mbox{vertices not connected to }\gM}
					\\
					\\
					\\
					\underbrace{\vm5 \quad\cdots\quad \vm6}_{\mbox{vertices connected to }\gM}
				\end{array}
			\end{array}}
		}
		\edges{m1/m5,m1/m6,m2/m5}
		\tikz[graphstyle]\draw[-] (m2\nodecode) to [bend right=10pt] (m6\nodecode);
		\hskip-1.5em=\hskip1.5em
		{
			\begin{array}{ccc}
				\vmodule{2}{
					\gM
					\begin{array}{c}
						\vm1
						\\
						\vdots
						\\
						\vm2
					\end{array}	
				}
				&
				\begin{array}{c}
					\overbrace{\bullet \quad\cdots\quad \bullet}^{\mbox{vertices not connected to }\gM}
					\\
					\\
					\\
					\underbrace{\vm5 \quad\cdots\quad \vm6}_{\mbox{vertices connected to }\gM}
				\end{array}
			\end{array}
		}
		\edges{M2/m5,M2/m6}
	}
	\end{equation}
	which allows us to reduce the number of drawn edges and to increase readability. 
	Using this notation the definition of the tensor operator in~\eqref{eq:TensPar} can be written as folllows: 
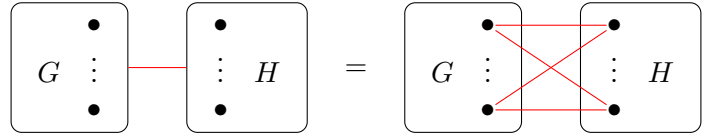
	$$
		\vmodule{1}{\begin{array}{cc}&\vm1 \\ \gG&\vdots \\ &\vm2 \end{array}}
		\qquad
		\vmodule{2}{\begin{array}{cc}\vm3 \\ \vdots&\gH \\ \vm4 \end{array}}
	\edges{M1/M2}
	\quad = \quad
		\vmodule{1}{\begin{array}{cc}&\vm1 \\ \gG&\vdots \\ &\vm2 \end{array}}
		\qquad
		\vmodule{2}{\begin{array}{cc}\vm3 \\ \vdots&\gH \\ \vm4 \end{array}}
	\joinedges{m1,m2}{m4,m3}
	$$
\end{nota}

\begin{nota}\label{not:context}
	Let $\gM$ be a module of a graph $\gG$,
	and
	let $\gC$ be the graph obtained from $\gG$ 
	by removing all vertices in $\gM$ (including incident edges).
	If 
	$R\subseteq\vC$ is the set of vertices that are in $\gG$ connected to a
	vertex in $\vM$ (and hence to all vertices in $\gM$ by modularity),
	then we write $\gG=\gC\coons{\gM}{R}$
	and we say that $\gC\coonso R$ is the \defn{context} of $\gM$ in $\gG$.  
	
	Alternatively, we can consider a context
	$\gC\coonso{R}$ as a graph with a distinguished vertex $\coonso{}$
	such that 
	$R$ is the set of neighbours of $\coonso{}$.
	In this case $\gC\coons{\gM}R$ can be defined as 
	the graph obtained from $\gC\coonso{}R$ by substituting the vertex $\coonso{}$
	with the graph $\gM$ and adding 
	all the edges from each vertex in $\vM$ to each vertex in $R$.
\end{nota}

\begin{exa}
  Consider the graph context on the left below with the hole denoted by an empty $\vmodule1{~}$. 
  If we substitute the two-vertex graph $a \lpar \cneg{a}$ for the hole, then we obtain the graph in the middle below:
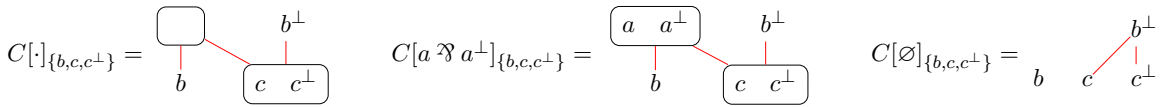
\begin{equation*}
	\scalebox{.85}{$
	C\coonso{\set{b,c,\cneg c}}
        =\!
	\begin{array}{c@{\quad\;\;}c@{\quad\;\;}c}
		\vmodule1{\quad} &  \vnb1 \\
		\\[-1ex]
		\vb1 & \vmodule{2}{c \quad \cneg c}
	\end{array}
	\edges{M1/b1,M1/M2,nb1/M2}
	\quad
	C\coons{a \lpar \cneg{a}}{\set{b,c,\cneg c}}
	= 
	\begin{array}{cc@{\quad\;\;}c}
		\vmodule{1}{a \quad \cneg a}  & \vnb1 \\
		\\[-1ex]
		\vb1 & \vmodule{2}{c \quad \cneg c}
	\end{array}
	\edges{M1/b1,M1/M2,nb1/M2}
	\quad
	C\coons{\emptyset}{\set{b,c,\cneg c}}
	=
	\!\begin{array}{c@{\quad\;\;}c@{\quad\;\;}c}
		&  & \vnb1 \\
		\\[-1ex]
		\vb1 & \vc1 &  \vnc1
	\end{array}
	\edges{nb1/nc1,nb1/c1}
        $}
\end{equation*}
Finally, if we replace the hole with the empty graph we obtain the graph on the right above.
\end{exa}

\begin{lemma}\label{lem:module}
  Let $\gG$ be a graph and $\gM,\gN$ be modules of $\gG$. Then
  \begin{enumerate}
  \item $\gM\cap\gN$ is a module of $\gG$;
  \item if $\gM\cap\gN\neq\unit$, then $\gM\cup\gN$ is a module of $\gG$; and
  \item if $\gN\not\subseteq\gM$ then $\gM\setminus\gN$ is a module of $\gG$.
  \end{enumerate}
\end{lemma}

\begin{proof}
  The first statement follows immediately from the definition. For the
  second one, let $\gL=\gM\cap\gN\neq\unit$, and let
  $v\in\gG\setminus(\vM\cup\vN)$ and $x,y\in \vM\cup\vN$. If $x,y$ are
  both in $\gM$ or both in $\gN$, then we have immediately $vx\in\eG$
  iff $vy\in\eG$. So, let $x\in\vM$ and $y\in\vN$, and let
  $z\in\gL$. We have $vx\in\eG$ iff $vz\in\eG$ iff
  $vy\in\eG$. Finally, for the last statement, let
  $x,y\in\vM\setminus\vN$ and let
  $v\in\vG\setminus(\vM\setminus\vN)$. If $v\notin\vM$, we immediately
  have $vx\in\eG$ iff $vy\in\eG$. So, let $v\in\vM$, and therefore
  $v\in\vM\cap\vN$. Let $z\in\vN\setminus\vM$. Then $vx\in\eG$ iff
  $zx\in\eG$ iff $zy\in\eG$ iff $vy\in\eG$.
\end{proof}

\begin{defi}
  Let $\gG$ be a graph. A module $\gM$ in $\gG$ is \defn{maximal}\/ if
  $\gM\neq\gG$ and
  for all modules $\gM'\neq \gG$ of $\gG$ 
  we have that
  $\gM\subseteq\gM'$ implies $\gM=\gM'$.
\end{defi}

\begin{defi}
  A module $\gM$ of a graph $\gG$ is \defn{trivial}\/ iff either
  $\vM=\emptyset$ or $\vM$ is a singleton or $\vM=\vG$. A graph $\gG$
  is \defn{prime} iff $\sizeof\vG\ge2$ and all modules of $\gG$ are trivial.
\end{defi}

\begin{defi}\label{def:graphCompisitionVia}
  Let $\gG$ be a graph with $n$ vertices $\vG=\set{v_1,\ldots,v_n}$
  and let $H_1,\ldots,H_n$ be $n$ graphs. We define the
  \defn{composition of $H_1,\ldots,H_n$ via $\gG$}, denoted as
  $G\connn{H_1,\ldots,H_n}$, by replacing each vertex $v_i$ of $\gG$
  by the graph $\gH_i$; and there is an edge between two vertices $x$
  and $y$ if either $x$ and $y$ are in the same $\gH_i$ and
  $xy\in\eHi$ or $x\in\vHi$ and $y\in\vHj$ for $i\neq j$ and
  $v_iv_j\in\eG$. Formally,
  $G\connn{H_1,\ldots,H_n}=\tuple{V^\ast,E^\ast}$ with
  \begin{equation*}
		V^\ast=\biguplus_{1\le i\le n}\vHi
		\quad\quand\quad
		E^\ast=\left(\biguplus_{1\le i\le n}\eHi\right)\uplus\left\{xy\mid x\in \vHi,y\in\vHj,v_iv_j\in\eG\right\}
\end{equation*}
\end{defi}

This concept allows us to decompose graphs into prime graphs (via
Lemma~\ref{lem:newhope} below) and recover a tree structure for an
arbitrary graph.\footnote{This also allows us to consider prime graphs as generalised {non-decomposable}
  $n$-ary connectives. We will come back to this point in Section~\ref{sec:generalised}.}
  The two operations $\lpar$ and $\ltens$, defined
in Definition~\ref{def:par} are then represented by the following two prime
graphs:
\begin{equation}
  \label{eq:par-tens}
  \lpar\colon\quad \vbul1\qquad\vbul2
  \qquand
  \ltens\colon\quad \vbul3\qquad\vbul4
  \edges{n3/n4}
  \hskip3em
\end{equation}
We have $\lpar\connn{G,H}=G\lpar H$ and $\ltens\connn{G,H}=G\ltens
H$.

\begin{lemma}[\cite{gallai:67}]
\label{lem:newhope}
  For every non-empty graph $\gG$ we have exactly one of the following four cases:
  \begin{enumerate}[(i)]
  \item $\gG$ is a singleton graph.
  \item $\gG=\gA\lpar\gB$ for some $\gA$, $\gB$ with $\gA\neq\unit\neq\gB$.
  \item $\gG=\gA\ltens\gB$ for some $\gA$, $\gB$ with $\gA\neq\unit\neq\gB$.
  \item $\gG=\gP\connn{\gA_1,\ldots,\gA_n}$ for some prime graph $\gP$
    with $n=\sizeof{\vP}\ge 4$ and $\gA_i\neq\unit$ for every $i\in\set{1,\ldots,n}$.
  \end{enumerate}
\end{lemma}

\begin{proof}
  Let $\gG$ be given. If $\sizeof\gG=1$, we are in case (i). Now
  assume $\sizeof\gG>1$, and let $\gM_1,\ldots,\gM_n$ be the maximal
  modules of $\gG$. Note that $n\geq 2$. Now we have two cases:
  \begin{itemize}
  \item For all $i,j\in\set{1,\ldots,n}$ with $i\neq j$ we have
    $\gM_i\cap \gM_j=\unit$.  Since every vertex of $\gG$ forms a
    module, every vertex must be part of a maximal module. Hence
    $\vG=\vertices[\gM_1]\cup\cdots\cup\vertices[\gM_n]$. Therefore
    there is a graph $\gP$ such that
    $\gG=\gP\connn{M_1,\ldots,M_n}$. Since all $\gM_i$ are maximal in
    $\gG$, we can conclude that $\gP$ is prime. If $\sizeof{\vP}\ge 4$
    we are in case (iv). If $\sizeof{\vP}<4$ we are either in case
    (ii) or (iii), as the two graphs in~\eqref{eq:par-tens} are only
    two prime graphs with $\sizeof{\vP}=2$, and there are no prime
    graphs with $\sizeof{\vP}=3$.
  \item We have some $i\neq j$ with $\gM_i\cap \gM_j\neq\unit$. Let
    $\gL=\gM_i\cap \gM_j$ and $\gN=\gM_i\setminus\gM_j$ and
    $\gK=\gM_j\setminus\gM_i$. By Lemma~\ref{lem:module}, $L$, $N$,
    $K$, and $\gM_i\cup\gM_j$ are all modules of $\gG$. Since $\gM_i$
    and $\gM_j$ are maximal, it follows that $\gG=\gM_i\cup\gM_j$, and
    therefore $\gG=N\ltens L\ltens K$ or $G=N\lpar L\lpar K$. \qedhere
  \end{itemize}
\end{proof}

\begin{exa}\label{ex:modDec}
	Consider the following graph
	
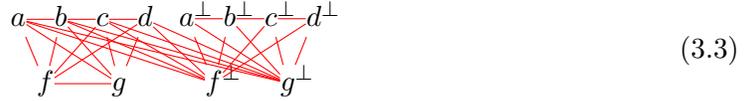
\begin{equation}
          \label{eq:spaghetti}
		\begin{array}{ccccc}
			{\va1 \quad \vb1 \quad \vc1 \quad \vd1 }
			&
			{\vna1 \quad \vnb1 \quad \vnc1 \quad \vnd1}
			\\[1em]
			\vf1
			\qquad
			\vg1
			&
			\vnf1
			\qquad
			\vng1
		\end{array}
		\joinedges{a1,b1,c1,d1}{f1,nf1,g1,ng1}
		\joinedges{na1,nb1,nc1,nd1}{nf1,ng1}
		\edges{f1/g1}
		\edges{a1/b1,b1/c1,c1/d1}
		\edges{na1/nb1,nb1/nc1,nc1/nd1}
	\end{equation}
	Its representation using the modular notation introduced in~(\ref{eq:module-dia}) is shown on the left below and its modular decomposition tree
	is shown on the right below:
	
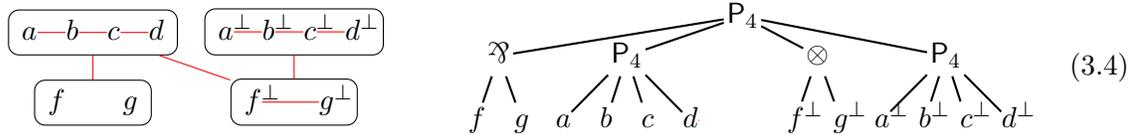
\begin{equation}
          \label{eq:modular}
		\begin{array}{@{\hskip-.5em}cccc}
			\vmodule2{\va1 \quad \vb1 \quad \vc1 \quad \vd1}&\vmodule4{\vna1 \quad \vnb1 \quad \vnc1 \quad \vnd1 ~}
			\\[1em]
			\vmodule1{\vf1\qquad\vg1}
			&
			\vmodule3{\vnf1\qquad\vng1 ~ }
		\end{array}
		\edges{M1/M2,M2/M3,M3/M4}
		\edges{a1/b1,b1/c1,c1/d1}
		\edges{na1/nb1,nb1/nc1,nc1/nd1}
		\edges{nf1/ng1}
		\hskip2em
		\begin{array}{ccccc}
			&&\vnode0{\mathsf P_4}
			\\
			\vnode1{\lpar}& \vnode2{\mathsf P_4}&& \vnode3{\ltens}& \vnode4{\mathsf P_4}
			\\[1em]
			\vf1 \quad \vg1& \va1\quad \vb1\quad\vc1\quad \vd1 && \vnf1 \quad \vng1 & \vna1\quad \vnb1\quad \vnc1\quad\vnd1
		\end{array}
		\blackedges{M0/M1,M0/M2,M0/M3,M0/M4}
		\blackedges{M1/f1,M1/g1,M2/a1,M2/b1,M2/c1,M2/d1,M3/nf1,M3/ng1,M4/na1,M4/nb1,M4/nc1,M4/nd1}
	\end{equation}
\def\Pfour{\mathsf P_4}
The modular decomposition tree can also be written in ``formula style'' as
\begin{equation}
  \label{eq:modform}
  \Pfour\connn{f\lpar g,\Pfour\connn{a,b,c,d},\cneg f\ltens\cneg g,\Pfour\connn{\cneg a,\cneg b,\cneg c,\cneg d}}
\end{equation}
In the rest of the paper we will use all four representations
in~\eqref{eq:spaghetti} and~\eqref{eq:modular} and~\eqref{eq:modform}
interchangeably.
\end{exa}

\section{The Proof System}\label{sec:system}

To define a proof system on graphs, we need a notion of implication on graphs. To do
so, we first introduce a notion of negation on graphs.

\begin{defi}\label{def:dual}
  For a graph $\gG=\tuple{\vG,\eG}$, we define its \defn{dual}
  $\cneg\gG=\tuple{\vG,\enG}$ to have the same set of vertices, and an
  edge $vw\in\enG$ iff $vw\notin\eG$ (and $v\neq w$). The label of a vertex $v$ in
  $\cneg\gG$ is the dual of the label of that vertex in $\gG$, i.e.,
  $\lnG v = \cneg{\lG v}$. For any two graphs $\gG$ and $\gH$, the
  \defn{implication} $\gG\limp\gH$ is defined to be the graph
  $\cneg\gG\lpar\gH$.
\end{defi}

\begin{exa}
  To give an example, consider the graph $\gG$ on the left below
  \begin{equation}
    \label{exa:neg}
    \gG\colon 
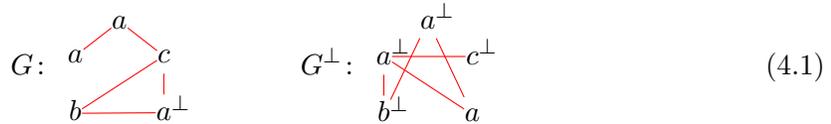
    \begin{array}{c@{\quad}c@{\quad}c}
      &\va1 \\
      \va2 &&\vc1\\
      \\[-1ex]
      \vb1 &&\vna1
    \end{array}
    \edges{a1/a2,a1/c1,c1/na1,d1/d1, b1/c1,b1/na1}
    \hskip4em
    \cneg \gG \colon
    \begin{array}{c@{\quad}c@{\quad}c}
      &\vna1 \\
      \vna2 &&\vnc1\\
      \\[-1ex]
      \vnb1 &&\va1
    \end{array}
    \edges{na1/nb1,na2/nb1,na1/a1, na2/nc1,na2/a1}
  \end{equation}
  Its negation $\cneg\gG$ is shown on the right above.
\end{exa}

\begin{obs}\label{obs:deMorgan}
The dual graph construction defines the standard De
Morgan dualities relating conjunction and disjunction, i.e., for every
formula $\phi$, we have
$\graphof{\cneg\phi}\isom\cneg{\graphof\phi}$. De Morgan
dualities extended to prime graphs as
$\cneg{\gP\connn{M_1,\ldots,M_n}} \isom
\cneg{\gP}\connn{\cneg{M_1},\ldots,\cneg{M_n}}$, where $\cneg{\gP}$ is
the dual graph of $\gP$. Furthermore, $\cneg{\gP}$ is prime if and
only if $\gP$ is prime. Thus each pair of prime graphs $\gP$ and
$\cneg{\gP}$ defines a pair of connectives that are De Morgan duals to
each other.
\end{obs}

We will now develop our proof system based on the above notion of
negation as graph duality. The requirements mentioned in the
introduction entail that:
\begin{enumerate}[(i)]
\item\label{l:id} for any isomorphic graphs $\gG$ and $\gH$, the graph
$\gG\limp\gH$ should be provable;
\item if $\gG\neq\emptyset$ then
$\gG$ and $\cneg\gG$ should not be both provable;
\item the implication $\limp$ should be transitive, i.e., if $\gG\limp\gH$, and
  $\gH\limp\gK$ are provable then so should be $\gG\limp\gK$;
\item the implication $\limp$ should be closed under context,
i.e., if $\gG\limp\gH$ is provable and
$\gC\coonso R$ is an arbitrary context, then
$\gC\coons{\gG}R\limp\gC\coons{\gH}R$ should be provable;
\item\label{l:ssu} if $\gA$ and $\gC$ are provable graphs, and $R\subseteq\vC$,
  then the graph $\gC\coons\gA R$ should also be provable.
\end{enumerate}

As our chosen notation suggests, our proof system will be an extension of multiplicative linear logic (MLL). Note that this is not our starting point, but a consequence of the conditions above. Since we have a common unit for $\lpar$ and $\ltens$, namely the empty graph, a form of \emph{mix} (i.e., $A\ltens B\limp A\lpar B$) will be derivable in our system. Indeed, it turns out that our proof system is a conservative extension of MLL with mix~\cite{girard:87,bellin:mix,fleury:retore:94}. This will be discussed in detail in Section~\ref{sec:MLL}.

\begin{exa}
As an example, consider the following three graphs:
\begin{equation}\label{eq:exa1}
  \gA_1\colon
  \begin{array}{c@{\quad\;\;}c@{\quad\;\;}c@{\quad\;\;}c}
    \vna1 & \va1  \\
    \\[-1ex]
    \vb1 & \vnb1 
  \end{array}
  \edges{na1/b1,a1/nb1,na1/nb1,a1/b1}
  \quad
  \gA_2\colon
  \begin{array}{c@{\quad\;\;}c@{\quad\;\;}c@{\quad\;\;}c}
    \vna2 & \va2  \\
    \\[-1ex]
    \vb2 & \vnb2 
  \end{array}
  \edges{na2/b2,a2/nb2,na2/nb2}
  \quad
  \gA_3\colon
  \begin{array}{c@{\quad\;\;}c@{\quad\;\;}c@{\quad\;\;}c}
    \vna3 & \va3  \\
    \\[-1ex]
    \vb3 & \vnb3 
  \end{array}
  \edges{na3/b3,a3/nb3}
\end{equation}
The graph $\gA_1$ on the left should clearly be provable, as it
corresponds to the formula $(\cneg a\lpar a)\ltens(b\lpar\cneg b)$,
which is provable in $\MLL$. 
The graph $\gA_3$ on the right should not be
provable, as it corresponds to the formula $(\cneg a\ltens
b)\lpar(a\ltens\cneg b)$, which is not provable in $\MLL$. But what about
the graph $\gA_2$ in the middle? It does not correspond to a formula,
and therefore we cannot resort to $\MLL$. Nonetheless, we can make the
following observations. If $\gA_2$ were provable, then so would be the
graph $\gA_4$ shown below:
\begin{equation}
  \label{eq:exa4}
  \gA_4\colon
  \begin{array}{c@{\quad\;\;}c@{\quad\;\;}c@{\quad\;\;}c}
    \vna4 & \va4  \\
    \\[-1ex]
    \va5 & \vna5 
  \end{array}
  \edges{na4/a5,a4/na5,na4/na5}
\end{equation}
as it is obtained from $\gA_2$ by a simple substitution, 
i.e., replacing $b$ with $a$. 
However,
$\cneg{\gA_4}\isom\gA_4$, and therefore $\cneg{\gA_4}$ and $\gA_4$ would
both be provable, which would be a contradiction and should be ruled
out. Hence, $\gA_2$ should not be provable.
It also follows that $\gA_1 \multimap \gA_2$ cannot hold, as
otherwise we would be able to use $\gA_1$ and modus ponens to
establish that $\gA_2$ is provable, which cannot hold as we just
observed. By applying a similar argument, we conclude that $\gA_2 \multimap \gA_3$ 
does not
hold either.
\end{exa}

This example also shows that that implication is not simply subset inclusion of
edges. However, in our minimal logic on graphs that we present here, the converse does hold: we will see later that whenever we
  have that $\gG\limp\gH$ is provable and $\vG=\vH$ then $\eH\subseteq\eG$.\footnote{But this
  observation is not true in general for all logics that might be designed on graphs. For example
  in the extension of Boolean logic, defined in~\cite{CDW:ext-bool},
  is is not necessarily true that implication preserves edges.}

For presenting the inference system we use a \emph{deep inference}
formalism~\cite{gug:str:01,gug:SIS}, which allows rewriting inside an
arbitrary context and admits a rather flexible composition of
derivations. In our presentation we will follow the notation of
\emph{open deduction}, introduced in~\cite{gug:gun:par:2010}.

Let us start with the following two inference rules
\begin{equation}\label{eq:id-ssu} 
  \vlinf{\idr}{}{\cneg\gA\lpar\gA}{\unit}
  \qquad\qquad
  \vlinf{\ssur}{~~S\subseteq\vB,~S\neq\vB}{\gB\coons{\gA}S}{\gB\ltens\gA}
\end{equation}
which are induced by the two Points~\ref{l:id} and~\ref{l:ssu}
above, and which are called \defn{identity down} and \defn{super switch up}, respectively.
The $\idr$ says that for arbitrary graphs $\gC$ and $\gA$ and any
$R\subseteq\vC$, if $\gC$ is provable, then so is the graph
$\gC\coons{\gA\lpar\cneg\gA}R$. Similarly, the rule $\ssur$ says that
whenever $\gC\coons{\gB\ltens\gA}R$ is provable, then so is
$\gC\coons{\gB\coons{\gA}S}R$ for any three graphs $\gA$, $\gB$, $\gC$
and any $R\subseteq\vC$ and $S\subseteq\vB$. The condition $S\neq\vB$
is there to avoid a trivial rule instance, as
$\gB\coons{\gA}S=\gB\ltens\gA$ if $S=\vB$.  Practically, the rule
$\ssur$ only removes some of the edges between the modules $A$ and $B$ in
the graph $\gC\coons{\gB\ltens\gA}R$. But it is important to observe
that we cannot simply remove arbitrary edges. We always have to ensure
that $A$ remains a module in the resulting graph. More precisely, in $\gB\ltens\gA$, both $\gB$ and $\gA$ are modules, but in $\gB\coons{\gA}S$ only $\gA$ is a module, but $\gB$ is not.

As the name ``deep inference'' suggests, we want to apply the inference rules inside any context $\gC\coonso R$, and not only at the top level in the modular decomposition tree. For this, we define now how we build derivations from inference rules.

\begin{defi}\label{def:derivation}
  An \defn{inference system} $\sysS$ is a set of inference rules. We
  define the set of \defn{derivations} in $\sysS$ inductively below,
  and we denote a derivation $\dD$ in $\sysS$ with premise $\gG$ and
  conclusion $\gH$, as follows:
  \begin{equation*}
    \od{\odd{\odh{\gG}}{\dD}{\gH}{\sysS}}
  \end{equation*}
  \begin{enumerate}
  	
  \item Every graph $\gG$ is a derivation (also denoted by $\gG$) with
    premise $\gG$ and conclusion $\gG$.

  \item   	
    If $\gG$ is a graph with $n$ vertices and
    $\dD_1,\dots,\dD_n$ are derivations with premise $\gG_i$ and conclusion
    $\gH_i$ for each $i\in\intset1n$,
    then $\gG\connn{\dD_1,\dots, \dD_n}$ is a derivation with
    premise $\gG\connn{\gG_1,\dots, \gG_n}$
    and conclusion $\gG\connn{\gH_1,\dots, \gH_n}$
    denoted as
    \begin{equation*}
    	\strut
    	\od{\odh{
    			\gG\Connn{
    				\od{\odd{\odh{\gG_1}}{\dD_1}{\gH_1}{\sysS}}
	    			,\dots,
	    			\od{\odd{\odh{\gG_n}}{\dD_n}{\gH_n}{\sysS}}}
    			}
    		}
    \end{equation*}
    respectively. If 
    $\gG=\lpar$ or $\gG=\ltens$
    we may respectively write 
    \begin{equation*}
    	\strut\qquad
    	\od{\odh{\od{\odd{\odh{\gG_1}}{\dD_1}{\gH_1}{\sysS}}
    			\lpar
    			\od{\odd{\odh{\gG_2}}{\dD_2}{\gH_2}{\sysS}}}}
    	\qquor
    	\od{\odh{\od{\odd{\odh{\gG_1}}{\dD_1}{\gH_1}{\sysS}}
    			\ltens
    			{\od{\odd{\odh{\gG_2}}{\dD_2}{\gH_2}{\sysS}}}}}
    \end{equation*}

  \item If $\dD_1$ is a derivation with premise $\gG_1$ and conclusion
    $\gH_1$, and $\dD_2$ is a derivation with premise $\gG_2$ and
    conclusion $\gH_2$, and
    $$
    \vlinf{\rr}{}{\gG_2}{\gH_1}
    $$ is an instance of an inference rule $\rr$, then
    $\dD_2\rcomp{\rr}\dD_1$ is a derivation with premise 
    {$\gG_1$} and
    conclusion $\gH_2$, denoted as
    \begin{equation*}
      \odn{
        \ods{\gG_1}{\dD_1}{\gH_1}{\sysS}
      }{
        \rr
      }{
        \ods{\gG_2}{\dD_2}{\gH_2}{\sysS}
      }{}
      \qquor
      \od{
        \odd{\odd{\odh{G_1}}{\dD_1}{\odn{H_1}{\rr}{G_2}{}
          }{\sysS}}{
              \dD_2}{H_2}{\sysS}}
      \qquor
      \od{\odd{\odi{\odd{\odh{G_1}}{
              \dD_1}{H_1}{\sysS}}{
            \rr}{G_2}{}}{
          \dD_2}{H_2}{\sysS}}
    \end{equation*}
    If $\gH_1\isom[f]\gG_2$ we can compose $\dD_1$ and $\dD_2$ directly to
    $\dD_2\fcomp\dD_1$, denoted as
    \begin{equation}\label{eq:iso}
      \quad\;
      \od{\odo
        {\odh{\ods{\gG_1}{\dD_1}{\gH_1}{\sysS}}}
        {f}
        {\ods{\gG_2}{\dD_2}{\gH_2}{\sysS}}
        {}}
      \;\quor\;
      \od{\odo
        {\odh{\ods{\gG_1}{\dD_1}{\gH_1}{\sysS}}}
        {\isom}
        {\ods{\gG_2}{\dD_2}{\gH_2}{\sysS}}
        {}}
     \;\quor\;
      \od{\odo
        {\odh{\ods{\gG_1}{\dD_1}{\gH_1}{\sysS}}}
        {}
        {\ods{\gG_2}{\dD_2}{\gH_2}{\sysS}}
        {}}
     \;\quor\;
    \od{\odd{\odo{\odd{\odh{G_1}}{
    				\dD_1}{H_1}{\sysS}}{
    			\rr}{G_2}{}}{
    		\dD_2}{H_2}{\sysS}}       
     \;\quor\;    	
    	      \od{
    		\odd{\odd{\odh{G_1}}{\dD_1}{\od{\odo{\odh{H_1}}{\rr}{G_2}{}}
    			}{\sysS}}{
    			\dD_2}{H_2}{\sysS}}    
    \end{equation}
    If $f$ is clear from context we will simply write
    $\dD_2\fcomp\dD_1$ as
    \begin{equation}\label{eq:iso-}
      \hskip5em
      \od{\odd{\odd{\odh{\gG_1}}
          {\dD_1}{\gH_1}{\sysS}}
        {\dD_2}{\gH_2}{\sysS}}
      \qquor
      \od{\odd{\odd{\odh{\gG_1}}
          {\dD_1}{\gG_2}{\sysS}}
        {\dD_2}{\gH_2}{\sysS}}
      \quad
    \end{equation}
    where we only write $\gH_1$ or $\gG_2$ in the middle, omitting the isomorphism step,
    in order to ease readability. However, even if $f$ is not written, we
      always assume it is part of the derivation and explicitely
      given.
  \end{enumerate}
  A \defn{proof} in $\sysS$ is a derivation in $\sysS$ whose premise
  is~$\unit$. A graph $\gG$ is \defn{provable} in $\sysS$ iff there is
  a proof in $\sysS$ with conclusion $\gG$. We denote this as
  $\proves[\sysS]{\gG}$ (or simply as $\proves{\gG}$ if $\sysS$ is
  clear from the context). The \defn{length} of a derivation $\dD$, denoted
  by $\sizeof\dD$, is the number of inference rule instances in $\dD$.
\end{defi}

\begin{rem}\label{rem:context}
  If we have a derivation $\dD$ from $A$ to $B$, and a context
  $\gG\coons{\cdot}R$, then we also have a derivation from
  $\gG\coons{\gA}R$ to $\gG\coons{\gB}R$. We can write this derivation
  as
  \begin{equation*}
    \ods{\gG\coons{\gA}R}{\gG\coons{\dD}R}{\gG\coons{\gB}R}{}
    \quor
    \od{\odh{
        \gG\Coons{\ods{\gA}{\dD}{\gB}{}}R}}
  \end{equation*}
  If $\dD$ consists of a single inference rule, we sometimes write
  (with a slight abuse of notation)
  \begin{equation}\label{eq:rule-context}
    \od{\odh{
        \gG\Coons{\od{\odi{\odh{\gA}}{\rr}{\gB}{}}}R}}
    \qquad\mbox{as}\qquad
    \od{\odi{\odh{\gG\Coons{\gA}R}}{\rr}{\gG\Coons{\gB}R}{}}
  \end{equation}
\end{rem}

\begin{rem}
  In order to ease readability of derivations, we will mostly omit
  the isomorphism steps shown in~\eqref{eq:iso}, and use the
  abbreviation~\eqref{eq:iso-}. This concerns also situations
  as shown below left, which are written as shown below right:
  \begin{equation}
    \label{eq:iso-rule}
    \od{\odi{\odo{\odi{\odh{A}
          }{\rr_1}{B}{}
        }{\isom}{B'}{}
      }{\rr_2}{C}{}}
    \qquad\leadsto\qquad
    \od{\odi{\odi{\odh{A}
        }{\rr_1}{B}{}
      }{\rr_2}{C}{}}
    \quor
    \od{\odi{\odi{\odh{A}
        }{\rr_1}{B'}{}
      }{\rr_2}{C}{}}
  \end{equation}
  In other words, in derivations, we consider graphs equal modulo
  isomorphisms. In particular, for the length $\sizeof\dD$ of a
  derivation, we count only the inference steps that are not
  isomorphisms. This abuse of notation is justified as these
  isomorphisms can be permuted with all inference rules that we
  discuss in this paper. Furthermore, when we draw the graphs as in
  \eqref{eq:exup} below, the isomorphisms are not visible.
  \end{rem}

\begin{exa}\label{ex:ssw}
  The following derivation is an example of a proof of length 2, using
  only $\idr$ and $\ssur$:
\begin{equation}\label{eq:exup}
\od{\odi{\odi{\odh{\unit}
    }{\idr}{
      \begin{array}{c@{\quad\;\;}c@{\quad\;\;}c@{\quad\;\;}c}
        \vna1 & \vnb1 & \va1 & \vb1 \\
        \\
        \vnc1 &\vnd1 & \vc1 &\vd1
      \end{array}
      \edges{a1/c1,b1/d1,a1/d1,na1/nb1,nb1/nc1,d1/c1}
    }{}
  }{\ssur}{
    \begin{array}{c@{\quad\;\;}c@{\quad\;\;}c@{\quad\;\;}c}
      \vna1 & \vnb1 & \va1 & \vb1 \\
      \\
      \vnc1 &\vnd1 & \vc1 &\vd1
    \end{array}
    \edges{a1/c1,b1/d1,a1/d1,na1/nb1,nb1/nc1}
  }{}}
\end{equation}
where the $\ssur$ instance moves the module $d$ in the context consisting of vertices labelled $a,b,c$.
  Let us emphasise that the conclusion of the derivation in~\eqref{eq:exup} is not a formula but a graph. 
It establishes that the following implication is provable:
\begin{equation}\label{eq:exup2}
\begin{array}{c@{\quad\;\;}c@{\quad\;\;}c@{\quad\;\;}c}
       \va1 & \vb1 \\
      \\[-1ex]
        \vc1 &\vd1
    \end{array}
\edges{a1/c1,b1/d1,a1/d1,c1/d1}
\!\!\limp\;
    \begin{array}{c@{\quad\;\;}c@{\quad\;\;}c@{\quad\;\;}c}
      \va1 & \vb1 \\
      \\[-1ex]
        \vc1 &\vd1
    \end{array}
    \edges{a1/c1,b1/d1,a1/d1}
\end{equation}
which is a fact beyond the scope of formulas.
\end{exa}

\begin{rem}\label{rem:DIseq}
  In \cite{gug:gun:par:2010} it is shown how a derivation on formulas
  in the style of open deduction (as we defined it in
  Definition~\ref{def:derivation}), can be translated into a derivation
  in the style of the calculus of
  structures~\cite{gug:SIS,gug:str:01}, which can be seen as a
  sequence of rewriting steps on formulas. The same also works on derivations for
  graphs. To give an example, consider the derivation on the left
  below:
  \begin{equation}
    \label{eq:od}
    \small
    \od{\odi{\odh{\unit}}{\aidr}{
        \od{
          \odi{
            \odh{\od{\odi{\odh{\unit}}{\aidr}{a\lpar \cneg a}{}}
              \ltens
              b}}{\swir}{[a;(\cneg a; b)]}{}
        }
        \lpar
        \od{
          \odi{
            \odh{
              \cneg b
              \ltens
              \od{\odi{\odh{\unit}}{\aidr}{c\lpar \cneg c}{}}
          }}{\swir}{[(\cneg b;c);\cneg c]}{}
      }}{}}
    \quad\leadsto\quad
    \od{
      \odi{
        \odi{
          \odi{
            \odi{
              \odi{\odh{\unit}}{\aidr}{b\lpar\cneg b}{}
            }{\aidr}{[([a;\cneg a];b);\cneg b]}{}
          }{\swir}{[a;(\cneg a;b);\cneg b]}{}
        }{\aidr}{[a;(\cneg a;b);(\cneg b;[c;\cneg c])]}{}
      }{\swir}{[a;(\cneg a;b);(\cneg b;c);\cneg c]}{}
    }
  \end{equation}
  which can also be written as shown on the right above. (Note that we
  make use of the notation
  in~\eqref{eq:rule-context}). In~\cite{gug:gun:par:2010}, the
  derivation on the left in~\eqref{eq:od} is called \defn{synchronal}
  and the one on the right in~\eqref{eq:od} is called \defn{sequential}. For every synchronal
  derivation there is at least one sequential variant (but it is
  not unique). In later parts of the paper we sometimes speak of the ``bottommost'' rule
  instance in a derivation. In these cases, we assume that we have fixed some sequential form.
  Observe that both derivations in~\eqref{eq:od} above make
  use of~\eqref{eq:iso-rule}.
\end{rem}

As in other deep inference systems, we can give for the rules in~\eqref{eq:id-ssu} their \emph{duals}. In general, if
$$
\vlinf{\rr}{}{\gH}{\gG}
$$
is an instance of a rule, then
$$
\vlinf{\cneg\rr}{}{\cneg\gG}{\cneg\gH}
$$
is an instance of the dual rule. The duals of the two rules in~\eqref{eq:id-ssu} are the following:
\begin{equation}\label{eq:iu-ssd}
  \vlinf{\iur}{}{\unit}{\gA\ltens\cneg\gA}
  \qquad\qquad
  \vlinf{\ssdr}{~~S\subseteq\vB,~S\neq\emptyset}{\gB\lpar\gA}{\gB\coons{\gA}S}
\end{equation}
called \defn{identity up}\/ (or \defn{cut}) and \defn{super switch
  down}, respectively. 
We have the side condition $S\neq\emptyset$ to avoid a triviality, as
$\gB\coons{\gA}S=\gB\lpar\gA$ if $S=\unit$. Dually to $\ssur$, the
$\ssdr$-rule removes all the edges between $A$ and $B$.

\begin{exa}
  The implication in~\eqref{eq:exup2} can also be proven using only only $\ssdr$ and $\idr$ instead of $\ssur$ and $\idr$, as the following proof of length 3 shows:
\begin{equation}\label{eg:proofA}
	\od{\odi{\odi{\odi{\odh{\unit}
			}{\idr}{
				\begin{array}{c@{\quad\;\;}c@{\quad\;\;}c@{\quad\;\;}c@{\quad\;\;}c}
					\vna1 & \vnb1 & \va1 & \vb1  & \\
					\\
					\vnc1 &  & \vc1 
				\end{array}
				\edges{a1/c1,na1/nb1,nb1/nc1}
			}{}
		}{\idr}{
			\begin{array}{c@{\quad\;\;}c@{\quad\;\;}c@{\quad\;\;}c@{\quad\;\;}c}
				\vna1 & \vnb1 & \va1 & \vb1  & \\
				\\
				\vnc1 &  & \vc1 & \vmodule{1}{d\quad \cneg d}
			\end{array}
			\edges{a1/c1,a1/M1,b1/M1,na1/nb1,nb1/nc1}
		}{}
	}{\ssdr}{
		\begin{array}{c@{\quad\;\;}c@{\quad\;\;}c@{\quad\;\;}c}
			\vna1 & \vnb1 & \va1 & \vb1 \\
			\\
			\vnc1 & \vnd1 & \vc1 &\vd1
		\end{array}
		\edges{a1/c1,a1/d1,b1/d1,na1/nb1,nb1/nc1}
	}{}}
\end{equation}
\end{exa}

\begin{defi}
  Let $\sysS$ be an inference system. We say that an inference rule
  $\rr$ is \defn{derivable} in $\sysS$ iff 
  \begin{equation*}
    \mbox{for every instance}
    \hskip.5em
    \vlinf{\rr}{}{\gH}{\gG}
    \quad
    \mbox{there is a derivation}
    \hskip.5em
    \smash{\od{\odd{\odh{\gG}}{\dD}{\gH}{\sysS}}}
    \;\;.
  \end{equation*}
  We say that $\rr$ is \defn{admissible} in $\sysS$ iff
  \begin{equation*}
    \mbox{for every instance}
    \hskip.5em
    \vlinf{\rr}{}{\gH}{\gG}
    \quad
    \mbox{we have that~~}
    \proves[\sysS]{\gG}
    \mbox{~~implies~~}
    \proves[\sysS]{\gH}
    \;\,.
  \end{equation*}
\end{defi}

If $\rr\in\sysS$ then $\rr$ is trivially derivable and admissible in $\sysS$.

Most deep inference systems in the literature (e.g.~\cite{gug:str:01,brunnler:tiu:01,gug:SIS,str:02,gug:str:02,ross:MAV1}) contain the \defn{switch} rule:\footnote{This rule has also been used in~\cite{retore:99} and~\cite{retore:03}, and is also known under the names \emph{weak distributivity}~\cite{blute:etal:96} or \emph{dissociativity}~\cite{dosen:petric:coherence-book}.}
\begin{equation}
  \label{eq:switch}
  \vlinf{\swir}{}{A\lpar(B\ltens C)}{(A\lpar B)\ltens C}
\end{equation}
One can immediately see that it is its own dual and is a special case of both
$\ssdr$ and $\ssur$. We therefore have the following:

\begin{lemma}\label{lem:switch}
  If in an inference system $\sysS$ one of the rules $\ssdr$ and
  $\ssur$ is derivable, then so is~$\swir$.
\end{lemma}

\begin{rem}
  In a standard deep inference system for formulas we
  also have the converse of Lemma~\ref{lem:switch}, i.e., if $\swir$
  is derivable, then so are $\ssur$ and $\ssdr$ (see,
  e.g.,~\cite{dissvonlutz}). However, in the case of
  arbitrary graphs this is no longer true, and the rules $\ssur$ and
  $\ssdr$ are strictly more powerful than~$\swir$.
\end{rem}

\begin{lemma}\label{lem:dual}
  Let $\sysS$ be an inference system. If the rules $\idr$ and
  $\iur$ and $\swir$ are derivable in $\sysS$, then for every
  rule $\rr$ that is derivable in $\sysS$, also its dual $\cneg\rr$ is
  derivable in $\sysS$.
\end{lemma}

\begin{proof}
  Suppose we have two graphs $\gG$ and $\gH$, and a derivation from
  $\gG$ to $\gH$ in $\sysS$. Then it suffices to show that we can
  construct a derivation from $\cneg\gH$ to $\cneg\gG$ in $\sysS$:
  \begin{equation*}
    \odn{\odn{\unit}
          {\idr}{\cneg\gG\lpar\gG}{}\ltens\cneg\gH}
      {\swir}
      {\cneg\gG\lpar\odn{
          \ods{\gG}{}{\gH}{\sysS}\ltens\cneg\gH}
        {\iur}{\unit}{}}
      {}
  \end{equation*}
  Note that $\unit\ltens\cneg\gH=\cneg\gH$ and $\cneg\gG\lpar\unit=\cneg\gG$.
\end{proof}

\begin{lemma}\label{lem:trans}
  If the rules $\iur$ and $\swir$ are admissible for an inference
  system $\sysS$, then $\limp$ is transitive, i.e., if
  $\proves[\sysS]{\gG\limp\gH}$ and $\proves[\sysS]{\gH\limp\gK}$ then
  $\proves[\sysS]{\gG\limp\gK}$.
\end{lemma}

\begin{proof}
  We can construct the following derivation
  \begin{equation}\label{eq:cut}
    \odn{
      \ods{\unit}{}{\cneg\gG\lpar\gH}{\sysS}
      \ltens
      \ods{\unit}{}{\cneg\gH\lpar\gK}{\sysS}
    }{\swir}{\cneg\gG\lpar
      \odn{(\gH;[\cneg\gH;\gK])}
          {\swir}
          {\odn{\gH\ltens\cneg\gH}{\iur}{\unit}{}\lpar\gK}
          {}
    }{}
  \end{equation}
  from $\unit$ to $\cneg\gG\lpar\gK$ in $\sysS$.
\end{proof}

Lemma~\ref{lem:trans} is the reason why $\iur$ is also called \emph{cut}.
In a well-designed deep inference system for formulas, the two rules
$\idr$ and $\iur$ can be restricted in a way that they are only
applicable to atoms, i.e., replaced by the following two rules that we
call \defn{atomic identity down} and \defn{atomic identity up},
respectively:
\begin{equation}
  \vlinf{\aidr}{}{\cneg a\lpar a}{\unit}
  \qquand
  \vlinf{\aiur}{}{\unit}{a\ltens\cneg a}
\end{equation}
We would like to achieve something similar for our proof system on
graphs. For this it is necessary to be able to decompose prime graphs
into atoms, but the two rules $\ssdr$ and $\ssur$ cannot do this, as
they are only able to move around modules in a graph. For this reason,
we add the following two rules to our system:
\begin{equation}
  \vlinf{\pdr}{~P\text{~prime,~}\sizeof{\vP}\ge 4,~M_1\neq\unit,\ldots,M_n\neq\unit}
        {P\connn{M_1,\ldots,M_n}\lpar\cneg P\connn{N_1,\ldots,N_n}}
        {\vls([M_1;N_1];\cdots;[M_n;N_n])}
\end{equation}
called \defn{prime down}, and
\begin{equation}
  \vlinf{\pur}{~P\text{~prime,~}\sizeof{\vP}\ge 4~M_1\neq\unit,\ldots,M_n\neq\unit}
        {\vls[(M_1;N_1);\cdots;(M_n;N_n)]}
        {P\connn{M_1,\ldots,M_n}\ltens\cneg P\connn{N_1,\ldots,N_n}}        
\end{equation}
called \defn{prime up}. In both cases, the side condition is that
$\gP$ needs to be a prime graph and has at least $4$
vertices. 
We also require that for all $i\in\set{1,\ldots,n}$ we have that $M_i$ is non-empty
in an application of $\pdr$ and $\pur$.\footnote{In the conference version~\cite{Acclavio2020} of this paper,
we had the weaker condition that for all $i\in\set{1,\ldots,n}$ at
least one of $M_i$ and $N_i$ is non-empty. The advantage of the
stronger condition that we use here is that we obtain a more
controlled proof system as the prime graphs involved in
a $\pdr$ can only be those that appear in the modular decomposition
of the graph.}  
The reason for the side conditions on the $\pdr$ is
not that the rules would become unsound otherwise. In fact, these
rules without side conditions are admissible in the general case. The
details will be discussed in the proof of Lemma~\ref{lem:g} and at the end of Section~\ref{sec:upfrag}.

\begin{exa}
  Below is a derivation of length 5 using the
  $\pdr$-rule, and proving that a prime graph implies itself. On the left, we use the ``formula style'' notation of the modular decomposition of graphs, as in~\eqref{eq:modform}, and on the right we use the 
  graphical notation, as in~\eqref{eq:spaghetti} and~\eqref{eq:modular}.
  \vadjust{\vskip-2ex}
\begin{equation*}
  \small
  \scalebox{.9}{$
	\odn{
		\odn{\unit}{\aidr}{\cneg{a} \lpar a}{}\ltens
		\odn{\unit}{\aidr}{\cneg{b} \lpar b}{}\ltens
		\odn{\unit}{\aidr}{\cneg{c} \lpar c}{}\ltens
		\odn{\unit}{\aidr}{\cneg{d} \lpar d}{}
	}{\pdr}{
		\cneg\Pfour\connn{\cneg a, \cneg b, \cneg c, \cneg d}
		\lpar 
		\Pfour\connn{ a,  b,  c,  d}
	}{}
        $}
\qquad
  \small
  \scalebox{.9}{$
	\od{\odi{
			\odi{
				\odh{\unit}
			}{4\times \aidr}{
				\begin{array}{c@{\hskip1.5em}c@{\qquad}c@{\hskip1.5em}c}
					\\
					\vmodule1{\begin{array}{c}a\\\cneg a\end{array}}&
                                        &&
					\vmodule4{\begin{array}{c}d\\\cneg d\end{array}}
                                        \\[-3ex]
                                        &
					\vmodule2{\begin{array}{c}b\\\cneg b\end{array}}&
					\vmodule3{\begin{array}{c}c\\\cneg c\end{array}}&
                                        \\[2.5ex]
				\end{array}
				\edges{M1/M2,M2/M3,M3/M4}
				\bentedges{M2/M4/20,M1/M3/20,M1/M4/20}
			}{}
		}{\pdr}
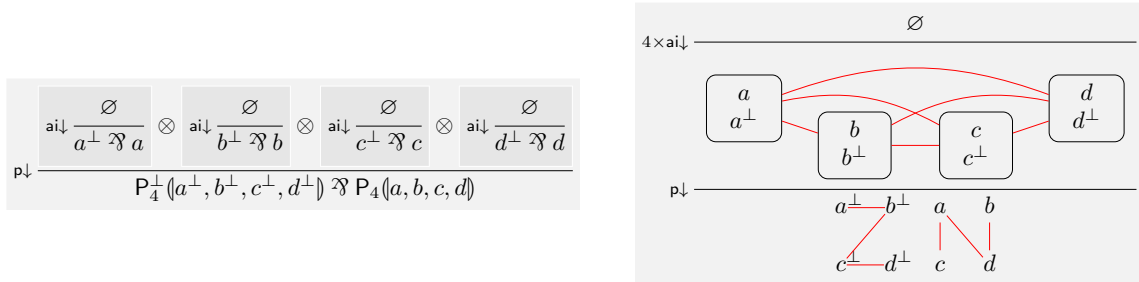
{
			\begin{array}{c@{\quad\;\;}c@{\quad\;\;}c@{\quad\;\;}c}
				\vna1 & \vnb1 & \va1 & \vb1 \strut\\
				\\
				\vnc1 & \vnd1 & \vc1 &\vd1
			\end{array}
			\edges{a1/c1,a1/d1,b1/d1,na1/nb1,nb1/nc1,nc1/nd1}
	  }{}}
        $}
\end{equation*}
But let us emphasize, that in both cases it is the same derivation. We could also use a mixed notation, as shown below.
  \begin{equation*}\label{eg:proofN}
    \small
    \odn{
      (
      \odn{\unit}{\aidr}{[\cneg{a} ; a]}{};
      \odn{\unit}{\aidr}{[\cneg{b} ; b]}{};
      \odn{\unit}{\aidr}{[\cneg{c} ; c]}{};
      \odn{\unit}{\aidr}{[\cneg{d} ; d]}{}
      )
    }{\pdr}
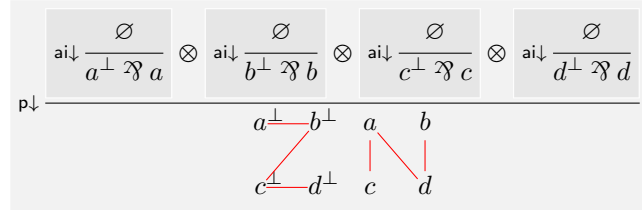
{
      \begin{array}{c@{\quad\;\;}c@{\quad\;\;}c@{\quad\;\;}c}
        \vna1 & \vnb1 & \va1 & \vb1 \\
        \\
        \vnc1 & \vnd1 & \vc1 &\vd1
      \end{array}
      \edges{a1/c1,a1/d1,b1/d1,na1/nb1,nb1/nc1,nc1/nd1}
    }{}
  \end{equation*}
\end{exa}

This completes the presentation of
our system, which is shown in Figure~\ref{fig:SGS}, and formally defined below.

\begin{figure*}[!t]
	$
\boxed{	
		\begin{array}{cc}
				&\SGS
			\\
			\;	\boxed{
						\begin{array}{cc}
								&\GS
							\\
								\vlinf{\aidr}{}{\cneg a\lpar a}{\unit}
								\qquad\qquad
								\vlinf{\ssdr}{
									~~S\subseteq\vB,~S\neq\emptyset,~\gA \neq \emptyset
								}{\gB\lpar\gA}{\gB\coons{\gA}S}
							\\[2.2em]
							\qquad
								\vlinf{\pdr}{
									P\text{~prime},~
									\sizeof\vP\ge4,~
									M_1\neq \unit,~\dots,~M_n\neq \unit
								} 
								{\cneg P\connn{M_1,\ldots,M_n}\lpar P\connn{N_1,\ldots,N_n}}
								{(M_1 \lpar N_1)\ltens \cdots \ltens(M_n \lpar N_n)}
							\\[2.2em]
					\end{array}
				}
			\\
			\\
			\begin{array}{cc}
					\vlinf{\aiur}{}{\unit}{a\ltens\cneg a}
					\qquad\qquad
					\vlinf{\ssur}{~~S\subseteq\vB,~S\neq\vB,~\gA \neq \emptyset}{\gB\coons{\gA}S}{\gB\ltens\gA}
				\\[2.2em]
					\vlinf{\pur}{
						P\text{~prime},~
						\sizeof\vP\ge4,~
						M_1\neq \unit,~\dots,~M_n\neq \unit 
					}
					{(M_1\ltens N_1)\lpar \cdots \lpar (M_n\ltens N_n)}
					{P\connn{M_1,\ldots,M_n}\ltens\cneg P\connn{N_1,\ldots,N_n}}
				\\[2.2em]
			\end{array}
	\end{array}
}
$
\medskip
  \caption{The inference rules for systems $\SGS$ 
    and $\GS$.
}
  \label{fig:SGS}
\end{figure*}

\begin{defi}\label{def:system}\sloppy
  We define \defn{system $\SGS$} to be the set
  $\set{\aidr,\ssdr,\pdr,\pur,\ssur,\aiur}$ of inference rules shown
  in Figure~\ref{fig:SGS}. The \defn{down-fragment}
  (resp.~\defn{up-fragment}) of $\SGS$ consists of the rules
  $\set{\aidr,\ssdr,\pdr}$ (resp.~$\set{\aiur,\ssur,\pur}$) and is
  denoted by $\SGSd$ (resp.~$\SGSu$). The down-fragment $\SGSd$ is
  also called \defn{system $\GS$}.
\end{defi}


\section{Properties of the Proof System}
\label{sec:properties}

The properties that we establish here are standard for deep inference systems (see, e.g., \cite{gug:str:01,esslli19}).
The first observation about $\SGS$ is that the general forms of the
identity rules $\idr$ and $\iur$ are derivable, as we will see in
Corollary~\ref{cor:identity} below. Before, we show an auxiliary lemma stating the corresponding result for the
prime rules, namely that their general form is also derivable, i.e., they
can be applied to any graph instead of only prime graphs.

\begin{lemma}\label{lem:g}
  Let $\gG$  and 
  $\gM_1,\dots,\gM_n, \gN_1, \dots, \gN_n$ be graphs with $\sizeof\vG=n$,
  such that for every $i\in\set{1,\ldots,n}$, 
  we have that $\gM_i=\unit$ implies $N_i=\unit$.
  Then there are derivations
  \begin{equation}\label{eq:gd}
    \ods{(\gM_1\lpar\gN_1)\ltens\cdots\ltens(\gM_n\lpar\gN_n)}
        {}{\gG\connn{\gM_1,\ldots,\gM_n}\lpar\cneg\gG\connn{\gN_1,\ldots,\gN_n}}{\SGSd}
  \end{equation}
  and dually
  \begin{equation}\label{eq:gu}
    \ods{\gG\connn{\gM_1,\ldots,\gM_n}\ltens\cneg\gG\connn{\gN_1,\ldots,\gN_n}}
        {}{(\gM_1\ltens\gN_1)\lpar\cdots\lpar(\gM_n\ltens\gN_n)}{\SGSu}
  \end{equation}
\end{lemma}

\begin{proof}\sloppy
  For showing~\eqref{eq:gd}, we proceed by induction on $\sizeof\vG$.  First consider the
  case that there is an $i\in\set{1,\ldots,n}$ such that
  $\gN_i=\emptyset=\gM_i$. Without loss of generality, we can assume
  that $i=1$. Then there is a graph $\gH$ with $\sizeof\vH=n-1$ and
  $\gG\connn{\gM_1,\ldots,\gM_n}=\gH\connn{\gM_2,\ldots,\gM_n}$ and
  $\cneg\gG\connn{\gN_1,\ldots,\gN_n}=\cneg\gH\connn{\gN_2,\ldots,\gN_n}$,
  and therefore we have the derivation
  \begin{equation*}
    \mod{\odo{\odd{\odo{\odh{
            (\unit\lpar\unit)\ltens(\gM_2\lpar\gN_2)\ltens\cdots\ltens(\gM_n\lpar\gN_n)}}{
          \isom}{(\gM_2\lpar\gN_2)\ltens\cdots\ltens(\gM_n\lpar\gN_n)}{}}{
          \dD}{\gH\connn{\gM_2,\ldots,\gM_n}\lpar\cneg\gH\connn{\gN_2,\ldots,\gN_n}}{\SGSd}}{
        \isom}{\gG\connn{\unit,\gM_2,\ldots,\gM_n}\lpar\cneg\gG\connn{\unit,\gN_2,\ldots,\gN_n}}{}}
  \end{equation*}
  where $\dD$ exists by induction hypothesis. We can therefore now
  assume that $\gM_i\neq\unit$ for all $i\in\set{1,\ldots,n}$. 
  We now make a case analysis on $\gG$ using Lemma~\ref{lem:newhope}:
  \begin{enumerate}[(i)]
  \item If $\gG$ is a singleton graph, the statement holds trivially.
  \item If $\gG=\gA\ltens\gB$ then
    $\gG\connn{\gM_1,\ldots,\gM_n}=\gA\connn{\gM_1,\ldots,\gM_k}\ltens\gB\connn{\gM_{k+1},\ldots,\gM_n}$
    and
    $\cneg\gG\connn{\gN_1,\ldots,\gN_n}=\cneg\gA\connn{\gN_1,\ldots,\gN_k}\lpar\cneg\gB\connn{\gN_{k+1},\ldots,\gN_n}$
    for some $1\le k\le n$. We therefore have
    \begin{equation*}
    \scalebox{.95}{$\strut\qquad
        \nvls{\odn{
          \ods{(\gM_1\lpar\gN_1)\ltens\cdots\ltens(\gM_k\lpar\gN_k)}
              {\dD_1}{\gA\connn{\gM_1,\ldots,\gM_k}\lpar\cneg\gA\connn{\gN_1,\ldots,\gN_k}}{\SGSd}
              \ltens
              \ods{(\gM_{k+1}\lpar\gN_{k+1})\ltens\cdots\ltens(\gM_n\lpar\gN_n)}
                  {\dD_2}{\gB\connn{\gM_{k+1},\ldots,\gM_n}\lpar\cneg\gB\connn{\gN_{k+1},\ldots,\gN_n}}{\SGSd}
        }
            {\ssdr}{
              \odn{(\gA\connn{\gM_1,\ldots,\gM_k}\lpar\cneg\gA\connn{\gN_1,\ldots,\gN_k})\ltens\gB\connn{\gM_{k+1},\ldots,\gM_n}}{
                \ssdr}{
                (\gA\connn{\gM_1,\ldots,\gM_k}\ltens\gB\connn{\gM_{k+1},\ldots,\gM_n})\lpar\cneg\gA\connn{\gN_1,\ldots,\gN_k}}{}                
              \lpar\cneg\gB\connn{\gN_{k+1},\ldots,\gN_n}}{}
      }
    $}
   \end{equation*}
    where $\dD_1$ and $\dD_2$ exist by induction hypothesis. Note that it is possible that $\cneg\gA\connn{\gN_1,\ldots,\gN_k}=\unit$ or $\cneg\gB\connn{\gN_{k+1},\ldots,\gN_n}=\unit$ (or both). In that case one (or both) of the two instances of $\ssdr$ is vacuous.
  \item If $\gG=\gA\lpar\gB$, we proceed similarly.
  \item If $\gG=\gP\connn{\gA_1,\ldots,\gA_k}$ for $\gP$ prime and
    $k=\sizeof{\vP}\ge 4$ and $\gA_l\neq\unit$ for each
    $l\in\set{1,\ldots,k}$, then we have that
    $\gG\connn{\gM_1,\ldots,\gM_n}=\gP\connn{\gA_1\connn{\gM_{11},\ldots,\gM_{1h_1}},\ldots,\gA_k\connn{\gM_{k1},\ldots,\gM_{kh_k}}}$
    where $\set{\gM_1,\ldots,\gM_n}=\cM_1\cup\cdots\cup\cM_k$ and
    $\cM_l=\set{\gM_{l1},\ldots,\gM_{lh_l}}$ and
    $\cM_l\cap\cM_j=\emptyset$ for $l\neq j$. Similarly, we have
    $\cneg\gG\connn{\gN_1,\ldots,\gN_n}=\cneg\gP\connn{\cneg{\gA_1}\connn{\gN_{11},\ldots,\gN_{1h_1}},\ldots,\cneg{\gA_k}\connn{\gN_{k1},\ldots,\gN_{kh_k}}}$
    where $\set{\gN_1,\ldots,\gN_n}=\cN_1\cup\cdots\cup\cN_k$ and
    $\cN_l=\set{\gN_{l1},\ldots,\gN_{lh_l}}$ and
    $\cN_l\cap\cN_j=\emptyset$ for $l\neq j$.  Since $\gM_i\neq\unit$ for all $i\in\set{1,\ldots,n}$, we also have that 
	{$\gA_l\connn{\gM_{l1},\ldots,\gM_{lh_l}}\neq\unit$ for each $l\in\set{1,\ldots,k}$}%
	.
	Therefore, we have the following derivation
    \begin{equation*}
    \scalebox{.9}{$\strut\qquad%
        \nvls{\odn{
          \ods{(\gM_{11}\lpar\gN_{11})\ltens\cdots\ltens(\gM_{1h_1}\lpar\gN_{1h_1})}
              {\dD_1}{\cneg{\gA_1}\connn{\gM_{11},\ldots,\gM_{1h_1}}\lpar\gA_1\connn{\gN_{11},\ldots,\gN_{1h_1}}}{\SGSd}
              \ltens
              \cdots
              \ltens
              \ods{(\gM_{k1}\lpar\gN_{k1})\ltens\cdots\ltens(\gM_{kh_k}\lpar\gN_{kh_k})}
                  {\dD_k}{\cneg{\gA_k}\connn{\gM_{k1},\ldots,\gM_{kh_k}}\lpar\gA_k\connn{\gN_{k1},\ldots,\gN_{kh_k}}}{\SGSd}
        }
            {\pdr}{
              \gP\connn{\gA_1\connn{\gM_{11},\ldots,\gM_{1h_1}},\ldots,\gA_k\connn{\gM_{k1},\ldots,\gM_{kh_k}}}                
              \lpar
              \cneg\gP\connn{\cneg{\gA_1}\connn{\gN_{11},\ldots,\gN_{1h_1}},\ldots,\cneg{\gA_k}\connn{\gN_{k1},\ldots,\gN_{kh_k}}}
              }{}
        }
    $}
   \end{equation*}
    where $\dD_1,\ldots,\dD_k$ exist by induction hypothesis. 
  \end{enumerate}
  The derivation in~\eqref{eq:gu} can be constructed dually.
\end{proof}

\begin{cor}\label{cor:identity}
  The rule $\idr$ is derivable in $\SGSd$, and dually, $\iur$ is derivable in $\SGSu$.
\end{cor}

\begin{proof}
  Each graph $\gG$, can be written as $\gG=\gG\connn{a_1, \dots, a_n}$ where $\vG=\set{a_1, \dots, a_n}$.
  Since ${a_1}, \dots, {a_n}$ are non-empty modules of $ \gG$ we can apply Lemma~\ref{lem:g}
  to obtain the following derivation of $\gG\lpar\cneg\gG$ from $\unit$ in $\SGSd$:
	  \begin{equation}\label{eq:gd1}
		\ods{
			\odn{\unit}{\aidr}{a_1\lpar\cneg{a_1}}{}
			\ltens\cdots\ltens
			\odn{\unit}{\aidr}{a_n\lpar\cneg{a_n}}{}}
		    {\dD}{\gG\connn{a_1,\ldots,a_n}\lpar\cneg\gG\connn{\cneg {a_1},\ldots,\cneg {a_n}}}{
                      \SGSd}
	  \end{equation}
          where $\dD$ exists by Lemma~\ref{lem:g}.
Dually, we can show that for any $\gG$ there is a derivation from $\gG\ltens\cneg\gG$ to $\unit$ in $\SGSu$. 	
\end{proof}

\begin{cor}\label{cor:multitensTOany}
	If $\gG,\gM_1, \dots, \gM_n$ are non-empty graphs, then there is a derivation 
\begin{equation*}
\ods{{\gM_1\ltens \cdots \ltens \gM_n}}{}{\gG\connn{\gM_1, \dots, \gM_n}}{\GS}
\end{equation*}
\end{cor}
\begin{proof}
  This follows from Lemma~\ref{lem:g}, using $\cneg\gG\connn{\unit,\ldots,\unit}$.
\end{proof}

Next, observe that Lemmas~\ref{lem:dual}
and~\ref{lem:trans} hold for system $\SGS$. In particular, we have that if $\proves[\SGS]{A\limp B}$ and $\proves[\SGS]{B\limp C}$ then $\proves[\SGS]{A\limp C}$ since $\iur\in\SGS$. The main result of this paper is that Lemma~\ref{lem:trans} also holds for $\GS$.
More precisely, we have the following theorem:

\begin{thm}[Cut Admissibility]\label{thm:cut}
  The rule $\iur$ is admissible for $\GS$.
\end{thm}

To prove this theorem, we will show that the whole up-fragment of $\SGS$ is admissible for $\GS$.

\begin{thm}\label{thm:up}
  The rules $\aiur$, $\ssur$, $\pur$ are admissible for $\GS$.
\end{thm}

Then Theorem~\ref{thm:cut} follows immediately from
Theorem~\ref{thm:up} and the second statement in
Corollary~\ref{cor:identity}.

The following two sections are devoted to the proof of
Theorem~\ref{thm:up}. But before, let us finish this section by
exhibiting some immediate consequences of Theorem~\ref{thm:cut}.

\begin{cor}\label{cor:gs-sgs}
  For every graph $\gG$, we have $\proves[\SGS]\gG$ iff $\proves[\GS]\gG$. 
\end{cor}

\begin{cor}\label{cor:deduction}
  For all graphs $G$ and $H$, we have
  \begin{equation*}
    \proves[\GS]{\gG\limp\gH}
    \iff
    \odv{\unit}{}{\cneg\gG\lpar\gH}{\GS}
    \iff
    \odv{\unit}{}{\cneg\gG\lpar\gH}{\SGS}
    \iff
    \odv{\gG}{}{\gH}{\SGS}
  \end{equation*}
\end{cor}

\begin{proof}
  The first equivalence is just the definition of $\proves{}$. The
  second equivalence follows from Theorem~\ref{thm:up}, and the last equivalence follows from the two derivations
  \begin{equation*}
    \odn{\unit}{
      \idr}{
      \cneg\gG\lpar
      \odv{\gG}{}{\gH}{}}{}
    \quand
    \odn{\gG\ltens
      \odv{\unit}{}{\cneg\gG\lpar\gH}{}}{
      \ssdr}{
      \odn{\gG\ltens\cneg\gG}{\iur}{\unit}{}
      \lpar
      \gH}{}
  \end{equation*}
  together with Corollary~\ref{cor:identity}.
\end{proof}

\begin{cor}\label{cor:context}
  For all graphs $\gG$ and $\gH$ and all contexts $\gC\coons{\cdot}R$, we have that
  $$
  \text{$
  	\proves[\GS]{\gG\limp\gH}\qquad\Longrightarrow\qquad
    \proves[\GS]{\gC\coons{\gG}R\limp\gC\coons{\gH}R}
    $}\quadfs
  $$
\end{cor}

\begin{proof}
  This is a consequence of Corollary~\ref{cor:deduction} and
  Remark~\ref{rem:context}. But it can also be proved directly using
  Lemma~\ref{lem:g}.
\end{proof}

\begin{cor}\label{cor:tensor}
  We have\;\; $\proves[\GS]{A\ltens B}$ \;\;$\iff$\/\;\; $\proves[\GS]{A}$ \;and\/\; $\proves[\GS]{B}$.
\end{cor}

\begin{proof}
  This follows immediately by inspecting the inference rules of $\GS$. 
\end{proof}

\begin{cor}\label{cor:prime}
  We have\; $\proves[\GS]{P\connn{M_1,\ldots,M_n}}$ with $P$ prime and $n\ge 4$ and $M_i\neq\unit$ for all $i=\set{1,\ldots,n}$, if and only if there is at least one $i=\set{1,\ldots,n}$ such that
  \;$\proves[\GS]{M_i}$\; and 
  \;$\proves[\GS]{P\connn{M_1,\ldots,M_{i-1,},\unit,M_{i+1},\ldots,M_n}}$\;.
\end{cor}

\begin{proof}
 As before, this follows immediately by inspecting the inference rules
 of $\GS$. In fact, this can be seen as a generalisation of the
 previous corollary.
\end{proof}

\begin{cor}[consistency]\label{cor:consistency}
If $\gG \neq \unit$ and $\proves[\GS]{\gG}$, then $\not\proves[\GS]{\cneg{\gG}}$.
\end{cor}
\begin{proof}
Since $\gG \neq \unit$ and $\proves[\GS]{\gG}$, we have for some $a$ that
$\ods{\odn{\unit}{\aidr}{\cneg{a} \lpar a}{}}{}{\gG}{\GS}$.
Hence, by Corollary~\ref{cor:deduction} we have $\proves[\GS]{\cneg{\gG} \multimap a \ltens \cneg{a}}$.
Thus, if we assume by way of contradiction that $\proves[\GS]{\cneg{\gG}}$ holds then by Theorem~\ref{thm:cut}
we have $\proves[\GS]{a \ltens \cneg{a}}$, which is impossible.
\end{proof}

\begin{obs}\label{obs:cook-reckhow}
  The system $\GS$ forms a proof system in the sense of Cook and
  Reckhow~\cite{Cook1979}, as the time complexity of checking the
  correct application of each inference step is polynomial. Note that the modular
  decomposition of graphs can be obtained in linear
  time~\cite{mcc:ros:spi:linear}. Furthermore, the use of the graph isomorphisms
  for the composition of derivations, as in~\eqref{eq:iso}, does not
  increase complexity, as we assume that the isomorphism $f$ is
  explicitly given, such that the correctness of the operation can
  also be checked in linear time.
\end{obs}

Let us now define the notion of \emph{size}
of a graph that will play a central role in the normalisation proof as 
an induction measure.

\begin{defi}\label{def:size}
  The \emph{size} of a graph $\gG$, denoted by $\gsize{\gG}$, is the
  lexicographic pair $\tuple{\sizeof\vG,\sizeof\enG}$.  That is, if
  $\gG$ and $\gH$ are graphs, then $\gless{\gG}{\gH}$ if
  $\sizeof{\vG}<\sizeof{\vH}$ or if $\sizeof{\vG}=\sizeof{\vH}$ and
  $\sizeof{\enG}<\sizeof{\enH}$.
\end{defi}

\begin{obs}\label{obs:size}
Every
inference rule discussed so far, except for $\aiur$ and $\iur$ do
strictly decrease the size of a graph when going bottom-up in a
derivation. That is, whenever we have an instance
${\vlinf{\rr}{}{\gG}{\gH}}$ of an inference rule
$\rr\in\set{\aidr,\idr,\ssdr,\ssur,\pdr,\pur,\gdr,\gur}$, then
$\gless{\gH}{\gG}$.
Since each such inference step either reduces the number of vertices or the number of edges in the dual graph (or both), and these  numbers are  never increased, the length of a derivation in
$\GS$ with conclusion $\gG$ is bound by $n^2+n$ where $n=\sizeof\vG$ is
the number of vertices in~$\gG$.\footnote{For the case that the derivation contains only cographs, it can in fact be shown that the length
    is bound by $\frac12n^2$ (see~\cite{bru:str:swi-med}), but as the details
    are not needed for this paper, we leave them to the reader.}
\end{obs}

Furthermore, to each graph only finitely many different non-trivial inference steps can be applied, and every derivation can be reorganised such that between any two (non-trivial) inference steps (see Remark~\ref{rem:DIseq}) there is at most one isomorphism step.  
Hence, it follows immediately that provability in $\GS$ is decidable and in
{\bf NP}. In Section~\ref{sec:MLL} we are going to show that it is also \NP-hard, and therefore \NP-complete.


\section{Splitting and Context Reduction}\label{sec:splitting}

We are now ready to present the fundamental properties of the system
$\GS$, namely the \emph{splitting lemmas} and \emph{context
  reduction}, that allow to prove cut elimination in the next section.\footnote{Proceeding via splitting and context reduction is not the only method to prove cut elimination for a deep inference proof system (see e.g.~\cite{brunnler:phd,str:MELL,acc:gue:21}), but it is the most organized and most universal method, and at the current state of art, it seems to be the only one that is applicable to $\GS$.}
First, recall that the standard syntactic method for proving cut elimination in the
sequent calculus is to permute the cut rule upwards in the proof,
while decomposing the cut formula along its main connective, and so
inductively reduce the cut rank.  However, in systems formulated using
deep inference this method cannot be applied, as derivations can be
constructed in a more flexible way than in the sequent calculus. For
this reason, the \emph{splitting} technique has been developed in the
literature on deep
inference~\cite{gug:SIS,dissvonlutz,SIS-V,ross:MAV1,tubella:phd}.  In
this section we discuss how to prove the splitting and the context
reduction lemmas for our system and how their proofs differs form the
ones in the literature.  The detailed proofs of the statement in this
section can be found in Appendix~\ref{sec:splittingproofs}.

\subsection*{From sequent calculus to splitting for graphs.}
Consider the typical rule for $\ltens$ in the sequent calculus.
\[
		\vliiinf{\tensrule}{}{\Gamma, \phi \ltens \psi, \Delta}{\Gamma, \phi \;\;}{}{\;\; \Delta, \psi}
\]
The effect of the above rule can be simulated by applying the following splitting lemma for the prime graph $\ltens$.
\begin{lemma}[Splitting Tensor]\label{lem:splitting:tens}
	Let $\gG$, $A$ and $B$ be graphs.
	If $\proves[\GS]{\gG\lpar (A\ltens B)}$ then there is a context $\gC$ and there are graphs $\gK_A$ and $\gK_B$ 
	such that 
	there are derivations
	\begin{equation*}
		\ods{\gC\coons{\gK_A \lpar \gK_B}{R}}{\dD_G}{\gG}{\GS}
		\qomma		
		\ods{\unit}{\dD_A}{\gK_A \lpar A}{\GS}
		\qomma
		\ods{\unit}{\dD_B}{\gK_B \lpar B}{\GS}
		\quand		
		\ods{\unit}{\dD_C}{\gC}{\GS}
		\quadfs
	\end{equation*}	  
\end{lemma}
The graphs $\gA$ and $\gB$ in Lemma~\ref{lem:splitting:tens}
 play the role of formulas $\phi$ and $\psi$ in the sequent calculus rule $\tensrule$;
while $\gK_A$ and $\gK_B$ play the role of sequents $\Gamma$ and $\Delta$.
Thus, when we say in Lemma~\ref{lem:splitting:tens} that $\proves[\GS]{ \gK_A \lpar A }$ and $\proves[\GS]{ \gK_B \lpar B }$ hold,  this corresponds to the provability of the premises $\Gamma, \phi$ and $\Delta, \psi$ in the sequent calculus $\tensrule$-rule.
This means that Lemma~\ref{lem:splitting:tens} allows us to rewrite a \textit{sequent-like graph} $\gG\lpar (A\ltens B)$ into a form where the effect of the $\tensrule$-rule of the sequent calculus can be simulated.

\begin{rem}
  Note that this splitting is fundamentally different from the
  well-known \emph{splitting tensor lemma} for MLL proof
  nets~\cite{girard:87,danos:regnier:89,retore:96:tokyo}. Whereas the
  MLL splitting tensor lemma says that \emph{there exists} at least
  one tensor that is splitting, our splitting tensor lemma says that
  \emph{every} tensor is splitting.\footnote{In the language of proof nets, our
  splitting lemma says that every tensor is splitting for its empire
  \emph{and} the proof system is expressive enough to reduce
  everything else without touching that empire. However, for our graphs, we do not (yet) have a notion of \emph{empire}.}
\end{rem}

Let us use the derivation in~\eqref{eg:proofA}, recalled below in~\eqref{eg:proofA1} for convenience, as an example to illustrate
this pattern.
Observe that the conclusion is of the form $\gG \lpar ((\cneg{a} \lpar \cneg{c}) \ltens \cneg{b})$ for some graph $\gG$, 
and that after applying $\ssdr$ and $\aidr$ we have split the graph $\gG$ into two graph $\gK_A = a \ltens c$ and $\gK_B = b$,
such that $\gK_A \lpar \cneg{a} \lpar \cneg{c}$ and $\gK_B \lpar \cneg{b}$ are provable.
\begin{equation}
  \label{eg:proofA1}
	\od{\odi{\odi{\odi{\odh{\unit}}{\idr}{
				\begin{array}{c@{\quad\;\;}c@{\quad\;\;}c@{\quad\;\;}c@{\quad\;\;}c}
					\vna1 & \vnb1 & \va1 & \vb1  & \\
					\\
					\vnc1 &  & \vc1 
				\end{array}
				\edges{a1/c1,na1/nb1,nb1/nc1}
			}{}
		}{\aidr}{
			\begin{array}{c@{\quad\;\;}c@{\quad\;\;}c@{\quad\;\;}c@{\quad\;\;}c}
				\vna1 & \vnb1 & \va1 & \vb1  & \\
				\\
				\vnc1 &  & \vc1 & \vmodule{1}{\vnd1 \quad \vd1}
			\end{array}
			\edges{a1/c1,a1/M1,b1/M1,na1/nb1,nb1/nc1}
		}{}
	}{\ssdr}{
		\begin{array}{c@{\quad\;\;}c@{\quad\;\;}c@{\quad\;\;}c@{\quad\;\;}c}
			\vna1 & \vnb1 & \va1 & \vb1 \\
			\\
			\vnc1 & \vnd1 & \vc1 &\vd1
		\end{array}
		\edges{a1/c1,a1/d1,b1/d1,na1/nb1,nb1/nc1}
	}{}}
\end{equation}
Thus applying Lemma~\ref{lem:splitting:tens} to the sub-graph $(\cneg{a} \lpar \cneg{c}) \ltens \cneg{b}$ above leads us to proof \eqref{eg:proofA1}.

There is a crucial difference between the statement of Lemma~\ref{lem:splitting:tens} and other splitting lemmas in the  in the literature on deep inference~\cite{gug:SIS,dissvonlutz,SIS-V,ross:MAV1,tubella:phd}, namely the need of the context $\gC$.
 To see that this context is necessary, consider the following example graph:
\begin{equation}\label{eg:deep-ten}
    \begin{array}{c@{\quad\;\;}c@{\quad\;\;}c@{\quad\;\;}c@{\quad\;\;}c@{\quad\;\;}c@{\quad\;\;}c@{\quad\;\;}c@{\quad\;\;}c@{\quad\;\;}c}
     \vng1& \vnf1 & \vmodule{3}{\begin{array}{c@{\quad\;\;}c} \vna1 &  \vnb1 \; \end{array}}  & \vf1 &     \va1  &  \vb1 \\
     & & &  \vg1   
    \end{array}
    \edges{a1/b1,nf1/ng1,g1/M3,nf1/M3,f1/M3}  
\end{equation}
which is of the form $\gG \lpar (a \ltens b)$ for some $\gG$
and
it is provable in $\GS$.
It is impossible however to apply a derivation to $\gG$ to obtain two disjoint graphs $\gK_A$ and $\gK_B$, such that $\gK_A \lpar a$ and $\gK_B \lpar b$ are provable, unless we define a suitable context for $\gK_A \lpar \gK_B$.
Indeed we do have the following context formed from the provable graph $\gC = \cneg{g} \lpar \cneg{f} \lpar (g \ltens f)$.
{
\[
\gC\coonso{\set{\cneg f, f, g}}=
\begin{array}{c@{\quad\;\;}c@{\quad\;\;}c@{\quad\;\;}c@{\quad\;\;}c@{\quad\;\;}c@{\quad\;\;}c@{\quad\;\;}c@{\quad\;\;}c@{\quad\;\;}c}
	\vng1 &\vnf1& \vmodule1{\quad}   & \vf1 \\
	& & &  \vg1
\end{array}
\edges{nf1/ng1,g1/M1,nf1/M1,f1/M1}
\]
}
Following the naming convention in the statement of Lemma~\ref{lem:splitting:tens},
we have $\gK_A = \cneg{a}$ and $\gK_A = \cneg{b}$ such that $\proves[\GS]{ \gK_A \lpar a }$ and $\proves[\GS]{ \gK_B \lpar b }$ hold, and furthermore $\gG$ is formed from plugging $\gK_A \lpar \gK_B$ insider the above mentioned context formed from $\gC$.

From the above ingredients satisfying the condition of the splitting lemma we can always construct a proof.
We can apply $\ssdr$ twice to bring $\gK_A$ and $a$ together inside a context formed from $\gC$, and similarly for $\gK_B$ and $b$.
After having completed the proofs that consume $a$ and $b$ (by a simple application of $\aidr$) we are left with the provable graph $\gC$.
The important observation here is that it is impossible to apply a derivation that changes the shape of $\gC$ until we have completed the derivation that consumes $\cneg{a}$ and $\cneg{b}$.

The problem illustrated above with the context $\gC$, that cannot be removed until the end of the proof, is a fundamental problem that we have with graphs and that we do not have with formulas. Essentially it is due to the presence of $\pfour$ induced sub-graphs that create indirect dependencies between atoms, such as the indirect dependency between $\cneg{g}$ and $g$ in the above example.

\subsection*{Generalising splitting to prime graphs.}
The general idea of
a splitting lemma is that, in a provable ``sequent-like graph'', consisting of
a number of disjoint connected components, we can select any of these
components as the \emph{principal graph} and apply a derivation to
the other components, such that eventually a rule breaking down the
principal component can be applied.  This allows us to approximate the
effect of applying rules in the sequent calculus, even when no sequent calculus exists, as is the case when the principal graph is formed using a prime graph that is not a tensor.

\begin{restatable}[Splitting Prime]{lemma}{lemSplitPrime}\label{lem:splitting:prime}
	Let $\gG$ be a graph
	and
	$\gP\neq \lpar $ be a prime graph with $\sizeof{\vP}=n$, and $\gM_1,\dots, \gM_n$ be non-empty graphs.
	If $\proves[\GS]{\gG\lpar \gP\connn{\gM_1, \dots, \gM_n}}$, then one of the following holds:

	\begin{enumerate}[(A)]
		\item\label{prime:A}
		either there is a context $\gC\coonso R$ 
		and graphs $\gK_1$, \dots, $\gK_n$, 
		such that  
		there are derivations
		\begin{equation*}
                  \scalebox{.92}{$
			\ods{\gC\coons{\cneg\gP\connn{\gK_1,\ldots,\gK_n}}R}{\dD_\gG}{\gG}{\GS}
			\qomma
			\ods{\unit}{\dD_i}{\gK_i\lpar\gM_i}{\GS}
			\quand
			\ods{\unit}{\dD_C}{\gC}{\GS}
                        $}
			\hskip-1em
		\end{equation*}
		for all $i\in\set{1,\ldots,n}$,

		\item\label{prime:B}
		or there is a context $\gC\coonso R$ and graphs $\gK_X$ and $\gK_Y$,
		such that 
		there are derivations
		\begin{equation*}
                  \scalebox{.92}{$
			\ods{\gC\coons{\gK_X\lpar\gK_Y}R}{\dD_\gG}{\gG}{\GS}
			\qomma
			\ods{\unit}{\dD_X}{\gK_X\lpar\gM_i}{\GS}
			\qomma
			\ods{\unit}{\dD_Y}{\gK_Y\lpar\gP\connn{\gM_1,\ldots,\gM_{i-1},\unit,\gM_{i+1},\ldots,\gM_n}}{\GS}
			\quand
			\ods{\unit}{\dD_C}{\gC}{\GS}
                        $}
			\hskip-1em
		\end{equation*}
		for some $i\in\set{1,\ldots,n}$.
	\end{enumerate}
	
\end{restatable}

\begin{rem}
The immediate question to address at this point is why we need two cases, where splitting lemmas in the literature require only one case.
Notice that Lemma~\ref{lem:splitting:tens} is a special case of Lemma~\ref{lem:splitting:prime}, since $\ltens$ is a prime graph, as indicated
in~\eqref{eq:par-tens} in Section~\ref{sec:modules}.
In the case of the tensor, the two sub-cases of Lemma~\ref{lem:splitting:prime} collapse;
hence for $\ltens$ the above lemma for prime graphs collapses to the form of Lemma~\ref{lem:splitting:tens}.

The existence of these two separate cases is a surprising novelty 
	since all proof systems, on which the splitting technique has been applied so far,
	only employ atoms, binary connectives~\cite{gug:SIS,tubella:phd}, modalities~\cite{dissvonlutz,SIS-V}, and quantifiers~\cite{ross:MAV1}. Here, for the first time we present splitting for $n$-ary connectives with $n>2$.
\end{rem}

When the principal graph is formed using a prime graph with at least four vertices, the two cases of Lemma~\ref{lem:splitting:tens} are properly distinct.
We illustrate first the case \ref{prime:A} of Lemma~\ref{lem:splitting:tens},
by renormalising the proof~\eqref{eg:proofA1} as an example to explain
this idea.

In the conclusion of~\eqref{eg:proofA1} we have the disjoint union of three connected
graphs. 
We can select 
$\Pfour\connn{c,a,d,b}$ (the connected component on the right)
as
the principal component, and apply Case~\ref{prime:A} of 
Lemma~\ref{lem:splitting:prime} to
reorganise the derivation in such a way 
an instance of $\pdr$
involving 
$\Pfour\connn{c,a,d,b}$
can be applied (see Example~\ref{eg:proofN}).
This allows us to conclude as in Case~\ref{prime:A} applies, since all of four graphs $\cneg{a} \lpar a$, $\cneg{b} \lpar b$, $\cneg{c} \lpar c$ and $\cneg{d} \lpar d$ are provable.
That is, we can obtain a derivation of the form
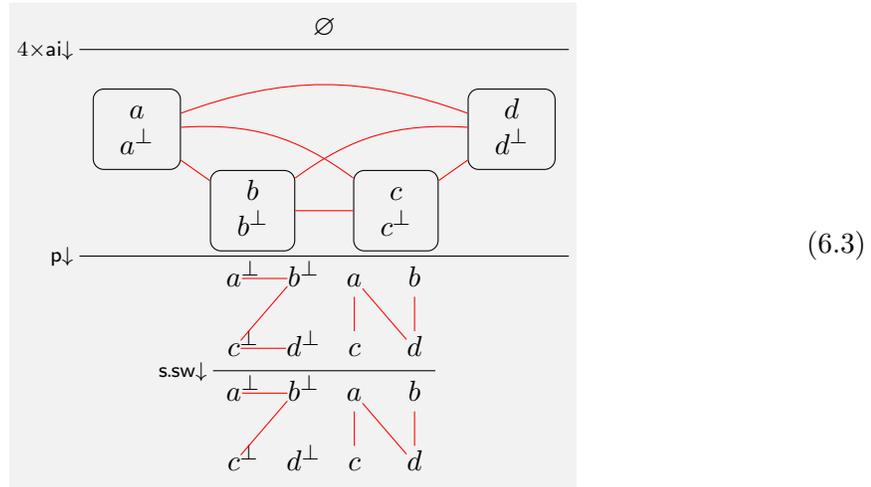
\begin{equation}\label{eg:proof2}
	\od{
		\odi{
			\odi{
				\odi{
					\odh{\unit}
				}{4\times \aidr}{
						\begin{array}{c@{\quad}c@{\qquad}c@{\quad}c}
							\\
							\vmodule1{\begin{array}{c}a\\\cneg a\end{array}}&
							&&
							\vmodule4{\begin{array}{c}d\\\cneg d\end{array}}
							\\
							&
							\vmodule2{\begin{array}{c}b\\\cneg b\end{array}}&
							\vmodule3{\begin{array}{c}c\\\cneg c\end{array}}&
							\\
						\end{array}
						\edges{M1/M2,M2/M3,M3/M4}
						\bentedges{M2/M4/20,M1/M3/20,M1/M4/20}
				}{}
			}{\pdr}{
				\begin{array}{c@{\quad\;\;}c@{\quad\;\;}c@{\quad\;\;}c}
					\vna1 & \vnb1 & \va1 & \vb1 \\
					\\
					\vnc1 & \vnd1 & \vc1 &\vd1
				\end{array}
				\edges{a1/c1,a1/d1,b1/d1,nc1/nd1,na1/nb1,nb1/nc1}
			}{}
		}{\ssdr}{
		\begin{array}{c@{\quad\;\;}c@{\quad\;\;}c@{\quad\;\;}c}
			\vna1 & \vnb1 & \va1 & \vb1 \\
			\\
			\vnc1 & \vnd1 & \vc1 &\vd1
		\end{array}
		\edges{a1/c1,a1/d1,b1/d1,na1/nb1,nb1/nc1}
	}{}}
\end{equation}

To see where Case~{\ref{prime:B}} is required in order to destroy a prime graph by removing one of its modules, consider the graph below:
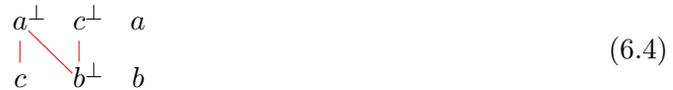
\begin{equation}\label{struct:deep2}
	\begin{array}{c@{\quad\;\;}c@{\quad\;\;}c@{\quad\;\;}c}
		\vna1 & \vnc1 &  \va1 \\
		\\[-1ex]
		\vc1 &  \vnb1 & \vb1 \\
	\end{array}
	\edges{na1/nb1,na1/c1,nc1/nb1}
\end{equation}
Here Case~\ref{prime:A} cannot be applied since when we select 
$\Pfour\connn{c,\cneg a,\cneg b,\cneg c}$
as the principal component and we try to apply
$\pdr$, then $c$ and $\cneg{c}$ can no longer communicate.  
Therefore, we must first move $a$ or $b$ into the structure and
apply an $\aidr$, in order to destroy the prime graph in such a way that $c$ and $\cneg{c}$ can still communicate.

Applying Case~{\ref{prime:B}} to the graph \ref{struct:deep2},  following  the naming convention in Lemma~\ref{lem:splitting:prime}, 
we have graphs $\gK_X = a$ and $\gK_Y = b$
such that the following two graphs are provable, as required for splitting.
	\begin{equation}
		\cneg a \qquad  a 
		\qquad\;
		\mbox{and}
		\qquad\;
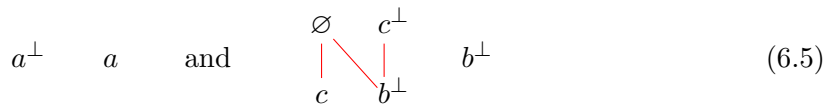
		\begin{array}{c@{\quad\;\;}c}
			\vu1 & \vnc1   \\\\
			\vc1 &  \vnb1 
		\end{array}
		\edges{c1/u1,u1/nb1,nc1/nb1}
		\qquad
		\cneg b
	\end{equation}
Then we can reassemble (using $\ssdr$)  the first steps of proof below
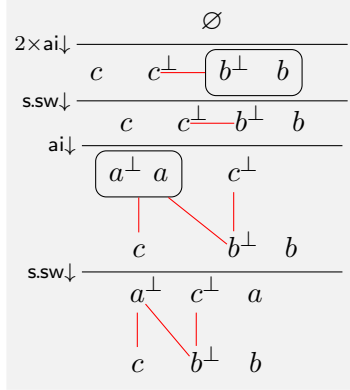
\begin{equation*}
	\od{
		\odi{
			\odi{
				\odi{
					\odi{
						\odh{\unit}
					}{2\times \aidr}{
						\begin{array}{c@{\quad\;\;}c@{\quad\;\;}c@{\quad\;\;}c}
							\vc1 & \vnc1 & \vmodule{1}{\vnb1 \quad\;\; \vb1 }
						\end{array}
						\edges{nc1/M1}
					}{}
				}{\ssdr}{
					\begin{array}{c@{\quad\;\;}c@{\quad\;\;}c@{\quad\;\;}c}
						\vc1 & \vnc1 & \vnb1 & \vb1 
					\end{array}
					\edges{nc1/nb1}
				}{}
			}{\aidr}{
				\begin{array}{c@{\quad\;\;}c@{\quad\;\;}c@{\quad\;\;}c}
					\vmodule1{\vna1 \quad \va1}& \vnc1 \\
					\\
					\vc1 &  \vnb1 & \vb1 \\
				\end{array}
				\edges{M1/nb1,M1/c1,nc1/nb1}
			}{}
		}{\ssdr}{
			\begin{array}{c@{\quad\;\;}c@{\quad\;\;}c@{\quad\;\;}c}
				\vna1 & \vnc1 &  \va1 \\
				\\
				\vc1 &  \vnb1 & \vb1 \\
			\end{array}
			\edges{na1/nb1,na1/c1,nc1/nb1}
		}{}
	}
\end{equation*}


\subsection*{Splitting for Atoms}\label{sec:splitproofs-atomic}

We also need a splitting lemma where the principal graph is just a singleton.

\begin{restatable}[Atomic Splitting]{lemma}{lemAtomSplit}\label{lem:splitting:atom}
	Let $\gG$ be a graph.
	If\/ $\proves[\GS]{G\lpar a}$~ for an atom $a$, then there is a context $\gC\coonso R$ such that 
 there are derivations
	\begin{equation*}
		\ods{\gC\coons{\cneg a}R}{\dD_G}{\gG}{\GS}
		\quand
		\ods{\unit}{\dD_C}{\gC}{\GS}
		\quadfs
	\end{equation*}
\end{restatable}
Atomic splitting simply states that for an atom $a$ in the proof, there is a matching $\cneg{a}$ somewhere else in the proof. Note that there may be more than one $\cneg{a}$, in which case the lemma will pick out exactly the instance that interacts with $a$ via an $\aidr$ instance in the original proof.

As with splitting for prime graphs the novelty compared to the literature is the presence of the context~$C$.

\subsection*{Context reduction.}
Observe that the splitting lemma is not applied to a $\lpar$-graph. The~$\lpar$ plays in $\GS$ a similar role as the comma in the sequent calculus.
Splitting lemmas such as Lemma~\ref{lem:splitting:prime} apply in a sequent-like 
context, sometimes called a \textit{shallow context}. Splitting cannot be applied to proofs where the principal graph (that is not a $\lpar$) appears at any depth in the conclusion of the proof. In order to be able to use splitting to eliminate cuts at arbitrary depth, we need to extend splitting from shallow contexts to deep contexts. For this, we need the context reduction lemma, stated below.

\def\XXX{{\mathcal X}}
\def\YYY{{\mathcal Y}}
\begin{restatable}[Context reduction]{lemma}{lemConRed}\label{lem:conred}
	Let $A$ be a graph and $\gG\coonso{S}$ be a context, such that\/
	$\proves[\GS]{\gG\coons{\gA}S}$. Then there are a
	context $\gC\coonso R$ and a graph $\gK$, such that there are derivations
	\begin{equation*}
		\ods{\gC\coons{\gK\lpar\XXX}R}{\dD_G}{\gG\coons{\XXX}S}{\GS}
		\quand
		\ods{\unit}{\dD_A}{\gK\lpar\gA}{\GS}
		\quand
		\ods{\unit}{\dD_C}{\gC}{\GS}
	\end{equation*}
	for any graph $\XXX$.
\end{restatable}

The idea of context reduction is that, if we select any module, called $\gA$ in Lemma~\ref{lem:conred} above, we can always rewrite its context, say $\gG\coons{\cdot}S$, such that it is of the form $\gK \lpar \coons{\cdot}{\emptyset}$ inside some provable context. The use of $\XXX$ reflects the fact that, in doing so, you never need to change $\gA$; hence $\gA$ could be replaced by another graph $\XXX$ and the same transformation can be applied to the context.

Similarly to splitting, the novelty compared to context reduction in the literature is that we must keep track of a provable context.
This is required to cope with graphs such as the following:
\begin{equation}\label{struct:deep}
  \begin{array}{c@{\quad\;\;}c@{\quad\;\;}c@{\quad\;\;}c@{\quad\;\;}c}
    & & \vna1 &  \\
    \\
    \va1 &\vnb1 & \vnc1 &  \vc1 &  \vb1 \\
  \end{array}
  \edges{na1/nb1,na1/nc1,na1/c1,b1/c1}
\end{equation}
Suppose that we are interested in the singleton module $\cneg{a}$ in the above graph, and we wish to apply context reduction to it.
When we apply Lemma~\ref{lem:conred}, clearly there is only one choice of graph $\gK$ such that $\gK \lpar \cneg{a}$ is provable, namely $\gK = a$.
Observe, furthermore, that in order to rewrite the rest of the above graph, so it is of the form $\gC\coons{\gK\lpar\cneg{a}}R$
such that $\gC$ is provable, the only choice is to set $\gC = \cneg{b} \lpar \cneg{c} \lpar (b \ltens c)$, as done in the following derivation:
\begin{equation}\label{proof:deep}
	\od{\odi{\odh{
			\begin{array}{cc@{\quad\;\;}c@{\quad\;\;}c@{\quad\;\;}c}
				\vmodule{1}{\va1 \qquad \vna1\;\; } \\
				\\
				\vnb1 & \vnc1 &  \vc1 &  \vb1 \\
			\end{array}
			\edges{M1/nb1,M1/nc1,M1/c1,b1/c1}
		}
	}{\ssdr}{
		\begin{array}{c@{\quad\;\;}c@{\quad\;\;}c@{\quad\;\;}c@{\quad\;\;}c}
			& & \vna1 &  \\
			\\
			\va1 &\vnb1 & \vnc1 &  \vc1 &  \vb1 \\
		\end{array}
		\edges{na1/nb1,na1/nc1,na1/c1,b1/c1}
	}{}}
\end{equation}
which can be completed to a proof, by applying $\aidr$ to the module $a \lpar \cneg{a}$ and then $\idr$ to prove the resulting graph $\gC$.
Examples such as~\eqref{eg:deep-ten}, \eqref{struct:deep2} and~\eqref{proof:deep}, where each conclusion can only be proven if rules can be applied deep inside a graph instead of to a graph at the top level of the modular decomposition of the conclusion, show that deep inference is necessary for developing a proof theory on graphs allowing $\pfour$ as an induced sub-graph.

\begin{figure}[!tb]
\[
{
\od{
  \odi{
    \odI{
      \odi{
	\odI{
   	  \odI{
                 \odh{
	\unit
		}      
	}{3\times \idr }{
\begin{array}{c@{\quad\;\;}c@{\quad\;\;}c@{\quad\;\;}c}
          \begin{array}{c}\vng1 \\\\ \vnf1 \end{array}
          & \vmodule{3}{
			 \begin{array}{c@{\quad\;\;}c@{\quad\;\;}c@{\quad\;\;}c}
				\vmodule{4}{\begin{array}{c@{\quad\;\;}c}  \va1  &  \vna1 \end{array}}   &  \vnc1   \\\\
			 	\vc1	&  \vmodule{5}{\begin{array}{c@{\quad\;\;}c}  \vb1  &  \vnb1 \end{array}}
			 \end{array}
		 }
          & \begin{array}{r}\vf1 \\\\ \vg1 \end{array}
        \end{array}
        \edges{nf1/ng1,g1/M3,nf1/M3,f1/M3,M4/M5,M4/c1,M5/nc1}  
}{}
}{2\times \ssdr }{
        \begin{array}{c@{\quad\;\;}c@{\quad\;\;}c@{\quad\;\;}c}
          \begin{array}{c}\vng1 \\\\ \vnf1 \end{array}
          & \vmodule{3}{
			 \begin{array}{c@{\quad\;\;}c@{\quad\;\;}c@{\quad\;\;}c}
				\vna1  &  \vnc1  & \va1 \\\\
			 	\vc1	&  \vnb1   & \vb1
			 \end{array}
	    }
          & \begin{array}{r}\vf1 \\\\ \vg1 \end{array}
        \end{array}
        \edges{nf1/ng1,g1/M3,nf1/M3,f1/M3,na1/nb1,na1/c1,nb1/nc1}        
	}{}
      }{\ssdr}{
        \begin{array}{c@{\quad\;\;}c@{\quad\;\;}c@{\quad\;\;}c}
          \begin{array}{c}\vng1 \\\\ \vnf1 \end{array}
          & \vmodule{3}{\begin{array}{c@{\quad\;\;}c}
                         \vna1  &  \vnc1 \\\\
			 \vc1	&  \vnb1
                         \end{array}}
          & \begin{array}{r}\vf1 \\\\ \vg1 \end{array}
          &
          \begin{array}{c}
           \va1 \\\\
           \vb1
          \end{array}
        \end{array}
      \edges{nf1/ng1,g1/M3,nf1/M3,f1/M3,na1/nb1,na1/c1,nb1/nc1}
      }{}
      }{4\times \aidr}{
\begin{array}{c@{\quad\;\;}c}
        \vmodule{2}{\begin{array}{c@{\quad\;\;}c@{\quad\;\;}c@{\quad\;\;}c}
     \begin{array}{c}\vng1 \\\\ \vnf1 \end{array}
     & \vmodule{3}{\begin{array}{c@{\quad\;\;}c}
                         \vna1  &  \vnc1 \\\\
			 \vc1	&  \vnb1
                         \end{array}}
     &
     \begin{array}{c@{\quad\;\;}c}\vf1 & \va1 \\\\ \vg1 & \vb1 \end{array}
    \end{array}}
&
\begin{array}{c@{\quad\;\;}c@{\quad\;\;}c@{\quad\;\;}c@{\quad\;\;}c@{\quad\;\;}c}
		\vmodule{4}{\begin{array}{c@{\quad\;\;}c}  \ve1  &  \vne1 \end{array}}
         \\&   \vmodule{6}{\begin{array}{c@{\quad\;\;}c}  \vh1  &  \vnh1 \end{array}} \\
            \vmodule{5}{\begin{array}{c@{\quad\;\;}c}  \vd1  &  \vnd1 \end{array}} 
        \end{array}
\end{array}
      \edges{nf1/ng1,g1/M3,nf1/M3,f1/M3,M2/M4,M2/M5,M2/M6,M4/M5,M4/M6,M5/M6,na1/nb1,na1/c1,nb1/nc1}
      }{}
    }{\pdr}{
    \begin{array}{c@{\quad\;\;}c}
        \vmodule{2}{\begin{array}{c@{\quad\;\;}c@{\quad\;\;}r}
     \begin{array}{c}\vng1 \\\\ \vnf1 \end{array}
     & \vmodule{3}{\begin{array}{c@{\quad\;\;}c}
                         \vna1  &  \vnc1 \\\\
			 \vc1	&  \vnb1
                         \end{array}}
     &
     \begin{array}{r}\vf1 \\\\ \vg1 \end{array}
    \end{array}}
&
    \begin{array}{cc@{\quad\;\;}c@{\quad\;\;}c@{\quad\;\;}c@{\quad\;\;}c@{\quad\;\;}c@{\quad\;\;}c@{\quad\;\;}c@{\quad\;\;}c}
                   &  \vne1 & \vnh1  &            \ve1 &        \vh1  \\\\
                                                    &     &  \vnd1    &   
	\vmodule{1}{\begin{array}{c@{\quad\;\;}c}  \va1  &  \vb1 \end{array}}  & \vd1 
    \end{array}
    \end{array}
    \edges{M1/h1,M1/d1,nh1/nd1,nd1/ne1,e1/h1,nf1/ng1,g1/M3,ne1/M2,nf1/M3,f1/M3,na1/nb1,na1/c1,nb1/nc1} 
    }{}
}
}
\]
\caption{A proof used to illustrate context reduction and splitting. It also shows a non-trivial use of the $\pdr$-rule that cannot be achieved using the $\idr$-rule.
}\label{fig:eg}
\end{figure}

For a more substantial example consider the derivation in Figure~\ref{fig:eg}.
Assume, we aim to apply splitting to the $N$-shaped induced subgraph in the conclusion consisting of vertices labelled $c$, $\cneg{a}$, $\cneg{b}$ and $\cneg{c}$.
Clearly, the splitting lemma cannot be directly applied to that induced subgraph,
so firstly we reduce the context.
The three bottommost inference steps shown in Figure~\ref{fig:eg} correspond to the steps taken by context reduction that do not touch the $N$-shaped subgraph but bring the vertices labelled $a$ and $b$ next to our subgraph.
At that point Case~\ref{prime:B} of splitting for prime graphs can be applied to our $N$-shaped subgraph, taking into account that, together with the vertices labelled $a$ and $b$, they form a sequent-like module of the graph at that point in the derivation. At this point, the splitting lemma provides us with the information needed to complete the proof.

An additional reason why Figure~\ref{fig:eg} is interesting is that it is an example of a proof where the effect of the $\pdr$-rule cannot be simulated by using the $\idr$-rule.
Therefore, this example emphasises the need for $\pdr$ or at least a rule to the effect of $\pdr$ in an analytic proof system on graphs.

\subsection*{Proving splitting and context reduction.}\label{sec:splitproofs-prime}

The detailed proofs of these lemmas can be found in Appendix~\ref{sec:splittingproofs}. Here we present an outline and the basic ideas.
As standard, we proceed by exhausting all possible permutations of rules where the main points of interest are where the rules change the principal connective in some way.
The novelty compared to results in the literature on deep inference is that, due to the important role of contexts we prove splitting and context reduction together by  mutual induction, using the size $\gsize{\gG}$ of a graph (see Definition~\ref{def:size}) as induction measure.

During the proof of splitting and context reduction we must handle cases involving the $\pdr$ rule. 
Observe that when the $\pdr$ rule is applied we obtain at least four graphs connected by tensors.
Since we will need to handle frequently such cases, we find it convenient to state splitting for such multi-tensors as a lemma.
In what follows we generalise the Lemma~\ref{lem:splitting:tens} to the $n$-ary tensor for any $n\geq 2$ (or \emph{multi-tensor}), which is clearly a consequence of splitting for tensor (which, in turn, we have discussed is a special case of splitting for prime graphs).

\begin{restatable}[Splitting Multi-tensor]{lemma}{lemSplitMulti}\label{lem:splitting:multitens}
	Let $\gG$ be a graph,  and $\gA_1,\dots, \gA_n$ be non-empty graphs.
	If $\proves[\GS]{\gG\lpar ({\gA_1\ltens \cdots \ltens \gA_n})}$, 
	then there is a context $\gC\coonso R$ and graphs $\gK_1$, \dots , $\gK_n$, 
	such that  
	there are derivations
	\begin{equation*}
		\ods{\gC\coons{{\gK_1\lpar \cdots\lpar \gK_n}}R}{\dD_\gG}{\gG}{\GS}
		\qomma
		\ods{\unit}{\dD_i}{\gK_i\lpar\gA_i}{\GS}
		\quand
		\ods{\unit}{\dD_C}{\gC}{\GS}
		\hskip-1em
	\end{equation*}
	for all $i\in\set{1,\ldots,n}$.
\end{restatable}

All lemmas of this section are proved simultaneously, by case analysis
and an inductive argument. Before we go into the details of the
case analysis, let us draw the attention to
Figure~\ref{fig:complete-roadmap}, which shows in more detail the
dependencies of the proofs of all lemmas in this section.
The main lemma is Lemma~\ref{lem:splitting:prime}. Then
Lemma~\ref{lem:splitting:tens} is just a special case of
Lemma~\ref{lem:splitting:prime}, and
Lemma~\ref{lem:splitting:multitens} follows immediately by iterating
Lemma~\ref{lem:splitting:tens}.
To ensure
that there is no circular reasoning in the other dependencies, we argue that whenever the proof
of one lemma refers to another then this happens with respect to a
strictly smaller graph (according to Definition~\ref{def:size}).
Let us now look more closely at the case analysis for Lemma~\ref{lem:splitting:prime}.

\begin{figure}[t]
	\small
	\begin{tikzpicture}
		\node[draw, ellipse] (prime)  at (0,0)   
		{\begin{tabular}{c}Splitting Prime\\ Lemma~\ref{lem:splitting:prime}\end{tabular}};
		\node[draw, ellipse] (tensor) at (5,-2)   
		{\begin{tabular}{c}Splitting Tensor\\ Lemma~\ref{lem:splitting:tens}\end{tabular}};
		\node[draw, ellipse] (multi)  at (4,-5) 
		{\begin{tabular}{c}Splitting Multi-Tensor\\ Lemma~\ref{lem:splitting:multitens}\end{tabular}};
		\node[draw, ellipse] (atom)  at (-3, -5) 	
		{\begin{tabular}{c}Atomic Splitting\\ Lemma~\ref{lem:splitting:atom}\end{tabular}};
		\node[draw, ellipse] (context)  at (-5, -2) 
		{\begin{tabular}{c}Context Reduction\\ Lemma~\ref{lem:conred}\end{tabular}};
		
		\draw [-latex, thick] (prime) to[out=130, in=160, looseness=4 ] node[midway,above]{(a) and (b.I) } (prime)  ;
		\draw [>=latex,->>, thick] (prime) to[out=50, in=80, looseness=5] node[midway,above]{(b.III)} (prime);
		\draw [-latex, thick] (prime) to  node[midway,above]{(b.II)} (tensor);
		\draw [-latex, thick] (prime) to [out=170, in=90] node[midway,above]{(c) } (context)  ;
		\draw [-latex, thick] (prime) to [out=170, in=200, looseness=4] node[midway,left]{(c) } (prime)  ;
		\draw [-latex, thick] (prime) to  node[midway,right]{(d)} (multi)  ;
		\draw [-latex, thick] (prime) to[out=-70, in=160] node[midway,left]{(e)} (multi)  ;
		\draw [>=latex,->>, thick] (prime) to[out=-70, in=-100,looseness=10] node[midway,below]{(e)} (prime)  ;
		
		\draw [-latex, thick] (atom) to [out=160, in=200,looseness=5] node[midway,left]{(a)} (atom)  ;
		\draw [-latex, thick] (atom) to [out=120, in=-60] node[midway,left]{(b)} (context)  ;
		\draw [-latex, thick] (atom) to [out=120, in=60,looseness=5] node[midway,above]{(b)} (atom)  ;
		\draw [-latex, thick] (atom) to  node[midway,above]{(c)} (multi)  ;
		
		\draw [-latex, thick] (tensor) to [bend right] node[midway,right]{special case of} (prime)  ;
		\draw [-latex, thick] (multi) to node[midway,right]{iterating} (tensor)  ;
		
		\draw [-latex, thick] (context) to [bend right] node[midway,right]{} (prime)  ;
	\end{tikzpicture}
	\caption{Complete roadmap of splitting and context reduction lemmas proofs. Double head arrows make use of Lemma~\ref{lem:g}. 
	}
	\label{fig:complete-roadmap}
\end{figure}

\subsection*{Case Analysis for the proof of Lemma~\ref{lem:splitting:prime} (Splitting Prime)}
There are five cases to consider which we enumerate below.
	\begin{enumerate}[(a)]
		\item 
		The last rule in the given proof acts inside $\gG$ or any $\gM_i$ with $i\in \set{1,\dots, n}$.
        This case is relatively simple, since the structure of the conclusion does not change.
		
		\item 
		The last rule in the given proof is a $\ssdr$ such that another graph, external to the principal prime graph, moves inside the principal prime graph, i.e.,
		$\gG=\gG'\lpar\gG''$ with $\gG'\neq \unit$ and $\dD$ is of shape
		\begin{equation*}
			\scalebox{.95}{%
				\odn{
					\ods{\unit}{\dD''}{\gG''\lpar \gP\connn{\gM_1,\ldots,\gM_n}\coons{\gG'}{R_P}}{\GS}}
				{\ssdr}
					{\gG''\lpar\gG'\lpar\gP\connn{\gM_1,\ldots,\gM_n}}{}
			}\hskip-3em
		\end{equation*}
		This can result in several scenarios, that are named
                (b.I), (b.II), and (b.III) in the appendix, and that
                are not all immediately obvious.
		
		\item The last rule in the given proof is a $\ssdr$ that moves the whole prime graph inside the rest of the graph. 
		In this case  
		$\dD$ is of shape
		\begin{equation*}
			\odn{\ods{\unit}{
					\dD'}{
					\gG\coons{\gP\connn{\gM_1,\ldots,\gM_n}}{S}}{\GS}}{
				\ssdr}{
				\gG \lpar\gP\connn{\gM_1,\ldots,\gM_n}}{}
		\end{equation*}
		We will explain that this case appeals to context reduction.

		\item
		The last rule in the given proof is a $\pdr$ directly applied to the principal connective,
		i.e., $\gG=\gG'\lpar\cneg\gP\connn{\gN_1,\ldots,\gN_n}$ 
		and 
		$\dD$ is of shape
		\begin{equation*}
			\odn{\ods{\unit}{\dD'}{
					[\gG'';([N_1;M_1];\cdots;[N_n;M_n])]}{\GS}}{
				\pdr}{
				\gG''\lpar\cneg\gP\connn{\gN_1,\ldots,\gN_n}\lpar\gP\connn{\gM_1,\ldots,\gM_n}}{}
			\;.\hskip-3em
		\end{equation*}
		In this case we have to refer to the multi-tensor splitting lemma.
		\item 
		The last rule is a $\pdr$ such that not all components of the principal connective are non-empty, and we have
		$\gG=\gG''\lpar\gQ\connn{\gN_1,\ldots,\gN_k}$ where $\gN_1, \dots, \gN_k$ are non-empty graphs, 
		$\gQ$ is a prime graph with $\sizeof\vQ > \sizeof\vP$ 
		and such that 
		w.l.o.g.
		$\gP\connn{\gM_1,\ldots,\gM_n} = \cneg{\gQ}\connn{ \unit, \gL_2, \ldots \gL_k }$ for some (possibly empty) graphs $\gL_2, \ldots \gL_k$, 
		and
		$\dD$ is of shape
		\begin{equation*}
			\small
			\odn{\ods{\unit}{
					\dD'}{
					[\gG'';(\gN_1 ;[\gN_2;\gL_2];\ldots;[\gN_k;\gL_k])]
				}{\GS}}{
				\pdr}{
				[\gG'';\gQ\connn{\gN_1,\ldots,\gN_k}; \gP\connn{\gM_1,\ldots,\gM_n}]
			}{}
			\hskip-2em       
		\end{equation*}
		In this case it is important to ensure the side-conditions on $\pdr$ are respected.
	\end{enumerate}

        \medskip

All technical details of this case analysis are in Appendix~\ref{sec:splittingproofs}. In the remainder of this section, we give some informal explanation why the cases in the proof of Lemma~\ref{lem:splitting:prime} can become complex. For this, consider the following example.
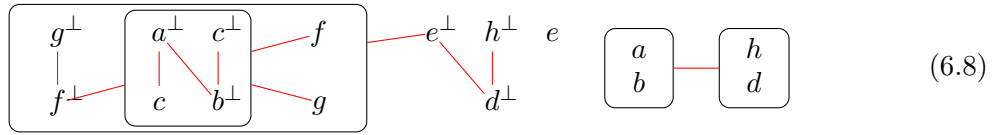
\begin{equation}\label{eg:exotic}
	\begin{array}{c@{\quad\;\;}c@{\quad\;\;}c@{\quad\;\;}c}
		\vmodule{2}{\begin{array}{c@{\quad\;\;}c@{\quad\;\;}r}
				\begin{array}{c}\vng1 \\\\ \vnf1 \end{array}
				& \vmodule{3}{\begin{array}{c@{\quad\;\;}c}
						\vna1  &  \vnc1 \\\\
						\vc1	&  \vnb1
				\end{array}}
				&
				\begin{array}{r}\vf1 \\\\ \vg1 \end{array}
		\end{array}}
		&
		\begin{array}{c@{\quad\;\;}c@{\quad\;\;}c@{\quad\;\;}c@{\quad\;\;}c@{\quad\;\;}c@{\quad\;\;}c@{\quad\;\;}c@{\quad\;\;}c}
			\vne1 & \vnh1  &            \ve1   \\\\
			&  \vnd1    
		\end{array}
		\vmodule{7}{\begin{array}{c} \va1 \\ \vb1 \end{array}}
		&
		\vmodule{8}{\begin{array}{c} \vh1 \\ \vd1 \end{array}}
	\end{array}
	\edges{M7/M8,nh1/nd1,nd1/ne1,nf1/ng1,g1/M3,ne1/M2,nf1/M3,f1/M3,na1/nb1,na1/c1,nb1/nc1}
\end{equation}
One possible way to prove the above graph is to first apply the $\ssdr$-rule to move the vertex labelled $e$ so that it is attached by an edge only to~$h$. This way we obtain exactly the graph in the conclusion of Figure~\ref{fig:eg}, therefore we know already how to complete the proof via that strategy.

Now observe that the above graph~(\ref{eg:exotic}) is of the form $\gG \lpar (( a \lpar b) \ltens ( h \lpar d ))$. Thus it is in a form to which we could apply tensor splitting (Lemma~\ref{lem:splitting:tens}) to this $\ltens$-operator as the principal graph. This  forces the proof to be renormalised such that firstly $e$ moves next to $\cneg{e}$ as shown in the derivation below.
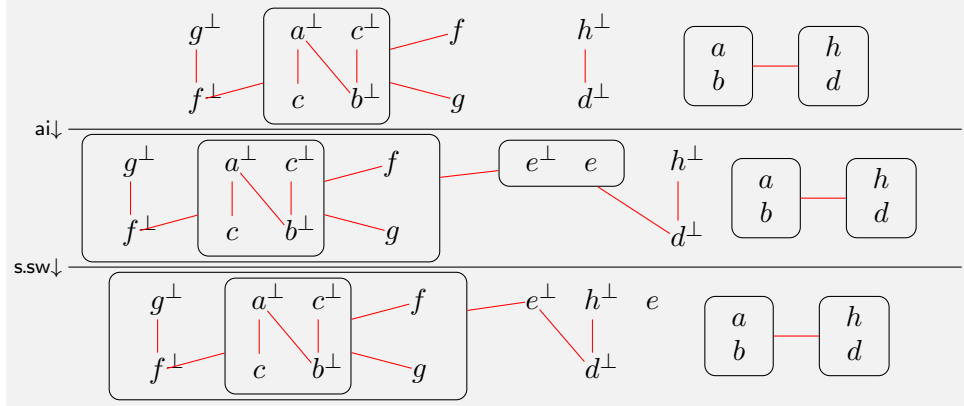
\[
\od{
\odi{
\odi{
 \odh{
    \begin{array}{c@{\quad\;\;}c@{\quad\;\;}c@{\quad\;\;}c}
	\begin{array}{c@{\quad\;\;}c@{\quad\;\;}r}
     \begin{array}{c}\vng1 \\\\ \vnf1 \end{array}
     & \vmodule{3}{\begin{array}{c@{\quad\;\;}c}
                         \vna1  &  \vnc1 \\\\
			 \vc1	&  \vnb1
                         \end{array}}
     &
     \begin{array}{r}\vf1 \\\\ \vg1 \end{array}
    \end{array}
&
    \begin{array}{cc@{\quad\;\;}c@{\quad\;\;}c@{\quad\;\;}c@{\quad\;\;}c@{\quad\;\;}c@{\quad\;\;}c@{\quad\;\;}c@{\quad\;\;}c}
                   &  
                     \vnh1  \\\\
                                                        &  \vnd1  
    \end{array}
&
\vmodule{7}{\begin{array}{c} \va1 \\ \vb1 \end{array}}
&
\vmodule{8}{\begin{array}{c} \vh1 \\ \vd1 \end{array}}
    \end{array}
    \edges{M7/M8,nh1/nd1,nf1/ng1,g1/M3,nf1/M3,f1/M3,na1/nb1,na1/c1,nb1/nc1}
 }
}{\aidr}{
    \begin{array}{c@{\quad\;\;}c@{\quad\;\;}c@{\quad\;\;}c}
        \vmodule{2}{\begin{array}{c@{\quad\;\;}c@{\quad\;\;}r}
     \begin{array}{c}\vng1 \\\\ \vnf1 \end{array}
     & \vmodule{3}{\begin{array}{c@{\quad\;\;}c}
                         \vna1  &  \vnc1 \\\\
			 \vc1	&  \vnb1
                         \end{array}}
     &
     \begin{array}{r}\vf1 \\\\ \vg1 \end{array}
    \end{array}}
&
    \begin{array}{c@{\quad\;\;}c@{\quad\;\;}c@{\quad\;\;}c@{\quad\;\;}c@{\quad\;\;}c@{\quad\;\;}c@{\quad\;\;}c@{\quad\;\;}c@{\quad\;\;}c}
                     \vmodule{5}{\begin{array}{c@{\quad\;\;}c} \vne1 & \ve1 \end{array}}
                    & \vnh1   \\\\
                                                        &  \vnd1    
    \end{array}
\vmodule{7}{\begin{array}{c} \va1 \\ \vb1 \end{array}}
&
\vmodule{8}{\begin{array}{c} \vh1 \\ \vd1 \end{array}}
    \end{array}
    \edges{M7/M8,nh1/nd1,nd1/M5,nf1/ng1,g1/M3,M5/M2,nf1/M3,f1/M3,na1/nb1,na1/c1,nb1/nc1}
 }{}
}{\ssdr}{
    \begin{array}{c@{\quad\;\;}c@{\quad\;\;}c@{\quad\;\;}c}
        \vmodule{2}{\begin{array}{c@{\quad\;\;}c@{\quad\;\;}r}
     \begin{array}{c}\vng1 \\\\ \vnf1 \end{array}
     & \vmodule{3}{\begin{array}{c@{\quad\;\;}c}
                         \vna1  &  \vnc1 \\\\
			 \vc1	&  \vnb1
                         \end{array}}
     &
     \begin{array}{r}\vf1 \\\\ \vg1 \end{array}
    \end{array}}
&
    \begin{array}{c@{\quad\;\;}c@{\quad\;\;}c@{\quad\;\;}c@{\quad\;\;}c@{\quad\;\;}c@{\quad\;\;}c@{\quad\;\;}c@{\quad\;\;}c}
                     \vne1 & \vnh1  &            \ve1   \\\\
                                                      &  \vnd1    
    \end{array}
\vmodule{7}{\begin{array}{c} \va1 \\ \vb1 \end{array}}
&
\vmodule{8}{\begin{array}{c} \vh1 \\ \vd1 \end{array}}
    \end{array}
    \edges{M7/M8,nh1/nd1,nd1/ne1,nf1/ng1,g1/M3,ne1/M2,nf1/M3,f1/M3,na1/nb1,na1/c1,nb1/nc1}
}{}
}
\]
Observe that in the topmost graph in the derivation above we have broken the graph on the left into two disjoint graphs. The first is provable when composed with $a \lpar b$, while the second is provable when composed with $d \lpar h$.

The choices between the different rules to first apply to graph~(\ref{eg:exotic}), where $\ssdr$ is applied to $e$ in order to either create a larger prime graph on the right (leading to the proof in Figure~\ref{fig:eg}) or to break up the structure on the left as shown above, lead to quite different proofs. Moving from the former to the latter proof is discussed in Case~\ref{split:b.III} in the proof of Lemma~\ref{lem:splitting:prime}, which is in fact the most involved sub-case.

Case~\ref{split:b.II} in the proof of Lemma~\ref{lem:splitting:prime} is not immediately obvious but is nonetheless usually present in standard splitting proofs on formulas, where similar cases are induced by associativity. A typical example is given by the following two isomorphic graphs, that are composed differently:
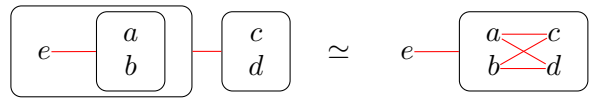
\begin{equation*}
  \label{eq:iso-exa}
  \vmodule{3}{
    \begin{array}{c@{\quad\;\;}c}
      \ve1
      &
      \vmodule{1}{
        \begin{array}{c}
          \va1 \\ \vb1
        \end{array}
      }
    \end{array}
  }
  \quad
  \vmodule{2}{
    \begin{array}{c}
      \vc1 \\ \vd1
    \end{array}
  }
  \edges{e1/M1,M3/M2}
  \quad
  \simeq
  \quad
  \begin{array}{c@{\quad\;\;}c}
    \ve1
    &
    \vmodule{1}{
      \begin{array}{c@{\quad\;\;}c}
        \va1 & \vc1  \\ \vb1 &  \vd1
      \end{array}
    }
  \end{array}
  \edges{e1/M1,a1/d1,b1/d1,a1/c1,b1/c1}
\end{equation*}
This isomorphism is used in the proof below in order to enable the $\ssdr$-rule:
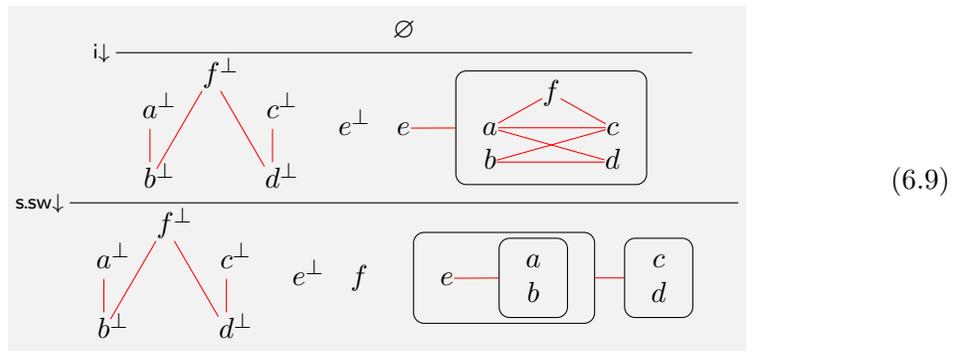
\begin{equation}\label{eg:assoc}
\od{
\odi{
\odi{
 \odh{\unit}
}{\idr}{
\begin{array}{c@{\quad\;\;}c@{\quad\;\;}c@{\quad\;\;}c@{\quad\;\;}c@{\quad\;\;}c@{\quad\;\;}c}
\begin{array}{c@{\quad\;\;}c@{\quad\;\;}c}
& \vnf1
\\
 \vna1 & & \vnc1  \\\\ \vnb1 & & \vnd1
\end{array}
& 
\vne1
&
\ve1
&
\vmodule{1}{
\begin{array}{c@{\quad\;\;}c@{\quad\;\;}c}
& \vf1
\\
 \va1  & & \vc1  \\ \vb1 &  & \vd1
\end{array}
}
\end{array}
\edges{e1/M1,nf1/nb1,nf1/nd1,na1/nb1,nc1/nd1,a1/d1,b1/d1,a1/c1,b1/c1,f1/a1,f1/c1}
}{}
}{\ssdr}{
\begin{array}{c@{\quad\;\;}c@{\quad\;\;}c@{\quad\;\;}c@{\quad\;\;}c}
\begin{array}{c@{\quad\;\;}c@{\quad\;\;}c}
& \vnf1
\\
 \vna1 & & \vnc1  \\\\ \vnb1 & & \vnd1
\end{array}
& 
\vne1
&
\vf1
&
\vmodule{3}{
\begin{array}{c@{\quad\;\;}c}
\ve1
&
\vmodule{1}{
\begin{array}{c}
 \va1 \\ \vb1
\end{array}
}
\end{array}
}
\quad
\vmodule{2}{
\begin{array}{c}
 \vc1 \\ \vd1
\end{array}
}
\end{array}
\edges{e1/M1,M3/M2,nf1/nb1,nf1/nd1,na1/nb1,nc1/nd1}
}{}
}
\end{equation}
We can apply splitting to the above graph in various ways.
One possibility is to apply splitting to the tensor on the right as shown in the conclusion of the above proof.
For that we appeal to Case~\ref{split:b.II} in the proof of Lemma~\ref{lem:splitting:prime}, which firstly observes that there is a proof of $\cneg{e} \lpar e$ 
and a proof of the following induced sub-graph.
\begin{equation}\label{eg:house}
\begin{array}{c@{\quad\;\;}c@{\quad\;\;}c@{\quad\;\;}c@{\quad\;\;}c}
\begin{array}{c@{\quad\;\;}c@{\quad\;\;}c}
 \vna1 & \vnf1 & \vnc1  \\\\ \vnb1 & & \vnd1
\end{array}
& 
\vf1
&
\vmodule{1}{
\begin{array}{c}
 \va1 \\ \vb1
\end{array}
}
\quad
\vmodule{2}{
\begin{array}{c}
 \vc1 \\ \vd1
\end{array}
}
\end{array}
\edges{M1/M2,nf1/nb1,nf1/nd1,na1/nb1,nc1/nd1}
\end{equation}
Case~\ref{split:b.III} is then used to renormalise the proof of~\eqref{eg:house} such that the $P_5$ on the left is split into two parts, from which we can reassemble a proof of the graph in the conclusion of~\eqref{eg:assoc} which satisfies the conditions of splitting for the tensor that we selected.
In particular, if $\gA = e \ltens (a \lpar b)$ and $\gB = c \lpar d$ in Lemma~\ref{lem:splitting:tens}, then we obtain $\gK_A = (\cneg{a} \ltens \cneg{b}) \lpar \cneg{e}$ and $\gK_B = \cneg{c} \lpar \cneg{d}$, as required (the context is empty for this example).

Beyond the above two illustrative examples of non-trivial normalisation steps, observe that the deepest nesting in the proof of splitting for prime graphs is inside Case~\ref{split:b.III}, which has two more levels of nested case analysis to cover all ways in which we can handle situations where the principal graph is turned into a larger prime graph using the $\ssdr$ rule.

\medskip


\section{Elimination of the Up-Fragment}\label{sec:upfrag}

We do now have all ingredients to prove the cut elimination result, i.e., the rule $\iur$ is admissible for $\GS$ (Theorem~\ref{thm:cut}). Recall that this is equivalent to the admissibility of the up-rules $\aiur$, $\ssur$, and $\pur$ (Theorem~\ref{thm:up}), and that this entails the transitivity of the consequence relation (Lemma~\ref{lem:trans}) and consistency of the proof system (Corollary~\ref{cor:consistency}).

In this section we are going to show how  splitting and context reduction
are employed to prove this result.
The procedure is similar to ordinary deep inference systems (see, e.g.,
\cite{dissvonlutz,SIS-V,ross:MAV1,CGS:foccos}). 

\begin{thm}\label{thm:ai-up}
  The rule $\aiur$ is admissible for $\GS$.
\end{thm}

\begin{proof}
  Assume we have a proof of $\gG\coons{a\ltens\cneg a}S$. By
  Lemma~\ref{lem:conred} we have a graph $\gL$ and a context
  $\gC_1\coonso{R_1}$, such that there are derivations
  \begin{equation*}
    \ods{\unit}{\dD_1}{\gC_1}{\GS}
    \qomma
    \ods{\unit}{\dD_2}{\gL\lpar(a\ltens\cneg a)}{\GS}
    \quand    
    \ods{\gC_1\coons{\gL\lpar\XXX}{R_1}}{\dD_3}{\gG\coons{\XXX}S}{\GS}
  \end{equation*}
   for any graph $\XXX$. We apply Lemma~\ref{lem:splitting:tens} to $\gL\lpar(a\ltens\cneg a)$
   and get $\gK_a$ and $\gK_\cna$ and a context
   $\gC_2\coonso{R_2}$ such that
   \begin{equation*}
     \ods{\gC_2\coons{\gK_a\lpar\gK_\cna}{R_2}}{\dD_4}{\gL}{\GS}
     \qomma
     \ods{\unit}{\dD_6}{\gK_a\lpar a}{\GS}
     \qomma
     \ods{\unit}{\dD_7}{\gK_\cna\lpar\cna}{\GS}
          \quand
     \ods{\unit}{\dD_5}{\gC_2}{\GS}
   \end{equation*}
   Applying Lemma~\ref{lem:splitting:atom} to $\gK_a\lpar a$ and $\gK_\cna\lpar\cna$ gives us
   $\gC_3\coonso{R_3}$ and $\gC_4\coonso{R_4}$ such that
   \begin{equation*}
     \ods{\gC_3\coons{\cna}{R_3}}{\dD_8}{\gK_a}{\GS}
     \qomma
     \ods{\unit}{\dD_9}{\gC_3}{\GS}
     \qomma
     \ods{\gC_4\coons{a}{R_4}}{\dD_{10}}{\gK_\cna}{\GS}
     \quand
     \ods{\unit}{\dD_{11}}{\gC_4}{\GS}
   \end{equation*}
   We can now give the following derivation
   \begin{equation*}
   	\scalebox{.9}
   	{
     \od{\odd{\odd{\odh{\unit}}{\dD_1}{\gC_1\Coons{
         \od{\odd{\odd{\odh{\unit}}{\dD_5}{
             \gC_2\Coons{
               \od{\odi{\odd{\odh{\unit}}{\dD_9}{
                     \gC_3\Coons{
                       \od{\odi{\odd{\odh{\unit}}{\dD_{11}}{\gC_4\Coons{
                               \odn{\unit}{\aidr}{\cna\lpar a}{}
                             }{R_4}}{}}{
                           \ssdr}{
                           \cna\lpar{\gC_4\coons{a}{R_4}}}{}}
                     }{R_3}}{}}{
                   \ssdr}{
                   \ods{\gC_3\coons{\cna}{R_3}}{\dD_8}{\gK_a}{}
                 \lpar
                 \ods{\gC_4\coons{a}{R_4}}{\dD_{10}}{\gK_\cna}{}}{}}
             }{R_2}}{}}{\dD_4}{\gL}{}}
         }{R_1}}{}}{\dD_3}{\gG\coons{\unit}S}{}}
     }
   \end{equation*}
   that proves $\gG=\gG\coons{\unit}S$ in $\GS$.
\end{proof}

\begin{thm}\label{thm:ss-up}
  The rule $\ssur$ is admissible for $\GS$.
\end{thm}

\begin{proof}
  Assume we have a proof of $\gG\coons{\gB\ltens\gA}S$ in $\GS$. By
  Lemma~\ref{lem:conred} we have a graph $\gL$ and a context
  $\gC_1\coonso{R_1}$, such that there are derivations
  \begin{equation*}
    \ods{\unit}{\dD_1}{\gC_1}{\GS}
    \qomma
     \ods{\unit}{\dD_2}{\gL\lpar(B\ltens\gA)}{\GS}
    \quand     
    \ods{\gC_1\coons{\gL\lpar\XXX}{R_1}}{\dD_3}{\gG\coons{\XXX}S}{\GS}
  \end{equation*}
   for any graph $\XXX$. We apply Lemma~\ref{lem:splitting:tens}
    to $\gL\lpar(B\ltens\gA)$
   and get $\gK_B$ and $\gK_A$ and a context
   $\gC_2\coonso{R_2}$ such that
   \begin{equation*}
     \ods{\gC_2\coons{\gK_B\lpar\gK_A}{R_2}}{\dD_4}{\gL}{\GS}
     \qomma
     \ods{\unit}{\dD_6}{\gK_B\lpar B}{\GS}
     \qomma
     \ods{\unit}{\dD_7}{\gK_A\lpar A}{\GS}
      \quand
     \ods{\unit}{\dD_5}{\gC_2}{\GS}
     \;.
   \end{equation*}
   We can now give a proof of $\gG\coons{\gB\coons{\gA}T}S$ as follows:
   \begin{equation*}
   	\scalebox{.9}
   	{
     \od{\odd{\odd{\odh{\unit}}{
           \dD_1}{\gC_1\Coons{
             \odn{\ods{\unit}{
                 \dD_5}{\gC_2\Coons{
                   \ods{\unit}{
                     \dD_6}{\gK_B\lpar
                     \odn{\gB\Coons{\ods{\unit}{\dD_7}{\gK_A\lpar A}{}}{T}}{
                       \ssdr}{
                       \gK_A\lpar\gB\coons{\gA}T}{}}{}}{R_2}}{}}{
               \ssdr}{
               \ods{\gC_2\coons{\gK_B\lpar\gK_A}{R_2}}{\dD_4}{\gL}{}
               \lpar\gB\coons{\gA}T}{}
           }{R_1}}{}}{
         \dD_3}{\gG\coons{\gB\coons{\gA}T}S}{}}
     }
   \end{equation*}
for any
   $\unit\subseteq T\subset\sizeof{\vB}$.
\end{proof}

\begin{thm}\label{thm:p-up}
  The rule $\pur$ is admissible for $\GS$.
\end{thm}

\begin{proof}
This proof is slightly different from the previous two and from what usually happens in a deep inference setting. In particular, we
need to invoke an induction on the ``size of the cut formula'', which in our case is the size of the module $P\connn{M_1,\ldots,M_n}\ltens\cneg	P\connn{N_1,\ldots,N_n}$ in the premise of the rule.
In other cut elimination proofs in deep inference, there is no
need for such an induction, as it is outsourced to the splitting
lemma.

Now, assume we have a proof of
	$\gG\coons{P\connn{M_1,\ldots,M_n}\ltens\cneg
		P\connn{N_1,\ldots,N_n}}S$ with $M_1, \dots, M_n $ non-empty graphs. 
	We apply Lemma~\ref{lem:conred} and
	get a graph $\gL$ and a context $\gC_1\coonso{R_1}$, such that there
	are derivations
	\begin{equation*}
		\ods{\unit}{\dD_1}{\gC_1}{\GS}
		\qomma
		\ods{\unit}{\dD_2}{\gL\lpar(P\connn{M_1,\ldots,M_n}\ltens\cneg P\connn{N_1,\ldots,N_n})}{\GS}
		\quand
		\ods{\gC_1\coons{\gL\lpar\XXX}{R_1}}{\dD_3}{\gG\coons{\XXX}S}{\GS}
	\end{equation*}
	for any graph $\XXX$. 
	
	\begin{itemize}
        \item First, consider the case where $\cneg P\connn{N_1,\ldots,N_n}\neq \unit$. We can apply Lemma~\ref{lem:splitting:tens}
		to $\gL\lpar(P\connn{M_1,\ldots,M_n}\ltens\cneg P\connn{N_1,\ldots,N_n})$
		and get graphs $\gL_P$ and $\gL_\cnP$ and a context
		$\gC_2\coonso{R_2}$ such that
		\begin{equation*}
			\ods{\gC_2\coons{\gL_P\lpar\gL_\cnP}{R_2}}{\dD_4}{\gL}{\GS}
			\qomma
			\ods{\unit}{\dD_5}{\gL_P\lpar P\connn{M_1,\ldots,M_n}}{\GS}
			\qomma
			\ods{\unit}{\dD_6}{\gL_\cnP\lpar\cneg P\connn{N_1,\ldots,N_n}}{\GS}
			\quand
			\ods{\unit}{\dD_7}{\gC_2}{\GS}
			\;.
		\end{equation*}

		Applying Lemma~\ref{lem:splitting:prime}
		to $\gL_P\lpar P\connn{M_1,\ldots,M_n}$  gives us two different cases.

		\begin{enumerate}[(A)]
			\item\label{p-up:A} 
			There are graphs $K_{1},\ldots,K_{n}$ 
			and a context $\gC_3\coonso{R_3}$ 
			such that
			\begin{equation*}
				\ods{\gC_3\coons{\cnP\connn{K_{1},\ldots,K_{n}}}{R_3}}{\dD_8}{L_P}{\GS}
				\qomma
				\ods{\unit}{\dD_i'}{K_{i}\lpar M_i}{\GS}
				\qomma
				\ods{\unit}{\dD_9}{\gC_3}{\GS}
				\qomma
			\end{equation*}
			for all $i\in\set{1,\ldots,n}$. 
			Then, our derivation  of  $\gG\coons{(M_1\ltens N_1)\lpar\cdots\lpar(M_n\ltens N_n)}{S}$ 
			is constructed as follows
			\begin{equation*}
				\scalebox{.85}{$\qquad\qquad%
				\odN{
				\ods{\unit}{\dD_1}{\gC_1\Coons{
						\ods{\unit}{\dD_7}{\gC_2\Coons{\ods{\unit}{\dD_9}{\gC_3\Coons{
								\ods{\unit}{\dD_6}{\gL_\cnP\lpar
									\cneg P\Connn{\odn{N_1\ltens \ods{\unit}{\dD_1'}{K_{1}\lpar M_1}{}}{\ssdr}{(M_1\ltens N_1)\lpar  K_1}{},\ldots,
										\odn{N_n\ltens \ods{\unit}{\dD_n'}{K_{n}\lpar M_n}{}}{\ssdr}{(M_n\ltens N_n)\lpar K_n}{}
								}}{}
							}{R_3}
							}{}}{R_2}
						}{}}{R_1}
				}{}				
			}{\ssdr}{
				\ods{\gC_1\Coons{\ods{\gC_2\Coons{			\ods{\gC_3\coons{\cnP\connn{K_{1},\ldots,K_{n}}}{R_3}}{\dD_8}{L_P}{}\lpar\gL_\cnP}{R_2}}{\dD_4}{\gL}{}\lpar(M_1\ltens N_1)\lpar \cdots \lpar (M_n \ltens N_n)}{R_1}
				}{\dD_3}{\gG\coons{(M_1\ltens N_1)\lpar \cdots \lpar (M_n \ltens N_n)}S}{}}{}
			$}
			\end{equation*}
			
	
\begin{figure*}[!t]
  \rotatebox{90}{
	\scalebox{.79}{\hskip-1em
		$
		\ods{
			\ods{\unit}{\dD_1}{
				\gC_1\Coons{
					\odn{
						\ods{\unit}{\dD_7}
						{\gC_2\Coons{
								\odN{\ods{\unit}{\dD_{11}}{\gC_4\Coons{
											\odN{
												\ods{\unit}{\dD_Y}{H_{Y}\lpar\cnP\connn{M_1,\ldots,M_{j-1},\unit,M_{j+1},\ldots M_n}}{}
												\ltens 	
												\odn{
													\ods{\unit}{\dD_{14}}
													{\gC_{5}\Coons{\ods{\unit}{\dD_{13}}{N_j \lpar \gK_{N_j}}{}}{R_3}}{}
												}{\ssdr}{N_j \lpar 
													\ods{\gC_5 \coons{\gK_{N_j}\lpar \unit}{R_3}}
													{\dD^\unit_{12}}{\gL_\cnP \lpar  \cnP\connn{N_1,\ldots, N_{j-1}, \unit, N_{j+1},\ldots, N_n}}{}
												}{}
											}{\ssdr}{
										\left(N_j \ltens \ods{\unit}{\dD_X}{H_{X}\lpar M_j}{}\right) \lpar L_\cnP \lpar H_Y
												\lpar
												\left(\gP\connn{M_1,\ldots,M_{j-1},\unit,M_{j+1},\ldots M_n}
												\ltens 								
												\cnP\connn{N_1,\ldots, N_{j-1}, \unit, N_{j+1},\ldots, N_n}
												\right)
											}{}
										}{R_4}}{}
								}{\ssdr}{
									(M_j \ltens N_j) \lpar L_\cnP \lpar	\ods{\gC_4\coons{H_{X}\lpar H_{Y}}{R_4}}{\dD_{10}}{L_\gP}{}
									\lpar
									\left(P\connn{M_1,\ldots, M_{j-1}, \unit, M_{j+1},\ldots, M_n}\ltens \cnP\connn{N_1,\ldots,N_{j-1},\unit,N_{j+1},\ldots N_n}\right)
								}{}
							}{R_2}}{}
					}
					{\ssdr}{
						\ods{\gC_2\coons{\gL_P\lpar\gL_\cnP}{R_2}}{\dD_4}{\gL}{}
						\lpar(M_j\ltens N_j)
						\lpar 
						\ods{P\connn{M_1,\ldots, M_{j-1}, \unit, M_{j+1},\ldots, M_n}\ltens \cnP\connn{N_1,\ldots,N_{j-1},\unit,N_{j+1},\ldots N_n}  }
						{\dD^*}
						{ (M_1\ltens N_1)\lpar \cdots \lpar (M_{i-1}\ltens N_{i-1}) \lpar  \unit \lpar  (M_{i+1}\ltens N_{i+1}) \lpar \cdots \lpar (M_n \lpar N_n)    }
						{\mbox{Lemma~\ref{lem:g}}}
					}{}
				}{R_1}}{}
		}
		{\dD_3}{\gG\coons{(M_1\ltens N_1) \lpar \cdots \lpar (M_n \ltens N_n)}S}{}
		$
	        \hskip-2em
  }}
	\caption{Derivation for case \ref{p-up:B} in the proof of Theorem~\ref{thm:p-up}}
	\label{fig:p-up(b1)}
\end{figure*}	

			
			\item\label{p-up:B}
			There are graphs
			$H_X$ and $H_Y$ 
			and
			a context $\gC_4\coonso{R_4}$, 
			such that
			\begin{equation*}
                          \scalebox{.9}{$
				\ods{\gC_4\coons{H_{X}\lpar H_{Y}}{R_4}}{\dD_{10}}{L_\gP}{\GS}
				\qomma
				\ods{\unit}{\dD_X}{H_{X}\lpar M_j}{\GS}
				\qomma
				\ods{\unit}{\dD_Y}{H_{Y}\lpar \gP\connn{M_1,\ldots,M_{i-1},\unit,M_{i+1},\ldots M_n}}{\GS}
				\quand
				\ods{\unit}{\dD_{11}}{\gC_4}{\GS}
                                $}
			\end{equation*}
			for some $i\in\set{1,\ldots,n}$. 
			
				We apply context reduction (Lemma~\ref{lem:conred})
				to the derivation $\dD_6$ with conclusion
				$$\gL_{\cneg P}\lpar P\connn{N_1,\ldots, N_{j-1}, N_j, N_{j+1},\ldots, N_n}$$
				to obtain a context 
				$\gC_{5}$ and a graph $\gK_{N_j}$ such that 
				\begin{equation*}
                 \scalebox{.9}{$
			    \ods{\gC_5\coons{\gK_{N_j} \lpar \YYY}{R_5}}{\dD_{12}^\YYY}{\gL_{\cnP} \lpar  \cnP\connn{N_1,\ldots, N_{j-1}, \YYY, N_{j+1},\ldots, N_n}}{\GS}
					\qomma
					\ods{\unit}{\dD_{13}}{\gN_j \lpar \gK_{N_j}}{\GS}
					\quand
					\ods{\unit}{\dD_{14}}{\gC_{5}}{\GS}
                                        \hskip-1em
                                $}
				\end{equation*}
                For every graph $\YYY$.
				Then, our proof of  $\gG\coons{(M_1\ltens N_1)\lpar\cdots\lpar(M_n\ltens N_n)}{S}$ is constructed as shown in Figure~\ref{fig:p-up(b1)}.
				
				Observe that the derivation $\dD^*$ in Figure~\ref{fig:p-up(b1)} is a derivation in $\SGSu$. 
				However,
				all instances of 
				$\aiur$ and $\ssur$ in $\dD^*$ can be eliminated using 
				using Theorems~\ref{thm:ai-up} and \ref{thm:ss-up},
				while,
				in order to eliminate the instances of $\pur$,
				we invoke the induction hypothesis 
				since
				the principal modules in the $\pur$ instances in $\dD^\ast$ 
				have strictly smaller size than the one we started with.
			
		\end{enumerate}
		
    \item Now assume $ \cneg P\connn{N_1,\ldots,N_n}=\unit$. 
	Then by context reduction we already have a derivation $\dD_5=\dD_2$ of 
	$\gL\lpar(P\connn{M_1,\ldots,M_n}\ltens\cneg P\connn{N_1,\ldots,N_n})\isom\gL_P\lpar P\connn{M_1,\ldots,M_n}$, to which we apply Lemma~\ref{lem:splitting:prime} as above. Then case (A) is identical to the one above, and case (B) is simplified to
	\begin{equation*}
                 \scalebox{.9}{$
		\ods{
			\odN{
				\ods{\unit}{\dD_3}
				{\gC_1\Coons{
						\ods{\unit}{\dD_4}{\gC_2\Coons{  
								\ods{\unit}{\dD_X}{\gH_X\lpar\gM_j}{}
								\lpar
								\ods{\unit}{\dD_Y}{\gH_Y\lpar\gP\connn{\gM_1,\ldots,\gM_{j-1},\unit,\gM_{j+1},\ldots,\gM_n}}{}	
							}{R_2}}{} 
					}{R_1}}{}
			}{\ssdr}{
				\gC_1\Coons{	\ods{\gC_2\coons{H_X\lpar H_Y}{R_2}}{\dD_\gL}{\gL}{}  \lpar (M_1 \lpar \cdots \lpar M_n)  }{} }{R_1}
		}
		    {\dD_1}{\gG\coons{M_1\lpar \cdots \lpar M_n }S}{}
                    $}
	\end{equation*}
        \vskip-2ex
        \qedhere
        \end{itemize}
\end{proof}


\subsection*{Alternative formulations of the rules of $\GS$}
Further admissibility results follow immediately  without having to pass via splitting and context reduction, now that we have established cut elimination.
We can in fact considerably strengthen Theorem~\ref{thm:up}, as follows.

\begin{cor}\label{cor:adm}
  Let $\rr$ be an inference rule. If for every instance
  $\smash{~\vlinf{\rr}{}{B}{A}~}$ of that rule we have \;$\proves[\GS]{A\limp B}$\;, then $\rr$ is admissible for $\GS$.
\end{cor}

\begin{proof}
  Assume we have \;$\proves[\GS]{\gA}$. By
  Corollary~\ref{cor:deduction} we have a derivation in $\SGS$ from
  $\gA$ to $\gB$. Hence, we have \;$\proves[\SGS]{\gB}$, and therefore by Corollary~\ref{cor:gs-sgs} also \;$\proves[\GS]{\gB}$.
\end{proof}

Now consider the following two inference rules:
\begin{equation}
  \vlinf{\gdr}{}
        {G\connn{M_1,\ldots,M_n}\lpar\cneg G\connn{N_1,\ldots,N_n}}
        {\vls([M_1;N_1];\cdots;[M_n;N_n])}
        \qquad
        \vlinf{\gur}{}
        {\vls[(M_1;N_1);\cdots;(M_n;N_n)]}
        {G\connn{M_1,\ldots,M_n}\ltens\cneg G\connn{N_1,\ldots,N_n}}        
\end{equation}
Their shape is similar to $\pdr$ and $\pur$, but there are no side
conditions, i.e., $\gG$ can be any graph, and there are no emptiness
or non-emptiness conditions for $\gM_i$ or $\gN_i$.

\begin{cor}\label{cor:g-adm}
  The two rules $\gdr$ and $\gur$ are admissible for $\GS$.
\end{cor}

\begin{proof}
  We are going to show that for all graphs $\gG$ with $\sizeof{\vG}=n$ and graphs $\gM_1,\gN_1,\ldots,\gM_n,\gN_n$ we have that
  \begin{equation}\label{eq:g-adm}
    \proves[\GS]{(\cneg{\gM_1}\ltens\cneg{\gN_1})\lpar\cdots\lpar(\cneg{\gM_1}\ltens\cneg{\gN_1})\lpar G\connn{M_1,\ldots,M_n}\lpar\cneg
    G\connn{N_1,\ldots,N_n}}\quadfs
  \end{equation}
  Then this corollary follows via Corollaries~\ref{cor:context}
  and~\ref{cor:adm}. To prove~\eqref{eq:g-adm}, we use the following derivation
  \begin{equation*}
    \nvls{
    \od{\odI{\odI{\odh{\unit}}{n\times\idr}{
     \ods{((\cneg{\gM_1}\ltens\cneg{\gN_1})\lpar\gM_1\lpar\gN_1)\ltens\cdots\ltens((\cneg{\gM_n}\ltens\cneg{\gN_n})\lpar\gM_n\lpar\gN_n)}
        {\dD}{\gG\connn{(\cneg{\gM_1}\ltens\cneg{\gN_1})\lpar\gM_1,\ldots,(\cneg{\gM_n}\ltens\cneg{\gN_n})\lpar\gM_n}\lpar\cneg\gG\connn{\gN_1,\ldots,\gN_n}}{\GS}
        }{}}{n\times\ssdr}{
        (\cneg{\gM_1}\ltens\cneg{\gN_1})\lpar\cdots\lpar(\cneg{\gM_n}\ltens\cneg{\gN_n})
        \lpar{\gG\connn{\gM_1,\ldots,\gM_n}\lpar\cneg\gG\connn{\gN_1,\ldots,\gN_n}}}{}
    }}
  \end{equation*}
  where $\dD$ exists by Lemma~\ref{lem:g}. To see this, observe that
  for each $i\in\set{1,\ldots,n}$, there are two cases: either
  $(\cneg{\gM_i}\ltens\cneg{\gN_i})\lpar\gM_i\neq\unit$ or
  $(\cneg{\gM_i}\ltens\cneg{\gN_i})\lpar\gM_i=\unit$. In the second
  case we also have $\cneg{\gN_i}=\unit$, and therefore also
  $\gN_i=\unit$. Hence, the condition for applying Lemma~\ref{lem:g}
  is fulfilled.
\end{proof}

The above corollary shows that a variant of $\GS$ where we use $\gdr$ instead of $\pdr$ is equivalent to the system $\GS$ using only $\pdr$, as it is defined in Figure~\ref{fig:SGS}.
Indeed, as noted previously in Section~\ref{sec:system}, a weaker constraint on the $\pdr$ and $\pur$ was used in the conference version of this paper~\cite{Acclavio2020},
which did not insist that all of the components of one graph are non-empty when applying these rules. Thus we have tightened the proof search space of $\GS$, but have not changed the expressive power of the logic.

The reason why it is not immediately obvious whether the stronger side conditions in this version of paper can be used, is that Corollary~\ref{cor:g-adm} is a proper admissibility result rather than a derivability result. As shown in Lemma~\ref{lem:g}, if we drop the condition in $\pdr$ that $\gP$ is prime then the rule for general graphs is derivable.
However, the condition that all components of one graph are non-empty is more subtle, since the resulting variant of $\pdr$ without that side-condition is not derivable using the rules of $\GS$ with the side-conditions fixed in Figure~\ref{fig:SGS}. 
For a counterexample to the derivability of the unconstrained $\gdr$ rule in $\GS$, consider the proof on the left below and observe that the instance of $\gdr$ cannot be derived using the rules of $\GS$.
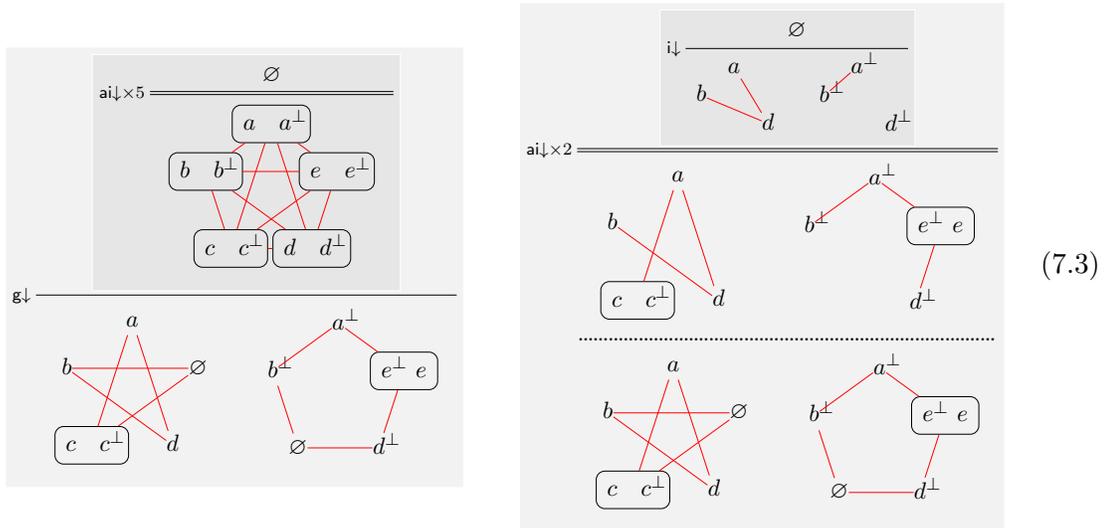
\begin{equation}\label{eq:gdr}
\adjustbox{max width=.9\textwidth}{
\def\diameter{32pt}
\newcommand{\pentagon}[5]{
\begin{array}{c}
	\begin{tikzpicture}
		\draw (90+0*360/5:\diameter) node {$#1$};
		\draw (90+1*360/5:\diameter) node {$#2$};
		\draw (90+2*360/5:\diameter) node {$#3$};
		\draw (90+3*360/5:\diameter) node {$#4$};
		\draw (90+4*360/5:\diameter) node {$#5$};
	\end{tikzpicture}
\end{array}
}
\odn{
\odN{
\unit
}{\aidr \times 5}{
\pentagon{\vmodule{3}{\va1 \quad \vna1~}}{\vmodule{4}{\vb1 \quad \vnb1~}}{\vmodule{1}{\vc1 \quad \vnc1~}}{\vmodule{5}{\vd1 \quad \vnd1~}}{\vmodule{2}{\ve1 \quad \vne1~}}
\edges{M2/M1,M3/M1,M4/M1,M5/M1,M3/M2,M4/M2,M5/M2,M4/M3,M5/M3,M5/M4}
}{}
    }{\gdr}{
    \pentagon{\va1}{\vb1}{\vmodule{1}{\vc1 \quad \vnc1~}}{\vd1}{\vu1}
\quad
\pentagon{\vna1}{\vnb1}{\vu2}{\vnd1}{\vmodule{2}{\vne1 \quad \ve1 }}
\edges{a1/M1,a1/d1,b1/u1,b1/d1,M1/u1}
\edges{na1/nb1,nb1/u2,u2/nd1,nd1/M2,M2/na1}
 }{}
\qquad
\odN{
\odn{
\unit
}{\idr}{
\begin{array}{ccc}
	& \va1\\
	\vb1 \\
	&&\vd1
\end{array}
\quad
\begin{array}{ccc}
	& \vna1\\
	\vnb1\\
	&&\vnd1
\end{array}
\edges{a1/d1,b1/d1}
\edges{na1/nb1}
}{}
}{\aidr\times 2}{
	\vlderivation{\vlid{}{}{
		\pentagon{\va1}{\vb1}{\vmodule{1}{\vc1 \quad \vnc1~}}{\vd1}{\vu1}
		\quad
		\pentagon{\vna1}{\vnb1}{\vu2}{\vnd1}{\vmodule{2}{\vne1 \quad \ve1 }}
		\edges{a1/M1,a1/d1,b1/u1,b1/d1,M1/u1}
		\edges{na1/nb1,nb1/u2,u2/nd1,nd1/M2,M2/na1}
	}{
		\vlhy{
			\pentagon{\va1}{\vb1}{\vmodule{1}{\vc1 \quad \vnc1~}}{\vd1}{}
			\quad
			\pentagon{\vna1}{\vnb1}{}{\vnd1}{\vmodule{2}{\vne1 \quad \ve1 }}
			\edges{a1/M1,a1/d1,b1/d1}
			\edges{na1/nb1,nd1/M2,M2/na1}
		}
	}
}
}{}
}
\end{equation}
However, observe that there is a proof of the same conclusion on the right above, using the rules of $\GS$ with the correct side conditions.

The reason why we use the rules $\pdr$ and $\pur$ instead of the rules $\gdr$ and
$\gur$ is to have more control over the structure of a derivation, such
that we are able to prove cut elimination. In fact, one can say that
the main purpose of using the modular decomposition of graphs into prime
graphs (instead of arbitrary quotient graphs) is to be able to prove
the splitting and context reduction lemmas of Section~\ref{sec:splitting}.



\section{Analyticity and Conservativity}\label{sec:MLL}

In sequent calculi, analyticity is often associated to the subformula
property, demanding that each subformula in the conclusion of a rule
also occurs in its premise. This, in turn, entails that every formula
that occurs in a proof is a subformula of the conclusion of the proof.\looseness=-1

This notion of analyticity cannot be applied to deep inference, where
there is no subformula property, and where in general, analyticity is
associated to the possibility of eliminating the up-rules from the
system~\cite{Bruscoli2016}, such that at any time in the proof search
only a finite number of rules can be applied.
	
In this sense, the rule $\gdr$ discussed above cannot be considered
analytic. The absence of constraints on the non-emptiness of the
modules in the conclusion allows us to apply the rule in infinitely many
ways.  By means of example, any application of the $\pdr$-rule to a prime
graph $\pfour$ could be considered as an instance of
$\gdr$ on any prime graph, since every prime graph contains an
induced $\pfour$. 

However, the step from formulas/cographs to general graphs allows us
to formulate a notion of analyticity for $\GS$, that is more similar
to the well-known notion of analyticity for sequent proofs than what
is usually present for deep inference systems.  For this we introduce
the notions of \emph{connector} and \emph{subconnector}, which are for our
notion of analyticity what formulas and subformulas are for the
standard notion.

\begin{defi}
A prime graph $\gP$ is a \defn{connector} of a graph $\gG$ 
if $\gP$ occurs in the modular decomposition of $\gG$.
A \defn{subconnector} of $\gG$ is a prime graph that is an induced subgraph of a connector of~$\gG$.
\end{defi}

\begin{exa}
  The graph from Example~\ref{ex:modDec} has three occurrences of the connector $\pfour$,
  one occurrence of the connector $\ltens$ and one occurrence of the connector $\lpar$.  Note that $\ltens$ is a
  subconnector of any graph containing at least one edge, and $\lpar$ is a
  subconnector of any graph containing two distinct vertices $u$ and $v$
  that are not connected by an edge in $\gG$. The graph $\pfour$ is a
  subconnector of any graph that has a prime connector $\gP$ with
  $\sizeof\vP\geq 4$.
\end{exa}
Intuitively, if we consider prime graphs as generalised connectives (see Section~\ref{sec:generalised}), a connector of $\gG$ is an occurrence of such a connective in the modular decomposition tree of~$\gG$.
With the help of connectors and subconnectors, we can now formulate our analyticity result.

\def\GSplus{\GS\cup\set{\ssur}}

\begin{thm}[Analyticity]\label{thm:anaCon}
  Let $\gG$ be a graph that is provable in $\GSplus$.
  Then $\gG$ admits a derivation $\tdD$ in $\GSplus$ 
  such that every connector of a graph occurring in $\tdD$ is a subconnector of $\gG$.
\end{thm}

\begin{figure}[t]
  \begin{adjustbox}{max width = \textwidth}$
    \vlnostructuressyntax
    \begin{array}{c}
      \gP\Connn{\ods{\gM_1}{\dD_1}{\gM_1'}{}, \dots, \ods{\gM_n}{\dD_n}{\gM_n'}{}}
      \isom
      \od{\odo{\odh{\gP\Connn{\gM_1, \dots, \gM_{i-1},\ods{\gM_1}{\dD_1}{\gM_1'}{}, \gM_{i+1}, \dots,  \gM_n}}
	}{}{\gP\Connn{\ods{\gM_i}{\dD_i}{\gM_i'}{},\ldots, \ods{\gM_{i-1}}{\dD_{i-1}}{\gM_{i-1}'}{},\gM',\ods{\gM_{i+1}}{\dD_{i+1}}{\gM_{i+1}'}{}, \ldots, \ods{\gM_n}{\dD_n}{\gM_n'}{}}}{}}
      \\\\[1em]
      \odn{\gC\Coons{\ods{\gG}{\dD}{\gG'}{}}{R}}{\ssdr}{
	\gC\Coons{\unit}{R} \lpar \gG'
      }{}
      \isom
      \odn{\gC\Coons{\gG}{R}}{\ssdr}{
	\gC\Coons{\unit}{R} \lpar \ods{\gG}{\dD}{\gG'}{}
      }{}		
      \qquad\qquad
      \odn{\gC\Coons{\gG\lpar \gH}{R}}{\ssdr}{
	\odn{\gC\Coons{\gG}{R}}{\ssdr}{\gC\Coons{\unit}{R}\lpar \gG}{} \lpar \gH
      }{}
      \isom
      \odn{\gC\Coons{\gG\lpar\gH}{R}}{\ssdr}{
	\gC\Coons{\unit}{R} \lpar \gG \lpar \gH
      }{}		
      \\\\[1em]
      \odn{
	(\gM_1 \lpar \gN_1)\ltens \cdots \ltens (\gM_{i-1} \lpar \gN_{i-1})
	\ltens 
	\left(\ods{\gM_i}{\dD_1}{\gM_i'}{} \lpar \ods{\gN_i}{\dD_2}{\gN_i'}{}\right)
	\ltens 
	(\gM_{i+1} \lpar \gN_{i+1})\ltens \cdots \ltens (\gM_n \lpar \gN_n)
      }{\pdr}{
	\gP\connn{\gM_1,\ldots, \gM_n}\lpar \cneg\gP\connn{\gN_1, \ldots, \gN_n}
      }{}
      \hskip4em
      \\\\[-1ex]
      \hskip10em
      \isom
      \odn{
	(\gM_1 \lpar \gN_1)\ltens \cdots  \ltens (\gM_n \lpar \gN_n)
      }{\pdr}{
	\gP\Connn{\gM_1,\ldots,\gM_{i-1},\ods{\gM_i}{\dD_1}{\gM_i'}{}  , \gM_{i+1} \gM_n}
	\lpar 
	\cneg\gP\Connn{\gN_1, \ldots,\gN_{i-1} ,\ods{\gN_i}{\dD_2}{\gN_i'}{} , \gN_{i+1}, \ldots, \gN_n}
      }{}						
    \end{array}
    $\end{adjustbox}
  \caption{Some of the rule permutations captured by the open deduction syntax.}
  \label{fig:rules-perm}
\end{figure}

\begin{proof}
  Let $\gG$ be given, and let $\dD$ be a derivation of
  $\gG$ in $\GSplus$. We are now going to 
  transform $\dD$ into a derivation $\tdD$ that does not contain any
  occurrences of a connector that is not a subconnector of $\gG$. Let $\gP$
  be such a connector occuring in $\dD$.  While going up in the
  derivation $\dD$, only an instance of $\ssdr$ can introduce $\gP$ as
  connector, i.e., the premise of the rule instance contains the prime
  graph $\gP$ as connector, but the conclusion does not. This happens by
  moving a module $\gM_i$ inside a context $\gC\Coons{\gP\connn{\gM_1,
      \ldots, \gM_{i-1}, \unit, \gM_{i+1}, \ldots, \gM_n}}R$ where
  $\gM_1, \dots,\gM_{i-1},\gM_{i+1}, \dots, \gM_n$ are non-empty. As
  the premise of the derivation $\dD$ is empty, this occurrence of
  $\gP$ must be destroyed eventually, i.e., $\dD$ contains a rule
  instance whose conclusion contains $\gP$ as connector, but the premise
  does not. This can be done via an instance of $\ssur$, $\pdr$, or
  $\aidr$. Below, we will list all cases how this can happen, and show
  how the derivation can be rewritten without introducing and
  destroying this occurrence of $\gP$.  For this, we proceed by
  induction on the multiset $\multiset{P_1,\ldots,P_n}$, where
  $P_1,\ldots,P_n$ are the occurrences of connectors in $\dD$ that are
  not subconnectors of $\gG$, and where we use 
  the multiset ordering:
  $\gQ<\gP$ iff $\sizeof\vQ<\sizeof\vP$.

  In the discussion of the cases below, we make use of the open
  deduction notation of derivations (as introduced in
  Section~\ref{sec:system}) and the possibility of rule permutations,
  which are shown in Figure~\ref{fig:rules-perm}. More precisely, we
  assume we have a connector $\gP$ introduced by an instance of $\ssdr$
  as described above, and because of the rule permutations, we can
  without loss of generality assume that the rule that destroys this
  connector $\gP$ is directly above, i.e., we have the following three
  cases:

  \begin{enumerate}
			
  \item 
    An instance of $\ssur$
    destroys the connector $\gP$ by extracting a module $\gM_j$ from $\gP\connn{\gM_1, \ldots, \gM_n}$.
    There are two cases to consider: $j=i$ and $j\neq i$.
    
    \begin{enumerate}
    \item If $j=i$, then we can without loss of generality assume that
      $i=j=1$, and with the rule permutations of
      Figure~\ref{fig:rules-perm}, we can move the instances of
      $\ssdr$ and $\ssur$ close to each other, such that we have one
      of the following two subderivations:
      $$
      	\odn{\gM_1\ltens \gC_1\Coons{\gC_2\Coons{
	      \gP\connn{\unit, \gM_2, \ldots, \gM_n}
	    }{R_2}}{R_1}
	}{\ssur}{
	  \gC_1\Coons{\odn{\gC_2\Coons{\gP\connn{\gM_1,\gM_2, \ldots, \gM_n}}{R_2}}{\ssdr}{
	      \gM_1\lpar \gC_2\Coons{\gP\connn{\unit, \gM_2, \ldots, \gM_n}}{R_2}}{}}{R_1}
	}{}
	$$
        or
        $$
	\odn{\gC_1\Coons{\odn{\gM_1\ltens \gC_2\Coons{
		\gP\connn{\unit, \gM_2, \ldots, \gM_n}
	      }{R_2}
	    }{\ssur}{
	      \gC_2\Coons{\gP\connn{\gM_1,\gM_2, \ldots, \gM_n}}{R_2}
	    }{}
	  }{R_1}}
	    {\ssdr}{
	      \gM_1\lpar \gC_1\Coons{\gC_2\Coons{\gP\connn{\unit, \gM_2, \ldots, \gM_n}}{R_2}}{R_1}
	    }{}
      $$
      where the graphs $\gM_1, \ldots, \gM_n$ are all non-empty and 
      $\gC_1\coonso{R_1}$ and $\gC_2\coonso{R_2}$ are arbitrary, possibly empty contexts.
      Now, we can remove the occurrence of the connector $\gP$ from the derivation $\dD$
      by replacing the subderivations above by
      $$
	\odn{\gM_1\ltens \gC_1\Coons{\gC_2\Coons{
	      \gP\connn{\unit, \gM_2, \ldots, \gM_n}
	    }{R_2}}{R_1}
	}{\ssur}{
	  \gC_2\Coons{\gM_1\lpar \gC_2\Coons{\gP\connn{\unit, \gM_2, \ldots, \gM_n}}{R_2}}{R_1}
	}{}
        $$
        or
        $$
	\odn{\gC_1\Coons{
	    \gM_1\ltens \gC_2\Coons{\gP\connn{\unit, \gM_2, \ldots, \gM_n}}{R_2}
	  }{R_1}}
	    {\ssdr}{
	      \gM_1\lpar \gC_1\Coons{\gC_2\Coons{\gP\connn{\unit, \gM_2, \ldots, \gM_n}}{R_2}}{R_1}
	    }{}
            $$
	    respectively.

    \item if $j\neq i$, then we can assume without loss of generality that $j=1$ and $i=2$, and we can (via the rule permutations in Figure~\ref{fig:rules-perm}) isolate a subderivation of one of the following two shapes	
      $$
      \odn{\gM_1\ltens \gC_1\Coons{\gC_2\Coons{
	    \gP\connn{\unit, \gM_2, \ldots, \gM_n}
	  }{R_2}}{R_1}
      }{\ssur}{
	\gC_1\Coons{\odn{\gC_2\Coons{\gP\connn{\gM_1, \ldots, \gM_n}}{R_2}}{\ssdr}{
	    \gM_2\lpar \gC_2\Coons{\gP\connn{\gM_1, \unit, \gM_3, \ldots, \gM_n}}{R_2}}{}}{R_1}
      }{}
      $$
      or
      $$
      \odn{
        \gC_1\Coons{\odn{\gM_1\ltens \gC_2\Coons{
	      \gP\connn{\unit, \gM_2, \ldots, \gM_n}
	    }{R_2}
	  }{\ssur}{
	    \gC_2\Coons{\gP\connn{\gM_1, \ldots, \gM_n}}{R_2}
	  }{}
	}{R_1}
      }{\ssdr
      }{
	\gM_2\lpar \gC_1\Coons{\gC_2\Coons{\gP\connn{\gM_1,\unit, \gM_3, \ldots, \gM_n}}{R_2}}{R_1}
      }{}
      $$
      where the graphs $\gM_1, \ldots, \gM_n$ are non-empty and 
      $\gC_1\coonso{R_1}$ and $\gC_2\coonso{R_2}$ are (possibly empty) contexts.
      We can replace these subderivations by 
      $$
      \odn{\gM_1\ltens \gC_1\Coons{
	  \odn{\gC_2\Coons{\gP\connn{\unit, \gM_2, \ldots, \gM_n}}{R_2}}{\ssdr}{
	    \gM_2\lpar \gC_2\Coons{\gP\connn{\unit, \unit, \gM_3, \ldots, \gM_n}}{R_2}
	  }{}
	}{R_1}
      }{\ssur}{
	\gC_1\Coons{\gM_2\lpar \gC_2\Coons{\gP\connn{\gM_1,\unit, \gM_3, \ldots, \gM_n}}{R_2}}{R_1}
      }{}
      $$
      or
      $$
      \odn{\gC_1\Coons{ \gM_1 \ltens \gC_2\Coons{
	    \gP\connn{\unit,\gM_2,\ldots, \gM_n}
	  }{R_2}}{R_1}
      }{\ssdr}{
	\gM_2\lpar \gC_1\Coons{\odn{\gM_1\ltens \gC_2\Coons{
	      \gP\connn{\unit,\unit, \gM_3, \ldots, \gM_n}
	    }{R_2}
	  }{\ssdr}{
	    \gC_2\Coons{\gP\connn{\gM_1,\unit\gM_3, \ldots, \gM_n}}{R_2}
	  }{}
	}{R_1}
      }{}
      $$ respectively.  In both cases, the connector $\gP$ disappears,
      but can potentially be replaced by the a smaller connector,
      determined by $\gP\connn{\unit,\unit, \gM_3, \ldots,
        \gM_n}$. However, in either case we can apply the induction
      hypothesis.
    \end{enumerate}
  \item An instance of $\pdr$ destroys the connector $\gP$ that has been
    introduced by a $\ssdr$. This means that modulo rule permutations,
    we have a subderivation of the following shape (where we assume
    without loss of generality that $i=1$):
    $$
    \odn{
      \gC\Coons{
	\odn{
	  {((\gM_1 \lpar \gN_1)) \ltens \cdots \ltens ((\gM_n \lpar \gN_n)) }
	}{\pdr}{
	  \gP\connn{\gM_1, \ldots , \gM_n}\lpar \cneg \gP\connn{\gN_1, \ldots, \gN_n}
	}{}
      }{R}
    }{\ssdr}{\gM_1 \lpar
      \gC\Coons{\gP \connn{\unit, \gM_2,\dots, \gM_n} \lpar
        \cneg \gP\connn{\gN_1, \ldots, \gN_n}}{R}}{} 
    $$
    where $\gM_1\neq\unit$. We can replace this subderivation by
    $$
    \vlnostructuressyntax
    \odn{
      \gC\Coons{
	\odn{
	  (\gM_1 \lpar \gN_1) \ltens
	  \ods{(\gM_2 \lpar \gN_2) \ltens \cdots \ltens (\gM_n \lpar \gN_n) }
	      {\dD^*}
	      {
		\gP\connn{\unit, \gM_2, \ldots , \gM_n}
		\lpar 
		\cneg \gP\connn{\unit, \gN_2, \ldots, \gN_n}}
	      {\mbox{Lemma~\ref{lem:g}}}
	}
	    {\ssur}{\gP\connn{\unit, \gM_2, \ldots , \gM_n}\lpar \cneg \gP\connn{\gM_1 \lpar \gN_1, \gN_2, \ldots, \gN_n}
	    }{}}R}{\ssdr}{\gM_1 \lpar \gC \Coons{\gP\connn{\unit, \gM_2, \dots,\gM_n} \lpar \cneg \gP\connn{\gN_1,\gN_2, \ldots, \gN_n}}{R}}{}
    $$ in which the connector $\gP$ is no longer present. The
    derivation $\dD^*$ exists by Lemma~\ref{lem:g}. Note that the
    construction of that lemma entails that for all connectors $\gQ$
    that occur in $\dD^*$ and that do not already occur in $\gM_j$
    (for $j\in\intset{1}{n}$) or the context, we have
    $\sizeof\vQ<\sizeof\vP$.  We can therefore proceed by induction
    hypothesis. Note that we did not demand that $\gM_2, \ldots,
    \gM_n$ are non-empty, and therefore this case also applies when
    the roles of $\gP$ and $\cneg\gP$ are exchanged.

  \item 
    Finally, we have the case where an instance of $\aidr$ destroys the
    connector $\gP$. This happens when one of the modules $\gM_j$ is
    $a\lpar\cneg a$ which is removed by the $\aidr$ instance. Again,
    we have two cases: $i=j$ and $i\neq j$.
    \begin{enumerate}
    \item If $i=j$, then we can without loss of generality assume that
      $i=j=1$. Modulo rule permutations, we have subderivation
      $$
      \odn{
	\gC\Coons{\gP\Connn{\odn{\unit}{\aidr}{a\lpar\cneg a}{},
	    \gM_2, \ldots, \gM_n
	  }
	}{R}
      }{\ssdr}{(a\lpar \cneg a) \lpar \gC\Coons{\gP\connn{\unit,
	    \gM_2, \ldots, \gM_n
	  }
	}{R}}{}
      $$
      that can be replaced by
      $$
      \od{\odh{
          \odn{\unit}{\aidr}{(a\lpar \cneg a)}{}
          \lpar \gC\Coons{\gP\connn{\unit,\gM_2, \ldots, \gM_n}}{R}}}
      $$
      where the connector $\gP$ does not occur anymore.
    \item If $i\neq j$, we can assume without loss of generality that
      $i=2$ and $j=1$, i.e., we have a subderivation
      $$
      \odn{
	\gC\Coons{\gP\Connn{\odn{\unit}{\aidr}{a\lpar \cneg a}{},
	    \gM_2, \ldots, \gM_n
	  }
	}{R}
      }{\ssdr}{\gM_2 \lpar \gC\Coons{\gP\connn{(a\lpar \cneg a),
	    \unit, \ldots, \gM_n
	  }
	}{R}}{}
      $$
      which can be replaced by
      $$
      \odn{\gC\Coons{\gP\connn{\unit, \gM_2, \ldots, \gM_n}}{R}
      }{\ssdr}{
	\gM_2 \lpar \gC\Coons{\gP\Connn{\odn{\unit}{\aidr}{a\lpar \cneg a}{}, \unit,\gM_3, \ldots, \gM_n}}{R}
      }{}
      $$ in which the connector $\gP$ is no longer present (but which
      might contain a new connector $\gQ$ that is a subconnector of $\gP$ and has therefore smaller size). \qedhere
    \end{enumerate}
  \end{enumerate}
  \end{proof}

\begin{cor}\label{cor:cocons}
	Let $A$ be a cograph. If $\proves[\GS\cup\set{\ssur}]A$, then there is a derivation 
	\begin{equation}
		\small
		\ods{\unit}{\dD}{A}{\GS\cup\set{\ssur}}
	\end{equation}
	such that every graph occurring in $\dD$ is a cograph.
\end{cor}
\begin{proof}
	All subconnectors of a cograph are $\ltens$ or $\lpar$.
	By Theorem~\ref{thm:anaCon} 
	there is a derivation $\dD$ of $A$ in $\GS\cup\set\ssur$ 
	in which only the connectors $\ltens$ and $\lpar$ occur, i.e., every graph in $\dD$ is a cograph.
\end{proof}

Using this result, we can now show that $\GS$ is a conservative extension of
unit-free multiplicative linear logic with mix (denoted $\MLLm$)
\cite{girard:87}. The formulas of $\MLLm$ are as in~\eqref{eq:gram},
but without the unit, and the inference rules of $\MLLm$ 
are given in 
Figure~\ref{fig:mll}, where $\Gamma$ and $\Delta$ are \defn{sequents},
i.e., multisets of formulas, separated by commas.  We write $\proves[\MLLm]
\Gamma$ if the sequent $\Gamma$ is provable in $\MLLm$, i.e.~if there
is a finite proof tree in $\MLLm$ with conclusion~$\Gamma$, where every leaf of the tree is an axiom.
\begin{figure}[!t]
	\begin{tabular}{c@{\hskip3em}c@{\hskip3em}c@{\hskip3em}c}
		$\vlinf{\axrule}{}{a, \cneg a}{}$
		&
		$\vliiinf{\tensrule}{}{\Gamma, \phi \ltens \psi, \Delta}{\Gamma, \phi}{\quad}{\Delta, \psi}$
		&
		$\vlinf{\parrule}{}{\Gamma, \phi \lpar \psi}{\Gamma, \phi , \psi}$
		&
		$\vliiinf{\mixr}{}{\Gamma, \Delta}{\Gamma}{\quad}{\Delta}$
	\end{tabular}
	\caption{The inference rules of the system $\MLLm$}
	\label{fig:mll}
\end{figure}

\begin{lemma}\label{lem:switch-deri}
  Let $A$ and $B$ be cographs. Then
  \begin{equation}
    \odn{A}{\ssdr}{B}{}
    \quad\implies\quad
    \odv{A}{}{B}{\set{\swir}}
    \qquand
    \odn{A}{\ssur}{B}{}
    \quad\implies\quad
    \odv{A}{}{B}{\set{\swir}}
  \end{equation}
\end{lemma}
\begin{proof}
  By Theorem~\ref{thm:cograph}, the graphs $A$ and $B$ are cographs iff there are
  formulas $\phi$ and $\psi$ with $\graphof\phi\isom A$ and $\graphof\psi\isom B$. Now the statement follows from
  the corresponding statement for formulas (see e.g., Lemma~4.3.20
  in~\cite{dissvonlutz}).
\end{proof}

\begin{thm}\label{thm:conservativity}
  Let $A$ be a cograph. Then  $\proves[\GSplus]A$ \;iff\; $\proves[\set{\aidr,\swir}]{A}$.
\end{thm}
\begin{proof}
  The implication from right to left follows immediately from the fact
  that $\swir$ is a special case of both $\ssdr$ and $\ssur$
  (see Lemma~\ref{lem:switch}). 
  For the implication from left to right,
  apply Corollary~\ref{cor:cocons} to get a derivation $\dD$ that only
  uses cographs. Hence the rule $\pdr$ is not used in
  $\dD$. Therefore, by Lemma~\ref{lem:switch-deri}, we can get a
  derivation $\dD'$ that only uses the rules $\aidr$ and~$\swir$.
\end{proof}

\begin{cor}\label{cor:MLL}
  For any unit-free formula $\phi$,
  \begin{equation*}
    \proves[\MLLm] \phi \quiff \proves[\GS] {\graphof{\phi}}
  \end{equation*}
\end{cor}

\begin{proof}
  It has been shown before (see, e.g., \cite{gug:str:01,dissvonlutz})
  that a formula $\phi$ is provable in $\MLLm$ iff it is derivable
  using only the rules $\aidr$ and $\swir$ (when interpreted for
  formulas) modulo the equivalence $\fequiv$, defined
  in~\eqref{eq:fequiv}. By Theorem~\ref{thm:conservativity}, this is
  equivalent to having $\graphof{\phi}$ provable in
  $\set{\aidr,\swir}$.\footnote{Note that even though $\phi$ is
    unit-free, we allow units to occur in the derivation of $\phi$,
    in particular, in~\eqref{eq:switch} we can have $B=\emptyset$.}
  By Theorem~\ref{thm:conservativity} and Theorem~\ref{thm:ss-up}
  This is equivalent to $\proves[\GS] {\graphof{\phi}}$.
\end{proof}

\begin{cor}\label{cor:NP}
  Provability in $\GS$ is {\bf NP}-complete. 
\end{cor}

\begin{proof}
  Since $\MLLm$ is {\bf NP}-complete, we can conclude from
  Corollary~\ref{cor:MLL} that $\GS$ is {\bf NP}-hard. Containment in
  {\bf NP} has been shown in Observation~\ref{obs:size}.
\end{proof}



\section{Graphs as  Generalised Connectives}\label{sec:generalised}

In the previous \Cref{sec:MLL} we have shown that our system is a conservative  extension of multiplicative linear logic with mix ($\MLLm$): the logic $\MLLm$ is the restriction of $\GS$ to $\pfour$-free graphs.
Lemma \ref{lem:newhope} allows us to associate a modular decomposition-tree to a graph, in the same way as we associate a formula-tree to a cograph.
This suggests the use of graphs as connectives in order to define \emph{generalised formulas} that are encoding graphs.
Our choices to name the two graphs on two vertices $\lpar$ and $\ltens$ (see \ref{eq:par-tens})
and to denote the graph operations of union and  join  with the same symbols (Definition \ref{def:par})
are coherent with this intent.
In fact, according with the notation for the composition-via-graph (Definition \ref{def:graphCompisitionVia}) we have 
$$\gG \ltens \gH=\ltens\connn{\gG, \gH} \quad\mbox{ and  }\quad \gG \lpar \gH=\lpar\connn{\gG, \gH}.$$

\begin{defi}[Connective-as-graph]
  \label{def:con-as-graph}
  Any graph $\gG$ with $\sizeof{\vertices[\gG]}=n$ \emph{describes} an $n$-ary connective $\gG$.
\end{defi}

It should be clear that this notion of connective goes beyond the one of \emph{synthetic connective}, which is a general connective definable as composition of standard binary connectives (see e.g.~\cite{andreoli:92,girard2000meaning,mil:pim:13}).
In this section we discuss the exact relation between our \emph{connectives-as-graphs} notion  with the notion of \emph{(multiplicative) generalised connectives} from \cite{girard:87:b,danos:regnier:89,mai:19,acc:mai:20}.
In particular we show that our notion fits the definition of multiplicative connective, 
but it describes  different mathematical objects from the ones known in the literature. 
For this, we first recall the notion of \emph{multiplicative generalised connective} 
from the early work in linear logic
\cite{girard:87:b,danos:regnier:89}, and their description as sets of
partitions. 
Then we show that the two notions are different. 
More precisely, we show that
\begin{enumerate}
	\item our system proves different generalised
	formulas than $\MLLm$ extended with generalised connectives,
	
	\item the symmetries induced by our generalised connectives-as-graphs
	are different from the generalised connectives-as-partitions
	from~\cite{girard:87:b,danos:regnier:89,mai:19,acc:mai:20},
	
	\item the notion of \emph{decomposable
	  connective} differs in the two settings, and 
	
	\item our generalised connectives-as-graphs do not suffer from the so called \emph{packaging
	  problem}~\cite{danos:regnier:89}.
\end{enumerate}

\begin{defi}[multiplicative generalised connective \cite{girard:87:b,danos:regnier:89}]
  \label{def:mult-gen-con}
  A \emph{multiplicative generalised connective} is an
$n$-ary connective $\Ccon$ which admits a \emph{linear} and
\emph{context-free} (introduction) rule in the sequent calculus, that
is, a rule of the form
\begin{equation}\label{eq:multRule}
\vlderivation{
\vliiin{}{}{\vdash\Gamma, \Ccon(\phi_{\sigma(1)}, \dots, \phi_{\sigma(n)})}
	{\vlhy{\vdash\Gamma_1,\phi_{1},\dots,\phi_{k_1}}}
	{\vlhy{\cdots}}
	{\vlhy{\vdash\Gamma_k, \phi_{1+i_{k-1}}, \dots, \phi_{n}}}
}
\quad
\begin{array}{l}
	\mbox{with }
	\Gamma=\Gamma_1\uplus\cdots\uplus \Gamma_k
	\\
	\mbox{and }
	\sigma \mbox{ a permutation of }\intset1n
\end{array}
\end{equation}
\end{defi}
The linearity condition demands that each occurrence of a formula in $\Gamma$ or sub-formula of the \emph{active formula} 
$\Ccon(\phi_{\sigma(1)}, \dots, \phi_{\sigma(n)})$
occurs exactly once in the premise, while the context-freeness condition demands that the application of the rule does not depend on the shape of the premises.

By means of example, consider the connectives of linear logic. The multiplicative conjunction $\ltens $ and $\lpar$ disjunction are multiplicative connectives, while the additive conjunction $\lwith$ and disjunction $\lplus$ of linear logic are not.
	$$
	\vliinf{\ltens}{\mbox{where $\Gamma=\Gamma_1\uplus\Gamma_2$}}{\vdash\Gamma, \phi\ltens \psi}{\vdash\Gamma_1, \phi}{\vdash\Gamma_2, \psi} 
	\quad
	\vlinf{\lpar}{}{\vdash\Gamma, \vls[\phi; \psi]}{\vdash\Gamma, \phi,\psi}
	\quad
	\vliinf{\lwith}{}{\vdash\Gamma, \phi\lwith \psi}{\vdash\Gamma, \phi}{\vdash\Gamma, \psi}
	\quad
	\vlinf{\lplus}{}{\vdash\Gamma, \phi\lplus \psi}{\vdash\Gamma, \phi}
	$$
In fact, each formula in $\Gamma$ and the two sub-formulas  $\phi$ and $\psi$ of the active formula occur exactly once in the premises of the rules for $\ltens$ and $\lpar$.
On the contrary, the rule for $\lwith$ is not context-free since it can be applied only if the premises are of the form $\Gamma, \phi$ and $\Delta,\psi$ with $\Gamma=\Delta$, that is, the rule depends on the context.
Moreover this rule is not linear since each occurrence of a formula in $\Gamma$ occurs twice in the premises.
The rule for $\lplus$ is not linear since it produces a sub-formula $\psi$ which does not occur in the premise.

Following Definition~\ref{def:mult-gen-con}, our connectives-as-graphs (Definition~\ref{def:con-as-graph}) are multiplicative connectives.
In fact, the $n$-ary connective described by a graph $\gG$ with $\sizeof{\vertices[\gG]}=n$ admits a unique sequent calculus (introduction) rule of the form
$$
\vlderivation{
	\vliiin{\gG}{}{\vdash\Gamma, \gG\connn{\phi_1, \dots, \phi_n}}
	{\vlhy{\vdash\Gamma_1,\phi_1}}
	{\vlhy{\cdots}}
	{\vlhy{\vdash\Gamma_n,\phi_n}}
}
\quad\mbox{ with }\Gamma=\Gamma_1\uplus\cdots\uplus \Gamma_n
$$
as it is derivable in $\GS$ if we interpret formulas as graphs and sequents of formulas as their disjoint union, and ``premises concatenation'' as the joint of the sequents in the premises, that is
$$
\vlinf{\gG}{}{\graphof{\Gamma_1}\lpar  \cdots\lpar \graphof{\Gamma_n} \lpar \gG\connn{ \graphof{\phi_1}, \dots, \graphof{\phi_n}}}{
  (\graphof{\Gamma_1}\lpar \graphof{\phi_1})\ltens \cdots \ltens(\graphof{\Gamma_n}\lpar \graphof{\phi_n})}
~\sim~ 
\od{
\odI
	{\odd
		{\odh
			{((\graphof{\Gamma_1}\lpar\graphof{ \phi_1}))\ltens \cdots \ltens((\graphof{\Gamma_n}\lpar \graphof{\phi_n}))}
		}
		{\dD}
		{ \cneg\gG\connn{\graphof{\Gamma_1}, \dots, \graphof{\Gamma_n}}\lpar \gG\connn{\graphof{\phi_1}, \dots,\graphof{\phi_n}}}
		{\GS}
	}
{\ssdr}
{\graphof{\Gamma_1}\lpar  \cdots\lpar \graphof{\Gamma_n}, \lpar \gG\connn{\graphof{\phi_1}, \dots,\graphof{\phi_n}}}
{}
}
$$
where $\dD$ exists by Lemma~\ref{lem:g}.

Let us now recall the description of generalised connectives by means of \emph{sets of partitions} from \cite{danos:regnier:89,mai:19,acc:mai:20}.
We denote by $\partofn$ the set of partitions of the set $\set{1,\dots, n}$ and we say that an $n$-ary connective is described by a non-empty subset of $\partofn$. 
To improve readability, we follow the following policy for parenthesis: $\set{~}$ for sets (of partitions), $\partition{~}$ for partitions (i.e. family of disjointed covering sets of subsets of $\set{1,\dots, n}$) and $\block{~}$ for the elements of each partition (i.e. subsets of $\set{1,\dots, n}$).
Each of these partitions can be interpreted intuitively as a way in which it is possible to gather the sub-formulas of the active formula into the premises of a multiplicative rule.
By means of example, the binary rules for $\ltens$ and $\lpar $ are described respectively by the (singleton) partition sets $\set{\partition{\block{1},\block 2}}$ and  $\set{\partition{\block{1,2}}}$. 
If we define a generalised (synthetic) connective $\Ccon(a,b,c)=(a\ltens b) \lpar c$, then we can associate to it the multiple rules corresponding to all possible ways to produce this formula in multiplicative linear logic:
\begin{equation*}\label{eq:simpleCon-as-Part}
	\small
\begin{array}{c@{\; \rightsquigarrow\; }c@{\; \rightsquigarrow\; }c@{\;\;\;}|@{\;\;\;}c@{\; \rightsquigarrow\; }c@{\; \rightsquigarrow\; }c}
	\vlderivation{\vlin{\lpar}{}{\vdash(a\ltens b)\lpar c}{\vliin{\ltens}{}{\vdash(a\ltens b),c}{\vlhy{\vdash a,c}}{\vlhy{\vdash b}}}}
	&
	\vliinf{\Ccon^{\block{1,3},\block 2}}{}{\Ccon(1,2,3)}{1,3}{2}
	&
	\partition{\block{1,3},\block 2}
	&
	\vlderivation{\vlin{\lpar}{}{\vdash(a\ltens b)\lpar c}{\vliin{\ltens}{}{\vdash(a\ltens b),c}{\vlhy{\vdash a}}{\vlhy{\vdash b,c}}}}
	&
	\vliinf{\Ccon^{\block{1},\block{2,3}}}{}{\Ccon(1,2,3)}{1}{2,3}
	&
	\partition{\block 1,\block{2,3}}
\end{array}
\end{equation*}

Because of De Morgan duality, we need for every connective a dual connective.
The interactions of dual connectives by means of the $\cutr$-rule, together with the resulting cut-elimination procedure, enforces some additional structure between the sets of partitions.
Given two partitions $p,q\in \partofn$, we define their \emph{incidence graph} as the multi-graph whose vertices are the elements of $p$ and $q$, such that there is an edge between two different vertices for every element in their intersection.
For example, we have the following incidence graphs between $p=\partition{\block{1,3},\block 2}$ and each of $q_1=\partition{\block 1, \block 2, \block{3}}$, $q_2=\partition{\block{1},\block{2,3}}$ and $q_3=\partition{\block{1,2,3}}$:
$$
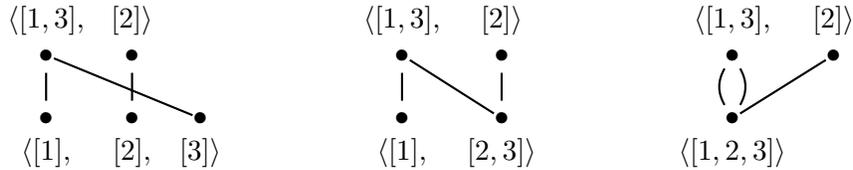
\begin{array}{cccccc}
\partition{{\block{1,3}} ,&{\block{2}}}
\\
\vbul{1}& \vbul{2}
\\[1em]
\vbul{4}&\vbul{5}&\vbul{6}
\\
\partition{{\block{1}} ,&{\block{2}}, & {\block{3}}}
\end{array}
\blackedges{n1/n4,n2/n5,n1/n6}
\qquad\qquad
\begin{array}{cccccc}
	\partition{{\block{1,3}} ,&{\block{2}}}
	\\
	\vbul{1}& \vbul{2}
	\\[1em]
	\vbul{4}& \vbul{5}
	\\
	\partition{{\block{1}} ,& {\block{2,3}}}
\end{array}
\blackedges{n1/n4,n2/n5,n1/n5}
\qquad\qquad
\begin{array}{cccccc}
	\partition{{\block{1,3}} ,&{\block{2}}}
	\\
	\vbul{1}& \vbul{2}
	\\[1em]
	\vbul{4}
	\\
	\partition{\block{1,2,3}}
\end{array}
\blackedges{n2/n4}
\blackbendedges{n1/n4,n4/n1}
$$
We say that two partitions $p,q\in\partofn$ are \emph{orthogonal} if their incidence graph is connected and acyclic.
In the example above, $p$ and $q_2$ are orthogonal, but $p$ and $q_1$ are not (because the incidence graph is not connected) and $p$ and $q_3$ are not (because the incidence graph is not acyclic).
We say that $P,Q\subseteq \partofn$ are \emph{orthogonal} if every partition in $P$ is orthogonal to every partition in $Q$.
\begin{defi}[Connective-as-partitions]
	Let $P$ be a set of partitions in $\partofn$ such that there is a set of partitions $Q\subseteq\partofn$ that is orthogonal to $P$.
        An $n$-ary connective $\Ccon$ \emph{is described} by $P$  
if $\Ccon$ admits exactly one multiplicative inference rule $\Ccon^p$ of shape~\eqref{eq:multRule} for each $p\in P$, such that any
two subformulas $\phi_i$ and $\phi_j$ of the principal formula $\Ccon(\phi_1, \dots, \phi_n)$ belong to the same premise iff $i$ and $j$ belong to a same element in $p$.
\end{defi}
 
 For the purpose of this section, we consider the pair of dual generalised $4$-ary connectives first introduced in \cite{girard:87:b} by means of permutations and later reformulated in \cite{danos:regnier:89} by means of partitions.
 Following the formalism in \cite{acc:mai:20}, we denote by $\gcon$ and $\ngcon$ the connectives respectively {described} by the following two sets of partitions in $\partof4$
\begin{equation}\label{eq:girCon}
  \gcon\colon
  \set{
 	\partition{\block{1,2},\block{3,4}}
 	,
 	\partition{\block{1,4},\block{2,3}}
 }
\quad \mbox{ and }\quad
\ngcon\colon
\set{
 	\partition{\block{1,3},\block{2},\block{4}}
 	,
 	\partition{\block{2,3},\block{1},\block{3}}
 }
 \end{equation}
That is, for the connective $\gcon$ we have the following two 
inference rules
\begin{equation}\label{eq:G4rules}
\vliiinf{\gcon^{\block{1,2},\block{3,4}}}{}{\Gamma, \Delta, \gcon(\phi,\psi,\rho,\chi)}{\Gamma, \phi, \psi}{\quad}{\Delta,\rho,\chi}
\qquad
\vliiinf{\gcon^{\block{1,4},\block{2,3}}}{}{\Gamma, \Delta, \gcon(\phi,\psi,\rho,\chi)}{\Gamma, \phi, \chi}{\quad}{\Delta,\psi,\rho}
\end{equation}
Let $\mathcal C$ be the set of non-decomposable (multiplicative) connectives-as-partitions.
We call a \emph{generalised formula} a formula that is generated by a countable set of positive or negative propositional atoms $\set{a,\cneg{a},b,\cneg b,\ldots}$ via the following grammar:
$$
\phi , \psi \Coloneqq a \mid \cneg a\mid \phi \lpar \psi \mid \phi \ltens \psi \mid \Ccon(\phi_1,\dots, \phi_n)
$$
where $\Ccon\in \mathcal C$ is any $n$-ary connective.
We denote by $\MLLgen$ the proof system over generalised formulas extending $\MLLm$ with, for each $\Ccon \in \mathcal C$, all the sequent rules of $\Ccon$.

\subsection*{Comparing Theorems}
We can now prove that $\MLLgen$ and $\GS$ deal with different objects by showing that  the connective $\gcon$ does not correspond to any instance of the connective $\pfour$, which is described by the only prime graph with 4 vertices.
From now on, we refer to $\pfour\connn{a,b,c,d}$ as the graph $\va1 \quad \vb1 \quad \vc1 \quad \vd1 \edges{a1/b1,b1/c1,c1/d1}$.

\begin{prop}
In $\MLLgen$, none of the following formulas or their equivalent sequents are provable:
\begin{equation}
	\label{eq:non-provableMLLG}
\begin{array}{c@{\quad}|@{\quad }c}
c\ltens(d\lpar (a\ltens b))
\limp
\gcon(a,b,c,d)
&
\gcon(a,b,c,d),  \cneg c\lpar (\cneg d\ltens (\cneg a \lpar \cneg b))
\\	\hline
d\ltens(c\lpar (a\ltens b))
\limp
\gcon(a,b,c,d)
&
\gcon(a,b,c,d), \cneg d\lpar (\cneg c\ltens (\cneg a \lpar \cneg b))
\\	\hline
a\ltens(c\lpar (b\ltens d))
\limp
\gcon(a,b,c,d)
&
\gcon(a,b,c,d),  \cneg a\lpar (\cneg c\ltens (\cneg b \lpar \cneg d))
\\	\hline
c\ltens(a\lpar (b\ltens d))
\limp
\gcon(a,b,c,d)
&
\gcon(a,b,c,d),  \cneg c\lpar (\cneg a\ltens (\cneg b \lpar \cneg d))
\\	\hline
a\ltens(d\lpar (b\ltens c))
\limp
\gcon(a,b,c,d)
&
\gcon(a,b,c,d), \cneg a\lpar (\cneg d\ltens (\cneg b \lpar \cneg c))
\\	\hline
d\ltens(a\lpar (b\ltens c))
\limp
\gcon(a,b,c,d)
&
\gcon(a,b,c,d),  \cneg d \lpar (\cneg a\ltens (\cneg b \lpar \cneg c))
\end{array}
\end{equation}
\end{prop}
\begin{proof}
We show that the first sequent is not provable. The same reasoning applies to the other sequents.
In a derivation of $\gcon(a,b,c,d),  \cneg c\lpar (\cneg d\ltens (\cneg a \lpar \cneg b))$ we can apply as first rule a $\lpar$.
Any derivation continuing with a $\ltens$ rule must have in one branch $\gcon(a,b,c,d)$  and either $\cneg d$ or $\cneg a \lpar \cneg b$. This makes it impossible to match all pairs of dual atoms in axioms. 
Similarly,  starting with one of the two $\gcon$ rules in \ref{eq:G4rules}, 
or starting with the $\lpar$ rule followed by one of these two rules, 
we encounter the same problem
of mismatched  atoms. 
\end{proof}

\begin{prop}
In $\GS$, the graphs corresponding to the following formulas are provable:
\begin{equation}
	\label{eq:provableGS}
\begin{array}{c@{\quad}|@{\quad }c}
	c\ltens(d\lpar (a\ltens b))
	\limp
	\pfour\connn{a,b,c,d}
	&
	\pfour\connn{a,b,c,d}\lpar   \cneg c\lpar (\cneg d\ltens (\cneg a \lpar \cneg b))
	\\
	c\ltens(d\lpar (a\ltens b))
	\limp
	\pfour\connn{b,a,c,d}
	&
	\pfour\connn{b,a,c,d}\lpar   \cneg c\lpar (\cneg d\ltens (\cneg a \lpar \cneg b))
	\\	\hline
	d\ltens(c\lpar (a\ltens b))
	\limp
	\pfour\connn{b,a,d,c}
	&
	\pfour\connn{b,a,d,c}\lpar   \cneg d\lpar (\cneg c\ltens (\cneg a \lpar \cneg b))
	\\
	d\ltens(c\lpar (a\ltens b))
	\limp
	\pfour\connn{a,b,d,c}
	&
	\pfour\connn{a,b,d,c}\lpar  \cneg d\lpar (\cneg c\ltens (\cneg a \lpar \cneg b))
	\\	\hline
	a\ltens(c\lpar (b\ltens d))
	\limp
	\pfour\connn{c,a,d,b}
	&
	\pfour\connn{c,a,d,b}\lpar   \cneg a\lpar (\cneg c\ltens (\cneg b \lpar \cneg d))
	\\
	a\ltens(c\lpar (b\ltens d))
	\limp
	\pfour\connn{c,a,b,d}
	&
	\pfour\connn{c,a,b,d}\lpar  \cneg a\lpar (\cneg c\ltens (\cneg b \lpar \cneg d))
	\\	\hline
	c\ltens(a\lpar (b\ltens d))
	\limp
	\pfour\connn{a,c,b,d}
	&
	\pfour\connn{a,c,b,d}\lpar   \cneg c\lpar (\cneg a\ltens (\cneg b \lpar \cneg d))
	\\
	c\ltens(a\lpar (b\ltens d))
	\limp
	\pfour\connn{a,c,d,b}
	&
	\pfour\connn{a,c,d,b}\lpar   \cneg c\lpar (\cneg a\ltens (\cneg b \lpar \cneg d))
	\\	\hline
	a\ltens(d\lpar (b\ltens c))
	\limp
	\pfour\connn{c,b,a,d}
	&
	\pfour\connn{c,b,a,d}\lpar   \cneg a\lpar (\cneg d\ltens (\cneg b \lpar \cneg c))
	\\
	a\ltens(d\lpar (b\ltens c))
	\limp
	\pfour\connn{b,c,a,d}
	&
	\pfour\connn{b,c,a,d}\lpar  \cneg a\lpar (\cneg d\ltens (\cneg b \lpar \cneg c))
	\\	\hline
	d\ltens(a\lpar (b\ltens c))
	\limp
	\pfour\connn{a,d,c,b}
	&
	\pfour\connn{a,d,c,b}\lpar   \cneg d \lpar (\cneg a\ltens (\cneg b \lpar \cneg c))
	\\
	d\ltens(a\lpar (b\ltens c))
	\limp
	\pfour\connn{a,d,b,c}
	&
	\pfour\connn{a,d,b,c}\lpar   \cneg d \lpar (\cneg a\ltens (\cneg b \lpar \cneg c))
\end{array}
\end{equation}
\end{prop}
\begin{proof}
	The proof of $(d\lpar(b\ltens(a\lpar c)))\limp\pfour\connn{c,a,d,b}$ is shown in \eqref{eg:proofA}.
	The other implications are proven similarly.
\end{proof}

Both $\MLLgen$ (by definition) and $\GS$ (by Corollary~\ref{cor:MLL}) are {conservative} extensions of $\MLL$.
The non $\ltens/\lpar$-decomposable $4$-ary connective $\gcon$  should  be translated as a non-decomposable graph with $4$ vertices, that is, a $\pfour$.
Thus, the translation of $\gcon(a,b,c,d)$ should be of the shape $\pfour\connn{\sigma(a),\sigma(b),\sigma(c),\sigma(d)}$ for some permutation $\sigma$ over the set $\set{a,b,c,d}$.
However for each possible translation of $\gcon$, one of the  implications in  \eqref{eq:non-provableMLLG}, which are not derivable in $\MLLgen$ and contain $\gcon$,
is translated into a $\GS$-derivable implication.
We can conclude that the systems $\MLLgen$ and $\GS$ are not equivalent.

\subsection*{Symmetries of a connective}

\newcommand{\fourraryConGate}[5]{
\!\raisebox{-25pt}{	\begin{tikzpicture}[x=7, y=7,font=\scriptsize]
		\node (gateout) at (0,0) {};
		\node (out) at (0,-1) {};
		\node at (0,1) {$#5\strut$};
		\node (v2) at (-2.5,2) {};	
		\node (v3) at (2.5,2) {};
		\draw  (gateout.center) edge (out);
		\draw (gateout.center) -- (v2.center) -- (v3.center) -- (gateout.center);
		\node (gatein1) at (-1.5,2) {};
		\node (gatein2) at (-0.5,2) {};
		\node (gatein3) at (0.5,2) {};
		\node (gatein4) at (1.5,2) {};
		\node[inner sep=0] (in1) at (-1.5,4) {$a\strut$};
		\node[inner sep=0] (in2) at (-0.5,4) {$b\strut$};
		\node[inner sep=0] (in3) at (0.5,4) {$c\strut$};
		\node[inner sep=0] (in4) at (1.5,4) {$d\strut$};
		\draw  (in1) edge[looseness=1, in=90,out=-90] (gatein#1.center);
		\draw  (in2) edge[looseness=1, in=90,out=-90] (gatein#2.center);
		\draw  (in3) edge[looseness=1, in=90,out=-90] (gatein#3.center);
		\draw  (in4) edge[looseness=1, in=90,out=-90] (gatein#4.center);
	\end{tikzpicture}
}\!
}
\newcommand{\gcongate}[4]{\fourraryConGate{#1}{#2}{#3}{#4}{\gcon}}
\newcommand{\gcongatep}[4]{\fourraryConGate{#1}{#2}{#3}{#4}{\gcon'}}
\newcommand{\gcongates}[4]{\fourraryConGate{#1}{#2}{#3}{#4}{\gcon''}}
\newcommand{\Ccongate}[4]{\fourraryConGate{#1}{#2}{#3}{#4}{\Ccon}}
\newcommand{\nCcongate}[4]{\fourraryConGate{#1}{#2}{#3}{#4}{\cneg\Ccon}}

\def\symgroup{\mathfrak S}
\def\symgroupof#1{\symgroup(#1)}

We now calculate the symmetries of the connectives $\gcon$ and $\pfour$ and show that they are different in a way such that $\GS$ and $\MLLgen$ are not comparable as logical systems. To understand what we mean by symmetries of a connective, consider for example the $\ltens $, which is commutative---as the formula $A\ltens B$ is logically equivalent to $B\ltens A$---that is, we can permute $A$ and $B$ without changing the formula.
On the other hand, the connective $\limp$ is not commutative.
But we can define a connective $\lcoimp$ distinct from $\limp$ but definable from $\limp$ as $A\lcoimp B=B\limp A$.

Once we start to analyse $n$-ary connectives with $n>2$, we can no more categorise  connectives only as {commutative} and {non-commutative}.
Consider the  $4$-ary connectives $\gcon$  and $\ngcon$ and their describing sets of partitions in \ref{eq:girCon}.
We have
$$
\begin{array}{r@{\;=\;}c@{\;=\;}c@{\;=\;}l}
	\gcon(a,b,c,d)&
	\gcon(b,a,c,d)&
	\gcon(a,b,d,c)&
	\gcon(b,c,d,a) =
	\\
	\gcon(c,d,a,b)&
	\gcon(d,a,b,c)&
	\gcon(b,a,d,c)&
	\gcon(d,c,b,a)
\end{array}
$$
which, using interaction net  syntax \cite{lafont:95},  could be represented as follows
$$
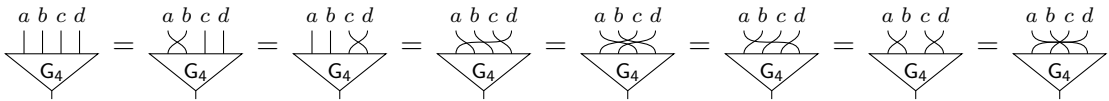
\begin{array}{c@{=}c@{=}c@{=}c@{=}c@{=}c@{=}c@{=}c}
\gcongate{1}{2}{3}{4}&
\gcongate{2}{1}{3}{4}&
\gcongate{1}{2}{4}{3}&
\gcongate{2}{3}{4}{1}&
\gcongate{3}{4}{1}{2}&
\gcongate{4}{1}{2}{3}&
\gcongate{2}{1}{4}{3}&
\gcongate{4}{3}{2}{1}
\end{array}
$$
Hence, the sets of partitions describing $\gcon$ and $\ngcon$ are \emph{stable} under the permutation\footnote{We write permutations using the \emph{cycle notation}, that is, the permutation is written as a product of cycles, and each element is mapped to the following element in the same cycle.} in the set $\symgroupof{\gcon}$, shown below:
$$
\symgroupof{\gcon}=\symgroupof{\ngcon}=\set{
	(1),
	(1,2),
	(3,4),
	(1,2,3,4), (1,3)(2,4),(1,4,2,3),
	(1,2)(3,4),
	(1,4)(2,3)
}
$$

We conclude that there are $\frac{\sizeof{\symgroup_4}}{\sizeof{\symgroupof{\gcon}}}=3$  non-isomorphic instances of $\gcon$, defining as many $4$-ary connectives:
$$
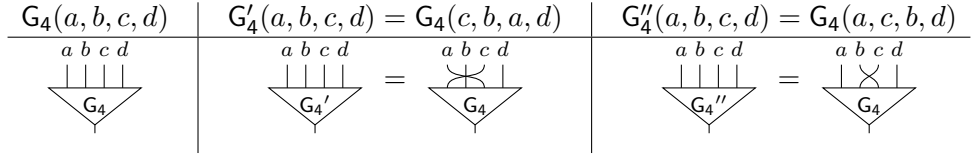
\begin{array}{c@{\quad}|@{\quad}c@{\quad}|@{\quad}c}
\gcon(a,b,c,d)
&
\gconp(a,b,c,d)=\gcon(c,b,a,d)
&
\gconpp(a,b,c,d)=\gcon(a,c,b,d)
\\
\hline
\gcongate{1}{2}{3}{4}
&
\gcongatep{1}{2}{3}{4}
=
\gcongate{3}{2}{1}{4}
&
\gcongates{1}{2}{3}{4}
=
\gcongate{1}{3}{2}{4}
\end{array}
$$
Moreover, we observe that $\ngcon$ cannot be expressed as a function of $\gcon$ since the two describing sets are made of partitions of different cardinality.
The connective $\gcon$ is a non-decomposable $4$-ary connective \cite{mai:19,acc:mai:20}. We can conclude that there are only three pairs of dual non-decomposable $4$-ary connectives-as-partitions.
In fact, a properly defined duality forces the dual of a set $P$ of partitions to be exactly the set all partitions that are orthogonal to the ones in $P$. This restrains the subsets of $\partof{4}$ which can be used to describe connectives, leaving only $6$ non-decomposable connectives.

At the same time,  $\pfour$ is the unique prime graph on four vertices. Hence, any non-decomposable $4$-ary connective-as-graph has to be an instance of $\pfour$.
Since the symmetry group of $\pfour$ (i.e., the group of its  isomorphisms) is the following
$$ 
\symgroupof{\pfour}=
\set{
(1),
(1,4)(2,3)
}
$$
we conclude that there are $\frac{\sizeof{\symgroup_4}}{\sizeof{\symgroupof{\pfour}}}=12$ different instances of $\pfour$ defining as many $4$-ary connectives:
$$
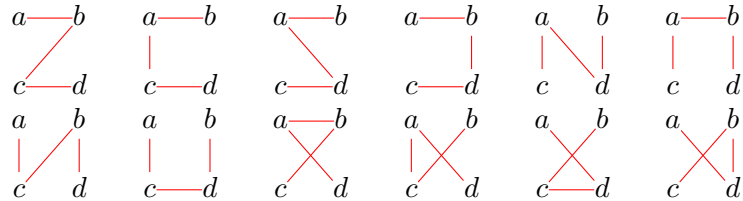
\begin{array}{c@{\quad }c@{\quad }c@{\quad }c@{\quad }c@{\quad }c}
\begin{array}{c@{\quad\;\;}c}
	\va1 & \vb1\\
	\\
	\vc1 & \vd1
\end{array}
\edges{a1/b1,b1/c1,c1/d1}
&
\begin{array}{c@{\quad\;\;}c}
	\va1 & \vb1\\
	\\
	\vc1 & \vd1
\end{array}
\edges{a1/b1,a1/c1,c1/d1}
&
\begin{array}{c@{\quad\;\;}c}
	\va1 & \vb1\\
	\\
	\vc1 & \vd1
\end{array}
\edges{a1/b1,d1/c1,a1/d1}
&
\begin{array}{c@{\quad\;\;}c}
	\va1 & \vb1\\
	\\
	\vc1 & \vd1
\end{array}
\edges{a1/b1,b1/d1,c1/d1}
&
\begin{array}{c@{\quad\;\;}c}
	\va1 & \vb1\\
	\\
	\vc1 & \vd1
\end{array}
\edges{a1/c1,a1/d1,b1/d1}
&
\begin{array}{c@{\quad\;\;}c}
	\va1 & \vb1\\
	\\
	\vc1 & \vd1
\end{array}
\edges{a1/b1,b1/d1,c1/a1}
\\
%
\begin{array}{c@{\quad\;\;}c}
	\va1 & \vb1\\
	\\
	\vc1 & \vd1
\end{array}
\edges{a1/c1,b1/c1,b1/d1}
&
\begin{array}{c@{\quad\;\;}c}
	\va1 & \vb1\\
	\\
	\vc1 & \vd1
\end{array}
\edges{a1/c1,b1/d1,c1/d1}
&
\begin{array}{c@{\quad\;\;}c}
	\va1 & \vb1\\
	\\
	\vc1 & \vd1
\end{array}
\edges{a1/b1,b1/c1,a1/d1}
&
\begin{array}{c@{\quad\;\;}c}
	\va1 & \vb1\\
	\\
	\vc1 & \vd1
\end{array}
\edges{a1/d1,b1/c1,c1/a1}
&
\begin{array}{c@{\quad\;\;}c}
	\va1 & \vb1\\
	\\
	\vc1 & \vd1
\end{array}
\edges{a1/d1,d1/c1,c1/b1}
&
\begin{array}{c@{\quad\;\;}c}
	\va1 & \vb1\\
	\\
	\vc1 & \vd1
\end{array}
\edges{a1/d1,b1/c1,b1/d1}
\end{array}
$$
If we denote $\Zcon\connn{a,b,c,d}=\pfour\connn{a,b,c,d}$ and $\Ncon\connn{a,b,c,d}=\pfour\connn{c,a,d,b}$ (the first and second-to-last graph in the first line in the above figure), we have that $\cneg \Zcon=\Ncon$.
More generally, if $\Ccon\connn{a,b,c,d}=\pfour\connn{\sigma(a),\sigma(b),\sigma(c),\sigma(d)}$ for a permutation $\sigma$ of the set $\set{a,b,c,d}$, then 
$$
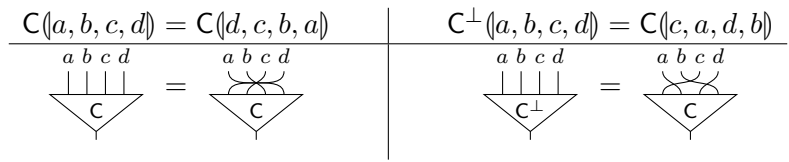
\begin{array}{c@{\qquad}|@{\qquad}c}
\Ccon\connn{a,b,c,d}=\Ccon\connn{d,c,b,a}
&
\cneg \Ccon\connn{a,b,c,d}=\Ccon\connn{c,a,d,b}	
\\
\hline
\Ccongate{1}{2}{3}{4}
=
\Ccongate{4}{3}{2}{1}
&
\nCcongate{1}{2}{3}{4}
=
\Ccongate{2}{4}{1}{3}
\end{array}
$$
We conclude that there are $6$ pairs of dual non-decomposable $4$-ary connectives-as-graphs.

Summing up, the symmetry group of $\gcon$ is different from the one of $\pfour$. Hence, the two connectives have to be distinct.
Moreover $\ngcon( a,  b, c, d)$ cannot be expressed as $\gcon(\sigma(a),\sigma(b),\sigma(c),\sigma(d))$ for any $\sigma $ permutation over $\set{a,b,c,d}$, while $\pfour\connn{a,b, c, d}=\cneg\pfour\connn{ b,  d ,  a, c}$.
In $\MLLgen$ there are $3$ pairs of dual non-decomposable $4$-ary connectives, in $\GS$ there are $6$ of such pairs.

\subsection*{Decomposable and non-decomposable connectives.}

In the connectives-as-partitions setting, a generalised $n$-ary connective described by a partition $P\subset\partofn$ is \emph{decomposable} if there is a $\ltens/\lpar$-formula $F$ such that $P$ is 
the set of partitions associated to all possible derivations of $F$ \cite{danos:regnier:89,acc:mai:20}.
For example, the connective described by $P=\set{\partition{\block{1,2},\block 2}, \partition{\block 1, \block{2,3}}}$ 
is decomposable since it is associated to the formula $F=(a\ltens b) \lpar c$ as shown above. 
However, the generalised connective which corresponds  to the formula $\gcon(1,2,3,\gcon(4,5,6,7))$ is considered to be non-decomposable, even if we define it by using other connectives.
A possible motivation for such a choice may be the lack of an efficient decomposition algorithm in the literature.
In fact, even from the set of partitions defined by a $\ltens/\lpar$-formula, it is not trivial to reconstruct the original formula.

However, for our connectives-as-graphs setting,
Lemma~\ref{lem:newhope} provides a finer definition of connective decomposition. 
In particular, every prime graph defines a non-decomposable connective, and every connective is either non-decomposable or admits a decomposition into non-decomposable ones.
That is, while  the notion of decomposition for connectives-as-partitions is still rough, the one for connectives-as-graph is well studied and comes with a linear decomposition algorithm~\cite{LinearModularDecomposition}.

\subsection*{The packaging problem}

The so called \emph{packaging problem}~\cite{danos:regnier:89} is the impossibility of deriving the axiom $\gcon \limp \gcon$ in the sequent calculus. This is due to the lack of the \emph{initial coherence} property  \cite{avr:canonical:01,mil:pim:13} for the $\MLLgen$ sequent calculus.
This can easily be shown in the case of $\gcon$ by remarking that the arities of the rules for $\gcon$ and $\ngcon$ make it impossible to gather together the premises as four smaller proofs.
The initial coherence property can be recovered using the proof net syntax of $\MLLgen$ \cite{acc:mai:20}.

As we have seen in Corollary~\ref{cor:identity}, the proof system $\GS$ has the initial coherence property, that is, there is no packaging problem in $\GS$.


\section{Related and Future Work}\label{sec:relatedWork}

In this section we would like to draw attention to certain challenges and open problems surrounding $\GS$.
Using examples, such as~(\ref{struct:deep}) and (\ref{struct:deep2}),
we have already explained why $\GS$ necessarily demands deep inference.
Let us now explain that 
simply taking an established semantics for $\MLLm$ based on graphs and dropping the restriction to cographs does not immediately yield a semantics for $\GS$.

\subsection{Linear inferences in classical logic.}

The switch rule has the property that it reflects edges and maximal cliques.
That is: if there is an edge in the conclusion it will also appear in the premise and every maximal clique in the premise is a superset of some maximal clique in the conclusion.
Indeed, mappings reflecting maximal cliques and preserving stable sets 
(mutually independent vertices) have a long history in
program semantics~\cite{Berry1978} which led to coherence spaces and the discovery of linear logic~\cite{girard:87}.
Therefore it is a reasonable starting point to try generalising switch by using such maximal clique reflecting homomorphisms, instead of $\ssdr$.
Indeed this is how we discovered $\ssdr$, which is sound with respect to such homomorphisms.

Unfortunately, replacing $\ssdr$ with maximal clique reflecting homomorphisms yields a system distinct from our graphical system, for example the following would be provable, but is not provable in $\GS$.
\begin{equation}\label{eg:eg}
  \begin{array}{c@{\quad\;\;}c@{\quad\;\;}c@{\quad\;\;}c}
    & \va1 & \vb1 &  \\
    \\[-1ex]
    \vna1 & & & \vnb1  \\
    \\[-1ex]
    & \vc1 & \vnc1 &
  \end{array}
  \edges{a1/b1,na1/c1,nc1/nb1,a1/c1,b1/nc1,na1/nb1}
\end{equation}
We may try replacing both $\ssdr$ and $\ssur$ using a stronger symmetric notion of homomorphism where, in addition, every maximal stable set in the conclusion is a superset of some maximal stable set in the premise.
Homomorphisms that are both maximal clique reflecting and stable set preserving
are special \textit{linear inferences}~\cite{das:str:RTA15,das:str:LMCS16} that preserve edges from the conclusion in the premise of a rule, which we refer to as \textit{linear homomorphisms} for convenience.
Using linear homomorphisms, the above example is not provable.
To see why, observe that at some point either $a$ and $\cneg a$ or $b$ and $\cneg{b}$ 
must be brought together into a module where they can interact, but this cannot be achieved while preserving the maximal stable set $\left\{a,\cneg{b}\right\}$.

Notice, however, that if we made the alternative design decision of replacing $\ssdr$ and $\ssur$ by linear homomorphisms described above, the implications below would be provable.
\begin{equation}\label{eg:Eiffel}
  \begin{array}{c@{\quad}c@{\quad}c}
\va1 && \ve1 \\[9pt]
\vb1 && \vd1 \\[9pt]
 &\vc1 &
  \end{array}
  \edges{a1/b1,b1/c1,c1/d1,d1/e1,b1/d1}
  \multimap 
  \begin{array}{c@{\quad}c@{\quad}c}
\va1 && \ve1 \\[9pt]
\vb1 && \vd1 \\[9pt]
 &\vc1 &
  \end{array}
  \edges{a1/b1,c1/d1,d1/e1,b1/d1}
\qquad
\qquad\qquad
\qquad
  \begin{array}{c@{\qquad}c}
 & \va1        \\[4pt]
      \vb1 \\[4pt]
 & \vc1        \\[4pt]
      \vd1 \\[4pt]
 & \ve1 
  \end{array}
  \edges{a1/b1,b1/c1,c1/d1,d1/e1,e1/c1}
  \multimap 
  \begin{array}{c@{\qquad}c}
 & \va1        \\[4pt]
     \vb1 \\[4pt]
 & \vc1        \\[4pt]
      \vd1 \\[4pt]
 & \ve1 
  \end{array}
  \edges{a1/b1,b1/c1,c1/d1,d1/e1}
\end{equation}
In contrast, the above examples are not provable in $\GS$, since both sides are
distinct prime graphs; and there is no suitable way to apply $\ssdr$.
Thus, we would obtain a distinct system from $\GS$ by using such
linear homomorphisms.

The above examples separating linear homomorphisms from $\GS$ made use of graphs with at least 5 nodes that are not cographs.
The smallest examples involving cographs only that separate $\MLLm$ from logics based on linear homomorphisms involve formulas with~8 nodes.\footnote{A new linear inference of size 8: \url{https://prooftheory.blog/2020/06/25/new-linear-inference/}
and
Linear inferences of size 7: \url{https://prooftheory.blog/2020/10/01/linear-inferences-of-size-7/}}
Indeed, studying logics defined using linear inferences that reflect maximal
cliques and preserve maximal stable sets is currently a topic of
active research and leads to possible extensions of Boolean logic to
graphs~\cite{%
	calk:graph, waring:master, CDW:ext-bool, das:19, das:rice:FSCD2021%
}.

Logics resulting from directly generalising linear inferences inspired by Boolean logic are incomparable to $\GS$. In one direction, the above examples show there are graphs that are not provable in $\GS$ that are provable using linear inferences; while in the converse direction the $\pdr$ rule is not admissible in a system with linear inferences. To see why, observe that in Example~\ref{eg:proofN} there is a clique $\left\{\cneg{a},\cneg{d},c,b\right\}$ in the premise of the $\pdr$ that is not reflected in the conclusion. Since $\pdr$ is necessary for examples such as in Figure~\ref{fig:eg}, we know that the logics are incomparable. Hence, although these proposed logics on graphs appear to be close to~$\GS$, they are founded on logical principles that are different from the ones we assume, as explained in detail in Appendix~\ref{sec:phil}.

\subsection{Criteria for proof nets.}
Graphical approaches to proof nets such as $R\&B$-graphs \cite{retore:03} have valid definitions when we drop the restriction to cographs.
However, we show that (at least without strengthening established criteria),
such definitions do not yield a semantics for a logic over graphs, since logical principles laid out 
in the introduction are violated.

Consider again graph~(\ref{eq:exa4}),
which is not provable in $\GS$.
In an $R\&B$-graph we draw blue edges representing the axiom links of proof nets, as shown below for graph~(\ref{eq:exa4}).
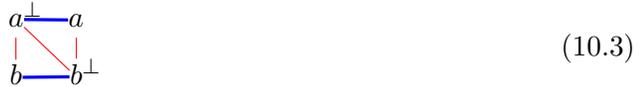
\begin{equation}\label{struct:5}
  \begin{array}{c@{\quad\;\;}c@{\quad\;\;}c@{\quad\;\;}c}
    \vna1 & \va1  \\
    \\[-1ex]
    \vb1 & \vnb1 
  \end{array}
  \edges{na1/b1,a1/nb1,na1/nb1}
  \intedges{na1/a1,nb1/b1}
\end{equation}
If we remove the restriction to cographs when we apply the established correctness criterion for $R\&B$-graphs the above graph would be accepted, which is something we aim to avoid for a semantics of $\GS$.
The reason is the cycle of 4 vertices alternating between red and blue edges has a chord.
This observation is independent of the rules of the system $\GS$, as we observed in Section~\ref{sec:system} that graph~(\ref{eq:exa4}) cannot be provable in a system subject to the logical principle of consistency.
Future work includes defining a stronger criterion for $R\&B$-graphs that works for graphs that are not cographs, with the aim of obtaining a sound and complete semantics for $\GS$.

\subsection{Beyond formulas-as-processes}

As mentioned in the introduction, this study originates from the limited expressive power of formulas in the formulas-as-processes interpretation.
In order to study the challenge of moving from formulas to graphs we have restricted ourselves to a minimal logic, whereas applications in concurrency typically require a logic with more features such as 
non-commutative operators for modelling the causal order of events~\cite{bruscoli:02,gug:SIS,retore:99}.
In the setting of graphical logics extending $\GS$, non-commutative operators
generalise to graphs incorporating directed edges which capture more general patterns of logical time,
as explained in~\cite{FSCD2022}.
To make such models useful we also require features such as:
additives for choices~\cite{miller:tiu:pi,ross:MAV}; 
quantifiers for dealing with message passing and private data~\cite{ross:MAV1,horne:19};
and even co-inductive proofs~\cite{Horne2020}.
Lifting such features to the setting of graphs seem feasible, but requires a certain amount of future work.

To hint at further possibilities enabled by $\GS$ in the direction of concurrency theory,
we return briefly to graph~\eqref{eq:ex2} mentioned in the introduction.
We explained there that an edge models separation, the absence of an edge models the ability to communicate,
and that we can express intransitive information flows using such graphs.
We elaborate a little more on such a scenario below.

Suppose that the processes $a$, $b$, $c$, and $d$ 
operate on three shared resources $x$, $y$, and $z$ such that
$c$ and $a$ operate (read and write) on  $x$;
	$a$ and $d$ operate on  $y$; and
	$d$ and $b$ operate on  $z$.
In addition, in this concurrency scenario, we assume that the information flow is strictly managed
such that, although information can flow via $x$ from $c$ to $a$, and from $a$ to $d$ via $y$,
those information flows must be strictly separated, and therefore $a$ must be trusted or designed such that the information from $c$ does not flow to $d$ indirectly via $a$. More concretely, an illegal indirect information flow
caused by $a$ failing to manage correctly the separation between $c$ and $d$ 
 may be achieved by the following series of read and write operations:
assuming $\mathsf{secret}$  is known initially by $c$ only, $c$ writes $\mathsf{secret}$ to $x$, $a$ reads $\mathsf{secret}$ from $x$ and then writes $\mathsf{secret}$ to $y$, and then $d$ reads $\mathsf{secret}$ from $y$.
Similarly, $d$ must be careful to ensure no information is leaked between $b$ and $a$, despite having a shared channel with both.
Notice that since we should also avoid information flows between $b$ and $c$,
both $d$ and $a$ must trust each other to adhere to policy in order to manage their own policy obligations.

This is a realistic information flow scenario. For example $a$ and $d$, may represent actions of service providers (e.g., a reviewer and editor) that make decisions based on confidential information from multiple parties (authors, reviewers, editors, and publishers) without disclosing confidential information to the other parties in a chain of conflicts of interest, and also without mutually sharing all information.
From this small example, we see already that intransitive information flows arise naturally in networks where privacy is important; that they model complex decentralised trust scenarios; and hence, we anticipate further work in this direction.

\def\Sfive{\mathsf{S5}}

\section{Conclusion}\label{sec:conclusion}

Guided by logical principles, we have devised a minimal proof system (called $\GS$ and shown in Figure~\ref{fig:SGS}) 
that operates directly over graphs, rather than formulas.
Negation is given by graph duality, while disjunction is disjoint union of graphs, allowing us to define the implication ``$\gG$~implies~$\gH$'' as the standard ``not $\gG$ or $\gH$'' (see Definition~\ref{def:dual}). All other design decisions are then fixed by our guiding logical principles (Appendix~\ref{sec:phil}).
Most of these principles follow from cut elimination (Theorem~\ref{thm:cut}), to which the majority of this paper is dedicated.
We also confirm that $\GS$ conservatively extends $\MLLm$ (Corollary~\ref{cor:MLL}) --- a logic at the core of many proof systems.

Surprisingly, even for such a minimal generalisation of logic to graphs, deep inference is necessary.
Proof systems for classical logic, intuitionistic logic, linear logic, and many other logics \textit{may} be expressed using deep inference, but deep inference is generally not necessary, as presentations in the sequent calculus do exist.
In contrast, for some logics (e.g., $\BV$~\cite{gug:SIS,tiu:SIS-II} and modal logic $\Sfive$~\cite{Stouppa2007,Poggiolesi2008}),
deep inference is necessary in order to define a proof system satisfying cut elimination.
System $\GS$ goes further than the aforementioned systems in that all intermediate lemmas such as splitting (Lemmas~\ref{lem:splitting:prime} and \ref{lem:splitting:atom}) and context reduction (Lemma~\ref{lem:conred})
demand a deep formulation that is more context aware than standard. As such we were required to generalise the basic mechanisms of deep inference itself in order to establish cut elimination (Theorem~\ref{thm:cut}) for a logic over graphs.
This need for deep inference for $\GS$ is due to a property of general graphs that is forbidden for graphs corresponding to formulas --- that the shortest path between any two connected vertices may be greater than two; and hence, when we apply reasoning inside a module (i.e., a context), there may exist dependencies that indirectly constrain the module.

\subsection*{Acknowledgements}
We are very grateful for insightful discussions with Anupam Das, particularly when exploring relationships between $\GS$ and linear inferences in 
classical logic. We would also like to thank the anonymous referees who made valuable suggestions for improving the content and quality of this article.

\renewcommand{\addcontentsline}[3]{}

\bibliographystyle{alphaurl}
\bibliography{LBFref}

\appendix


\section{Proofs of Splitting and Context Reduction}\label{sec:splittingproofs}

The following simple lemma is used in several places of the main proof:

\begin{lemma}\label{lem:context}
	Let $C_1\coonso{R_1},\ldots,C_n\coonso{R_n}$ be contexts. If
	$\proves[\GS]{C_i}$ for all $i\in\set{1,\ldots,n}$, then
	$\proves[\GS]{C_1\coons{ C_{2}\coons{\ldots C_n\coonso{R_n}}{R_{2}}}{R_1}}$.
\end{lemma}

\begin{proof}
	We proceed by induction on $n$. The base case for $n=1$ is trivial, and the inductive case for $n>1$ is this derivation:
	\begin{equation*}
		\small
		\ods{\unit}{\dD_1}{\gC_1\Coons{\ods{\unit}{\dD'}{C_{2}\coons{ \ldots C_n\coonso{R_n}}{R_{2}}}{}}{R_1}}{}
		\;,
	\end{equation*}
	where $\dD_n$ is the derivation for $C_n$ and $\dD'$ exists by
	induction hypothesis.
\end{proof}

Let us now restate and prove the main lemmas of Section~\ref{sec:splitting}.

\lemSplitPrime*
\begin{proof}
	
	We prove the two cases simultaneously, by induction on the size of the graph $\gG\lpar \gP\connn{\gM_1, \dots, \gM_n}$, as defined in  Definition~\ref{def:size}. Let $\dD$ be the given derivation of $\gG\lpar \gP\connn{\gM_1, \dots, \gM_n}$.

	We aim to construct 
	$\gC\coonso{R}$, $\dD_\gG$, $\dD_\gC$, 
	and 
	either $\gK_i$, $\dD_i$  for all $i\in \set{1,\dots, n}$, 
	or $\gK_X$, $\gK_Y$, $\dD_X$, $\dD_Y$, 
	satisfying the condition of the statement.
	We make a case analysis on the bottommost rule instance 
	in the derivation $\dD$.
	
	We now proceed, systematically exhausting the cases (a) to (e) described in Section~\ref{sec:splitting}.
	To avoid redundancy in the proof,  we assume that whenever we define the context of the shape
	$\gC={C_n\coons{\ldots C_2\coons{C_1\coons{\cdot}{R_1}}{R_2}\ldots}{R_n}}$ 
	using some derivable contexts  $C_1\coonso{R_1},\ldots,C_n\coonso{R_n}$, 
	then we have a derivation of $\gC$ defined by Lemma~\ref{lem:context}.
	
	\begin{enumerate}[(a)]      
		
		\item
		If rule $\rr$ acts inside $\gG$ or any $\gM_i$ with $i\in \set{1,\dots, n}$, then the derivation $\dD$ is of shape 
		\begin{equation*}
			\odv{\unit}
			{\dD'}{[\odn{\gG'}{\rr}{\gG}{};\gP\connn{\gM_1,\ldots,\gM_n}]}{\GS}
			\quor
			\ods{\unit}{\dD'}{\gG\lpar\gP\Connn{\gM_1,\ldots,\gM_{i-1},\odn{\gM_i'}{\rr}{\gM_i}{},\gM_{i+1},\ldots,\gM_n}}{\GS}
			\hskip-1em
		\end{equation*}
		for some $1\le i\le n$,
		and
		the size of the conclusion of $\dD'$ is smaller than the one of $\dD$.
		We apply the induction hypothesis 
		and conclude immediately by adding the corresponding application of $\rr$ to $\dD_G$ or $\dD_i$ respectively.

		{
			We have a special case for $\rr=\aidr$ where $\gM_i'=\unit$ and 
			\begin{equation*}
				\small
				\ods{\unit}{\dD'}{\gG\lpar\gP\Connn{\gM_1,\ldots,\gM_{i-1},\odn{\unit}{\aidr}{\cneg a\lpar a}{},\gM_{i+1},\ldots,\gM_n}}{\GS}
				\hskip-2em
			\end{equation*}
			This is a base case of the induction,
			and it is of the form of case \ref{prime:B} where $\gC\coonso{R}=\gK_X=\unit$, and $\gK_Y=\gG$.
		}	
		
		%
		

		\item\label{split:b}
		
		If $\gG=\gG'\lpar\gG''$ with $\gG'\neq \unit $ and $\dD$ is of shape
		\begin{equation*}
			\scalebox{.95}{%
				\odn{
					\ods{\unit
					}{\dD'}{
						\gG''\lpar \gP\connn{\gM_1,\ldots,\gM_n}\coons{\gG'}{R_P}
					}{\GS}}{
					\ssdr}{
					\gG''\lpar\gG'\lpar\gP\connn{\gM_1,\ldots,\gM_n}}{}
			}\hskip-3em
		\end{equation*}
		The graph 
		$\gP\connn{\gM_1,\ldots,\gM_n}\coons{\gG'}{R_P}$ cannot be an atom since $\gG'$ and $\gM_i$ are non-empty, or a par of two graphs, due to conditions on $\ssdr$.
		Then by  Lemma~\ref{lem:newhope}, 
		$\gP\connn{\gM_1,\ldots,\gM_n}\coons{\gG'}{R_P}$ is composed via a prime graph.
		We apply the inductive hypothesis.
		We have the three following possibilities: 
		\begin{enumerate}[(b.I)]
			
			\item\label{split:b.I}
			
			In this case $\gG$ moves inside some $\gM_i$. W.l.o.g., we can assume $i=1$,  
			and $\dD$ is of  the shape
			\begin{equation*}
				\odn{\ods{\unit}{
						\dD'}{
						\gG''\lpar\gP\connn{ \gM_1\coons{\gG'}{S},\gM_{2},\dots,\gM_n}}{\GS}}{
					\ssdr}{
					\gG''\lpar\gG'\lpar\gP\connn{\gM_1,\ldots,\gM_n}}{}
				\hskip-3em
			\end{equation*}
			
			We apply the induction hypothesis to $\dD'$ and get one of the following three sub-cases:

			\begin{enumerate}[(b.{I}.A), wide=0pt]
				
				
				\item 
				there is a context  $\gC'\coonso {R' }$ and there are graphs $\gL_1,\gK_{2},\dots,\gK_n$ such that there are the following derivations 
				\begin{equation*}
					\ods{\gC'\coons{\cneg\gP\connn{\gL_1,\gK_2,\dots,\gK_n}}{R'}}{\dD''_{G}}{\gG''}{\GS}
					\qomma
					\ods{\unit}{\dD_1'}{\gL_1\lpar\gM_1\coons{\gG'}{S}}{\GS}
					\qomma
					\ods{\unit}{\dD_i}{\gK_i\lpar\gM_i}{\GS}
					\quand
					\ods{\unit}{\dD_C'}{\gC'}{\GS}
					\hskip-3em
				\end{equation*}
				for all $i\in\set{2,\dots,n}$.
				We let $\gK_1=\gL_1\lpar\gG'$ and $\gC\coonso{R}=\gC'\coonso{R'}$. 
				Then we conclude since  $\dD_G$ and $\dD_1$ are the following derivations
				\begin{equation*}
					\small
					\odn{\gC\coons{\cneg\gP\connn{\gL_1\lpar\gG',\gK_2,\dots,\gK_n}}{R}}{
						\ssdr}{
						\ods{\gC\coons{\cneg\gP\connn{\gL_1,\dots,\gK_2,\dots,\gK_n}}{R}}{
							\dD''_{G}}{\gG''}{}
						\lpar\gG'}{}
					\quand
					\odn{
						\ods{\unit}{\dD_1}{\gL_1\lpar\gM_1\coons{\gG'}{S}}{}}{
						\ssdr}{
						\gL_1\lpar\gG'\lpar\gM_1}{}
					\hskip-4em
				\end{equation*}
				
				\item there is a context $\gC\coonso R$ and there are two graphs $\gL_X$ and $\gK_Y$ such that
				there are the following derivations 
				\begin{equation*}
					\small
					\ods{\gC\coons{\gL_X\lpar\gK_Y}{R}}{\dD_{G}'}{\gG''}{\GS}
					\qomma
					\ods{\unit}{\dD_X'}{\gL_X\lpar\gM_1\coons{\gG'}{S}}{\GS}
					\qomma
					\ods{\unit}{\dD_Y}{\gK_Y\lpar\gP\connn{\unit,\gM_2,\ldots,\gM_n}}{\GS}
					\quand
					\ods{\unit}{\dD_C}{\gC}{\GS}
					\hskip-5em
				\end{equation*}
				We let $\gK_X=\gL_X\lpar\gG'$.
				Then we conclude since  $\dD_G$ and $\dD_X$ are the following derivations
				\begin{equation*}
					\small
					\odn{\gC\coons{\gL_X\lpar\gG'\lpar\gK_Y}{R}}{
						\ssdr}{
						\ods{\gC\coons{\gL_X\lpar\gK_Y}{R}}{\dD_{G}'}{\gG''}{}
						\lpar\gG'}{}
					\quand
					\odn{
						\ods{\unit}{\dD_X'}{\gL_X\lpar\gM_1\coons{\gG'}{S}}{}}{
						\ssdr}{
						\gL_X\lpar\gG'\lpar\gM_1}{}
					\hskip-5em
				\end{equation*}

				\item there is a context  $\gC\coonso R$ and there are graphs $\gK_X$ and $\gL_Y$ such that
				\begin{equation*}
					\small
					\ods{\gC\coons{\gK_X\lpar\gL_Y}{R}}{\dD_{G}'}{\gG''}{\GS}
					\qomma
					\ods{\unit}{\dD_X}{\gK_X\lpar\gM_i}{\GS}
					\qomma
					\ods{\unit}{\dD_Y'}{\gL_Y\lpar\gP\connn{\gM_1\coons{\gG'},\gM_1,\ldots,\gM_{i-1},\unit,\gM_{i+1},\ldots,\gM_n}}{\GS}
					\quand
					\ods{\unit}{\dD_C}{\gC}{\GS}
				\end{equation*}
				for some $i\in\set{2,\dots, n}$. 
				We let $\gK_Y=\gG'\lpar\gL_Y$.
				Then we conclude since  $\dD_G$ and $\dD_Y$ are the following derivations
				\begin{equation*}
					\small
					\vlnostructuressyntax
					\odn{\gC\coons{\gL_X\lpar\gG'\lpar\gK_Y}{R}}{
						\ssdr}{
						\ods{\gC\coons{\gL_X\lpar\gK_Y}{R}}{\dD_{G}'}{\gG''}{}
						\lpar\gG'}{}
					\quand
					\odn{
						\ods{\unit}{\dD_Y'}{\gL_Y\lpar\gP\connn{\gM_1,\ldots,\gM_{i-1},\unit,\gM_{i+1},\ldots,\gM_j\coons{\gG'},\ldots,\gM_n}}{}}{
						\ssdr}{
						\gG'\lpar\gL_Y\lpar
						\gP\connn{\gM_1,\ldots,\gM_{i-1},\unit,\gM_{i+1},\ldots,\gM_n}}{}
					\;.
					\hskip-3em
				\end{equation*}
			\end{enumerate}

			\item \label{split:b.II}
			There is a special case when the prime graph $\gP$ is $\ltens$ induced by graph isomorphism, which has the effect of making tensor associative. 
			Assume $\gP=\ltens$ and  $\dD$ is of the following shape, where $\gM_1 = \gA$ and $\gM_2 = \gB$.
			\begin{equation*}
				\scalebox{.95}{%
					\odn{\ods{\unit}{
							\dD'}{
							\gG''\lpar (\gA\ltens\gB)\coons{\gG'}{S}}{\GS}}{
						\ssdr}{
						\gG''\lpar\gG'\lpar( A\ltens B)}{}
				}\hskip-3em
			\end{equation*}			
			with $(\gA\ltens\gB)\coons{\gG'}S\isom\gA''\ltens(\gA'\ltens\gB)\coons{\gG'}{S'}$
			for some $S'\subseteq S$, where $\gA\isom\gA''\ltens\gA'$ and
			$\gA''\neq\unit\neq\gA'$. 
			
			Notice that this is neither the case that $\gG'$ moves entirely inside $\gA$ or $\gB$ (hence Case~\ref{split:b.I} cannot be applied), nor is it the case that $\gG'$ is entirely a module outside the prime graph $\ltens$ connecting $\gA''$ and $(\gA'\ltens\gB)\coons{\gG'}{S'}$, (in which we could move forwards to Case~\ref{split:b.III}). Observe furthermore that such a situation can never occur when $\gP$ is not $\ltens$.

			Since $\gless{\gG''\lpar (\gA\ltens\gB)\coons{\gG'}{S}}{\gG''\lpar\gG'\lpar (A\ltens B)}$ we can apply the induction hypothesis in the form of Lemma~\ref{lem:splitting:tens}.
			Then there is a context $\gC'$ and there are graphs $\gK_A''$ and $\gK_Y'$ such that there are derivations
			\begin{equation*}
				\small
				\ods{\gC'\coons{\gK''_A\lpar\gK_Y}{R_1}}{\dD''_{G}}{\gG''}{\GS}
				\qomma
				\ods{\unit}{\dD_{A}''}{\gK''_A\lpar\gA''}{\GS}
				\qomma
				\ods{\unit}{\dD_Y}{\gK_Y\lpar(\gA'\ltens\gB)\coons{\gG'}{S'}}{\GS}
				\quand
				\ods{\unit}{\dD'_C}{\gC'}{\GS}
				\;.\hskip-3em
			\end{equation*}
			From $\dD_Y$ we get  that
			$\proves[\GS]{\gK_{Y} \lpar\gG'\lpar(\gA'\ltens\gB)}$ (via the rule $\ssdr$).
			
			It is important to observe at this point that we have the following inequality.
			\[
			\gsize{ \gK_{Y} \lpar\gG'\lpar(\gA'\ltens\gB) } < \gsize{ \gG \lpar (\gA\ltens\gB) }
			\]
			This is because, since $\gA''$ is non-empty and hence $\gsize{ \gA' } < \gsize{ \gA }$ since $\gA'$ has strictly less vertices, and also $\gsize{\gK_{Y} \lpar\gG'} \leq \gsize{\gG}$, by Observation \ref{obs:size}.
			Therefore, to the proof of $\proves[\GS]{\gK_{Y} \lpar\gG'\lpar(\gA'\ltens\gB)}$ we
			can apply the induction hypothesis to get a context
			$\gC''\coonso{R''}$ and $\gK'_{A}$ and $\gK_B$ such that there are derivations as follows.
			
			\begin{equation*}
				\small
				\ods{\gC''\coons{\gK_A\lpar\gK_B}{R''}}{\dD'_{G}}{\gK_Y\lpar\gG'}{\GS}
				\qomma
				\ods{\unit}{\dD'_{A}}{\gK'_{A}\lpar\gA'}{\GS}
				\qomma
				\ods{\unit}{\dD_B}{\gK_B\lpar\gB}{\GS}
				\quand
				\ods{\unit}{\dD''_{C}}{\gC''}{\GS}
				\;.\hskip-5em
			\end{equation*}
			We conclude by letting $\gC\coonso{R}=\gC'\coons{\gC''\coonso{R''}}{R'}$,
			$\gK_A=\gK_{A}''\lpar\gK_{A}'$, and  $\dD_G$ and $\dD_A$ be the derivations defined  as 
			\begin{equation*}
				\small
				\vlnostructuressyntax
				\odn{\gC'\Coons{
						\odn{\gC''\coons{\gK''_A\lpar\gK'_A\lpar\gK_B}{R''}}{
							\ssdr}{\gK''_A\lpar
							\ods{\gC''\coons{\gK'_A\lpar\gK_B}{R''}}{\dD'_{G}}{\gK_Y\lpar\gG'}{}}{}}{R'}}{
					\ssdr}{
					\ods{\gC'\coons{\gK''_A\lpar\gK_Y}{R'}}{\dD''_{G}}{\gG''}{}\lpar G'}{}
				\quand
				\odn{
					\ods{\unit}{\dD_{A}''}{\gK''_A\lpar\left(\gA''\ltens
						\ods{\unit}{\dD_{A'}}{\gK_{A'}\lpar\gA'}{}
						\right)}{}}{
					\ssdr}{
					\gK''_A\lpar\gK'_A\lpar(\gA''\ltens\gA')}{}
				\;.
				\hskip-7em
			\end{equation*}


			\item\label{split:b.III}
			
			This is the most involved sub-case of the proof, induced when the movement of $\gG'$ via the $\ssdr$ rule results in a larger prime graph, where the $\gG'$ does not overlap with any existing module and possibly adds edges that break up some existing modules of $\gP$ into smaller modules in the resulting prime graph.  
			Assume  $\dD$ is of shape
			\begin{equation*}
				\scalebox{.95}{%
					\odn{\ods{\unit}{
							\dD'}{
							\gG''\lpar\gQ\connn{\gG',\gN_2,\ldots,\gN_{k}}}{\GS}}{
						\ssdr}{
						\gG''\lpar\gG'\lpar\gP\connn{\gM_1,\ldots,\gM_n}}{}
				}\hskip-3em
			\end{equation*}
			for a unique prime graph $\gQ$ (up to permutation of modules) such that $k=\sizeof\vQ > \sizeof{\vP}$ (hence $k \geq 4$) 
			such that
			\begin{equation}\label{eq:ross2}
				\gQ\connn{\unit,\gN_2,\ldots,\gN_{k}} \isom \gP\connn{\gM_1,\ldots,\gM_n}
			\end{equation}
			where for each $i \in \set{2, \ldots k}$ we have that $\gN_i \neq \unit$ and $\gN_i$ is a module of some $\gM_j$ where $j \in \set{1, \ldots n}$.
			We apply the induction hypothesis to $\gG''\lpar\gQ\connn{\gG',\gN_2,\ldots,\gN_{k}}$ 
			and get one of the following three sub-cases:

			\begin{enumerate}[(b.{III}.A), wide=0pt]
				\item \label{split:b.III.A}
				there is a context $\gC'\coonso{R' }$ and $\gK'_1, \dots, \gK'_k$ such that there are derivations
				\begin{equation*}
					\ods{\gC'\coons{\cneg \gQ\connn{\gK_1', \dots,  \gK'_k}}{R'}}{\dD_\gG''}{\gG''}{\GS}
					\qomma
					\ods{\unit}{\dD_1'}{\gK'_1\lpar \gG'}{\GS}
					\qomma
					\ods{\unit}{\dD_i'}{\gK'_i\lpar \gN_i}{\GS}
					\quand
					\ods{\unit}{\dD_C'}{\gC'}{\GS}
					\hskip-1em
				\end{equation*}
				for all $i\in\intset2k$.
				Since the graphs $\gN_2, \dots, \gN_n$  are non-empty by hypothesis, we can apply Lemma~\ref{lem:g} and obtain a derivation
				\begin{equation*}
					\ods{
						\left(
						\ods{\unit}{\dD'_2}{\gK'_2\lpar\gN_2}{}
						\ltens\cdots \ltens
						\ods{\unit}{\dD'_k}{\gK'_k\lpar\gN_k}{}
						\right)
					}{}{
						\cneg{\gQ} \connn{\unit, \gK'_2, \dots, \gK'_k}
						\lpar 
						\gQ\connn{\unit,\gN_2,\ldots,\gN_{k}} 
					}{\mbox{\small Lemma \ref{lem:g}}}
				\end{equation*}	
				By (\ref{eq:ross2}), 
				such a derivation is a derivation of 
				$\cneg \gQ\connn{\unit,\gK_2,\ldots,\gK_{k}} \lpar  \gP\connn{\gM_1,\ldots,\gM_n}$
				on which we can now apply the inductive hypothesis.
				This gives the two following cases
				
				\begin{itemize}
					\item 
					
					either there is a context $\gC''\coonso{R''}$ and graphs $\gK_1, \dots, \gK_k $ such that there are derivations	
					\begin{equation*}
						\ods{\gC''\coons{\cneg \gP\connn{\gK_1, \dots, \gK_k}}{R''}}{\dD_\gQ}{\cneg{\gQ} \connn{\unit, \gK'_2, \dots, \gK'_k}}{\GS}
						\qomma
						\ods{\unit}{\dD_i}{\gK_i\lpar \gM_i}{\GS}
						\quand
						\ods{\unit}{\dD_C''}{\gC''}{\GS}
					\end{equation*}
					for all $i \in\intset1n$. Then we conclude by letting $\gC\coonso{R}=\gC'\coons{\gC''\coonso{R''}}{R'}$ and $\dD_G$ be the derivation
					\begin{equation*}
						\odn{
							\gC'\Coons{	\ods{\gC''\coons{\cneg \gP\connn{\gK_1, \dots, \gK_k}}{R''}}{\dD_\gQ}{
									\cneg{\gQ} \Connn{\ods{\unit}{\dD_1'}{\gK'_1\lpar \gG'}{}, \gK'_2, \dots, \gK'_k}
								}{}}{R'}
						}{\ssdr}{
							\ods{\gC'\coons{\cneg \gQ\connn{\gK_1', \dots,  \gK'_k}}{R'}}{\dD_\gG''}{\gG''}{}
							\lpar \gG'}{}
					\end{equation*}
					
					\item 
					or  there is a context $\gC''\coonso{R''}$ and graphs $\gK_X$ and $ \gK_Y $ such that, w.l.o.g. there are derivations	
					\begin{equation*}
						\ods{\gC''\coons{\gK_X \lpar \gK_Y }{R''}}{\dD_\gQ}{\cneg{\gQ} \connn{\unit, \gK'_2, \dots, \gK'_k}}{\GS}
						\qomma
						\ods{\unit}{\dD_X}{\gK_X\lpar \gM_1}{\GS}
						\qomma
						\ods{\unit}{\dD_Y}{\gK_Y\lpar \gP \connn{\unit, \gM_2, \dots, \gM_n}}{\GS}
						\quand
						\ods{\unit}{\dD_C''}{\gC''}{\GS}
						\hskip-6em
					\end{equation*}
					for an $i \in\intset1n$. 
					Then we conclude by letting $\gC\coonso{R}=\gC'\coonso{R'}$ and $\dD_G$ be the derivation
					\begin{equation*}
						\odn{
							\gC'\Coons{
								\ods{\gC''\Coons{\gK_X \lpar \gK_Y }{R''}}{\dD_\gQ}{
									\cneg{\gQ} \Connn{\ods{\unit}{\dD_1'}{\gK'_1\lpar \gG'}{}, \gK'_2, \dots, \gK'_k}}{}
							}{R'}
						}{\ssdr}{
							\ods{\gC'\coons{\gK'\lpar \gK_2 \lpar  \cdots\lpar \gK_n}{R'}}{\dD_\gG''}{\gG''}{}
							\lpar \gG'}{}
					\end{equation*}
					
				\end{itemize}

				\item 
				there is a context $\gC'\coonso{R' }$ and $\gK'_X$ and $\gK'_Y$ such that (after \ref{eq:ross2}) there are derivations
				\begin{equation*}
					\ods{\gC'\coons{{\gK_X'\lpar  \gK'_Y}}R}{\dD_\gG''}{\gG''}{\GS}
					\qomma
					\ods{\unit}{\dD_X'}{\gK'_X\lpar \gG'}{\GS}
					\qomma
					\ods{\unit}{\dD_Y'}{\gK'_Y\lpar \gP\connn{\gM_1, \dots, \gM_n}}{\GS}
					\quand
					\ods{\unit}{\dD_C'}{\gC'}{\GS}
					\hskip-6em
				\end{equation*}
				We can now apply inductive hypothesis on $\gK'_Y\lpar \gP\connn{\gM_1, \dots, \gM_n}$ and we have the two following cases:
				\begin{itemize}
					\item there is a context $\gC''\coonso{R''}$ and graphs $\gK_1, \dots, \gK_n$ such that there are derivations
					\begin{equation*}
						\ods{\gC''\coons{\cneg \gP\connn{\gK_1, \dots, \gK_n}}{R''}}{\dD_\gK}{\gK_Y'}{\GS}
						\qomma
						\ods{\unit}{\dD_i}{\gK_i\lpar \gM_i}{\GS}
						\quand
						\ods{\unit}{\dD_C''}{\gC''}{\GS}
						\hskip-1em
					\end{equation*}
					In this case we can conclude similarly to the first case of \ref{split:b.III.A}, that is, 
					by letting $\gC\coonso{R}=\gC'\coons{\gC''\coonso{R''}}{R'}$ 
					and $\dD_G$ be the derivation
					\begin{equation*}
						\odn{
							\gC'\Coons{	\ods{\gC''\coons{\cneg \gP\connn{\gK_1, \dots, \gK_n}}{R''}}{\dD_\gK}{\gK_Y'}{}
								\lpar 					
								\ods{\unit}{\dD_X'}{\gK'_X\lpar \gG'}{}
							}{R'} 
						}{\ssdr}{
							\ods{\gC'\coons{\gK_X' \lpar \gK_Y'}{R'}}{\dD_\gG''}{\gG''}{}
							\lpar \gG'}{}
					\end{equation*}
					
					\item
					or  there is a context $\gC''\coonso{R''}$ and graphs $\gK_X$ and $ \gK_Y $ such that, w.l.o.g. there are derivations	
					\begin{equation*}
						\ods{\gC''\coons{\gK_X \lpar \gK_Y }{R''}}{\dD_\gK}{\gK_Y'}{\GS}
						\qomma
						\ods{\unit}{\dD_X}{\gK_X\lpar \gM_1}{\GS}
						\qomma
						\ods{\unit}{\dD_Y}{\gK_Y\lpar \gP \connn{\unit, \gM_2, \dots, \gM_n}}{\GS}
						\quand
						\ods{\unit}{\dD_C''}{\gC''}{\GS}
						\hskip-6em
					\end{equation*}
					we conclude	by letting $\gC\coonso{R}=\gC'\coons{\gC''\coonso{R''}}{R'}$  and $\dD_G$ be the derivation
					\begin{equation*}
						\odn{
							\gC'\Coons{	\ods{\gC''\coons{\gK_X\lpar \gK_Y}{R''}}{\dD_\gK}{\gK_Y'}{}
								\lpar 					
								\ods{\unit}{\dD_X'}{\gK'_X\lpar \gG'}{}
							}{R'} 
						}{\ssdr}{
							\ods{\gC'\coons{\gK_X' \lpar \gK_Y'}{R'}}{\dD_\gG''}{\gG''}{}
							\lpar \gG'}{}
					\end{equation*}	
					
				\end{itemize}

				\item 
				There is a context  $\gC'\coonso{R'}$ and there are graphs $ \gK_X' $ and $ \gK_Y'$ such that for have the following derivations for some $\ell \in \set{2, \dots k}$.
				\begin{equation*}
					\small
					\ods{\gC'\coons{ \gK_X' \lpar \gK_Y'}{R'}}{\dD''_{G}}{\gG''}{\GS}
					\qomma
					\ods{\unit}{\dD_X'}{\gK_X' \lpar \gN_\ell}{\GS}
					\qomma
					\ods{\unit}{\dD_Y'}{\gK_Y' \lpar \gQ\connn{\gG', \gN_2, \ldots, \gN_{\ell-1}, \unit, \gN_{\ell+1}, \ldots \gN_k }}{\GS}
					\quand
					\ods{\unit}{\dD_\gC'}{\gC'}{\GS}
                                        \hskip-4em
				\end{equation*}
				There are two cases to consider: either $N_\ell = M_m$ for some $m \in \intset1k$ or we have $\gM_m = \gM'\coons{N_\ell}{R_m}$ for some non-empty $\gM'$.
				
				We consider first the former case where $N_\ell = M_m$ for some $m \in \intset1k$,
				the derivation $\dD_Y$ is defined, recalling that by~(\ref{eq:ross2}) we have $\gQ\connn{\unit,\gN_2,\ldots,\gN_{k}} \isom \gP\connn{\gM_1,\ldots,\gM_n}$
				and hence\\ $\gQ\connn{\unit, \gN_2, \ldots, \gN_{\ell-1}, \unit, \gN_{\ell+1}, \ldots \gN_k } \isom\gP\connn{\gM_1, \ldots, \gM_{m-1}, \unit, \gM_{m+1}, \ldots \gM_n }$, i.e.,
				\begin{equation*}
					\vlnostructuressyntax
					\odn{
						\ods{
							\unit
						}{\dD_Y'}{
							\gK_Y' \lpar \gQ\connn{\gG', \gN_2, \ldots, \gN_{\ell-1}, \unit, \gN_{\ell+1}, \ldots \gN_k }
						}{}
					}{\ssdr}{
						\gK_Y' \lpar \gG' \lpar \gP\connn{\gM_1, \ldots, \gM_{m-1}, \unit, \gM_{m+1}, \ldots \gM_n }
					}{}
				\end{equation*}
				We can conclude almost immediately by letting $\gC\coonso{R} = \gC'\coonso{R'}$, $\gK_Y = \gK_Y' \lpar \gG'$, $\gK_X = \gK_X'$, $\dD_X = \dD_X'$, $\dD_Y = \dD_Y'$, and $\dD_G$ be the following derivation.
				\begin{equation*}
					\odn{
						\gC'\coons{
							\gK_Y \lpar \gK_X
						}{R'} 
					}{\ssdr}{
						\ods{\gC'\coons{\gK_X' \lpar \gK_Y'}{R'}}{\dD_\gG''}{\gG''}{}
						\lpar \gG'}{}
					\quad .
				\end{equation*}

				Otherwise, 
				we pursue the case where $\gM_m = \gM'\coons{N_\ell}{R_m}$ for some non-empty $\gM'$, which is more involved than the case above.
				By (\ref{eq:ross2}) we have 
				\begin{equation*}
					\gQ\connn{\unit, \gN_2, \ldots, \gN_{\ell-1}, \unit, \gN_{\ell+1}, \ldots \gN_k } 
					\isom
					\gP\connn{\gM_1, \ldots, \gM_{m-1}, \gM', \gM_{m+1}, \ldots \gM_n }
                                        \hskip-6em
				\end{equation*}
				Hence, we have a proof of $\gK_Y' \lpar \gG' \lpar \gP\connn{\gM_1, \ldots, \gM_{m-1}, \gM', \gM_{m+1}, \ldots \gM_n }$ obtained by applying the rule $\ssdr$ to the conclusion of $\dD_Y'$ above, as follows.
				\begin{equation*}
					\vlnostructuressyntax
					\odn{
						\ods{
							\unit
						}{\dD_Y'}{
							\gK_Y' \lpar \gQ\connn{\gG', \gN_2, \ldots, \gN_{\ell-1}, \unit, \gN_{\ell+1}, \ldots \gN_k }
						}{}
					}{\ssdr}{
						\gK_Y' \lpar \gG' \lpar \gP\connn{\gM_1, \ldots, \gM_{m-1}, \gM', \gM_{m+1}, \ldots \gM_n }
					}{}
				\end{equation*}

				By Observation \ref{obs:size} $\gsize{\gK_Y' \lpar \gG'} \leq \gsize{\gG}$ and also $\gN_\ell \neq \unit$, hence we have the following inequality.
				\[
				\gsize{\gK_Y' \lpar \gG' \lpar \gP\connn{\gM_1, \ldots, \gM_{m-1}, \gM', \gM_{m+1}, \ldots \gM_n }} < \gsize{ \gG \lpar \gP\connn{\gM_1, \ldots, \gM_{n} } }
                                \hskip-6em
				\]
				Therefore,  we can apply the induction hypothesis to the above mentioned proof of $\gK_Y\lpar\gG' \lpar \gP\connn{\gM_1, \ldots, \gM_{m-1}, \gM', \gM_{m+1}, \ldots \gM_n }$
				to obtain one of the following three sub-cases:

				\begin{enumerate}[(b.{III}.{C}.A), wide=0pt]
					
					\item\label{split:b:II:C:A}
					There is a context  $\gC''$ and there are graphs $\gK_1$, $\dots$, $\gK_{m-1}$, $\gL$, $\gK_{m+1}$, $\dots$, $\gK_n$ such that there are derivations
					\begin{equation*}
						\ods{\gC''\coons{\cneg{\gP}\connn{\gK_1, \dots,\gK_{m-1}, \gL , \gK_{m+1}, \dots, \gK_n}}{R''}}{\dD_G'''}{	\gK_Y'\lpar \gG'}{\GS}
						\qomma
						\ods{\unit}{\dD'_m}{\gL \lpar \gM'}{\GS}
						\qomma
						\ods{\unit}{\dD_i}{\gK_i \lpar \gM_i}{\GS}
						\quand
						\ods{\unit}{\dD_\gC''}{\gC''}{\GS}
                                                \hskip-4em
					\end{equation*}
					for $i \in \intset1n$ such that $i\neq m$.
					
					To conclude we let $\gC\coonso{R} = \gC'\coons{\gC''\coonso{R''}}{R'}$ 
					(see Lemma~\ref{lem:context}), 
					and $\gK_m= \gK_X'\lpar  \gL$ 
					Then let $\dD_G$ 
					be defined as 
					\begin{equation*}
						\vlnostructuressyntax
						\small
						\odn{
							\gC'\Coons{
								\odn{
									\gC''\coons{\cneg{\gP}\connn{\gK_1, \hdots,
											\gK_{m-1}, \gK_X' \lpar \gL , \gK_{m+1}, \ldots, \gK_n}}{R''}
								}{\ssdr}{
									\gK_X'
									\lpar
									\ods{
										\gC''\coons{\cneg{\gP}\connn{\gK_1, \hdots, \gK_{m-1},\gL , \gK_{m+1}, \ldots, \gK_n}}{R''}
									}{\dD_G'''}{
										\gK_Y' \lpar \gG'
									}{}
									]
								}{}
							}{R'}
						}{\ssdr}{
							\ods{
								\gC'\coons{
									\gK_X'\lpar \gK_Y'
								}{R'}
							}{\dD_G''}{
								\gG''
							}{}
							\lpar \gG'
						}{}
					\end{equation*}
					and $\dD_m$ be defined as
					\begin{equation*}
						\odn{
							\ods{\unit
							}{\dD'_m}{
								[\gL ; \gM'\Coons{
									\ods{\unit}{\dD'_X}{\gK_X' \lpar N_\ell}{}
								}{R_m}]
							}{}
						}{\ssdr}{
							[\gK_X' ; \gL ; M'\coons{N_\ell}{R_m}]
						}{}
					\end{equation*}

					\item\label{split:b:II:C:B}
					there is a context  $\gC''\coonso{R''}$ and there are graphs $\gK_X $ and $\gK_Y''$ such that there are derivations
					\begin{equation*}
						\ods{\gC''\coons{\gK_X \lpar \gK_Y''}{R''}}{\dD_G''' }{\gK_Y'\lpar \gG'}{\GS}
						\qomma
						\ods{\unit}{\dD_X}{\gK_X \lpar M_1}{\GS}
						\qomma
						\ods{\unit}{\dD''_Y}{\gK_Y'' \lpar \gP\connn{\unit,\gM_2,  \ldots, \gM_{m-1}, \gM', \gM_{m+1}, \ldots \gM_n }}{\GS}
						\quand
						\ods{\unit}{\dD_\gC''}{\gC''}{\GS}
					\end{equation*}
					We conclude by setting $\gK_Y = \gK_X' \lpar \gK_Y''$  and $\gC\coonso{R} = \gC'\coons{\gC''\coonso{R''}}{R'}$, where
					the derivations $\dD_G$ is defined as 
					\begin{equation*}
						\odn{
							\gC'\Coons{[\gK_X' ;
								\ods{[\gK_X ; \gK_Y'']}{\dD_G'''}{[\gK_Y' ; \gG']}{}
								]}{R'}
						}{\ssdr}{[
							\ods{\gC'\coons{[\gK_X'; \gK_Y']}{R'}}{\dD_G''}{\gG''}{}
							; \gG'
							]}{}
					\end{equation*}
					and the derivation $\dD_Y$ is defined as (recall that $\gM'\coons{\gN_\ell}{R_m}\isom\gM_m$):
					\begin{equation*}
						\vlnostructuressyntax
						\odn{
							\ods{\unit}{\dD''_Y}{
								\gK_Y'' 
								\lpar  
								\gP\Connn{
									\unit, \gM_2\ldots, \gM_{m-1},
									\gM'\Coons{
										\ods{\unit}{\dD_X'}{\gK_X' \lpar \gN_\ell}{}
									}{R_m},
									\gM_{m+1}, \ldots \gM_n 
								}
							}{}
						}{\ssdr}{
							\gK_X' \lpar  \gK_Y'' \lpar \gP\connn{\unit, \gM_2, \ldots, \gM_{m-1},\gM'\Coons{\gN_\ell}{R_m}, \gM_{m+1}, 
								\ldots \gM_n }
						}{}
					\end{equation*}
					
					\item\label{split:b:II:C:C}
					there is a context  $\gC'$ and there are graphs $\gK_X'' $ and $ \gK_Y $ such that,  there are derivations 
					\begin{equation*}
                                          \scalebox{.9}{$
						\ods{\gC''\coons{\gK_X'' \lpar \gK_Y}{R''}}{\dD_G'''}{\gK_Y'\lpar \gG'}{}
						\qomma
						\ods{\unit}{\dD_X''}{\gK_X''  \lpar \gM'}{\GS}
						\qomma
						\ods{	\unit	}{\dD_Y}{\gK_Y \lpar \gP\connn{\gM_1, \gM_2,\dots, \gM_{m-1}, \unit, \gM_{m+1}, \ldots \gM_n }}{\GS}
						\quand
						\ods{\unit}{\dD_\gC''}{\gC''}{\GS}
                                                $}
                                          \hskip-4em
					\end{equation*}
					We conclude by letting $\gK_X = \gK_X' \lpar \gK_X''$  and $\gC\coonso{R} = \gC'\coons{\gC''\coonso{R''}}{R'}$, where the  derivations $\dD_\gG$ and $\dD_X$ are defined as follows
					\begin{equation*}
                                          \scalebox{.9}{$
						\odn{
							\gC'\Coons{
								\odn{\gC''\coons{\gK_X'\lpar \gK_X''\lpar \gK_Y}{R''}}
								{\ssdr}
								{\gK_X'\lpar 
									\ods{\gC''\coons{\gK_X'' \lpar \gK_Y}{R'}}{\dD_G'''}{\gK_Y'\lpar \gG'}{}}{}
							}{R''}
						}{\ssdr}{
							\ods{
								\gC'\coons{
									\gK_X'\lpar \gK_Y'
								}{R'}
							}{\dD_G'' }{
								\gG''
							}{}
							\lpar \gG'
						}{}
						\quand
						\odn{
							\ods{
								\unit
							}{{\dD_X''}}{
								\gK_X'' \lpar \gM'\Coons{\ods{\unit}{\dD_X'}{\gK_X' \lpar N_\ell}{}
								}{R_m}
							}{}
						}{\ssdr}{
							\gK_X'' \lpar \gK_X' \lpar \gM'\coons{\gN_\ell}{R_m}
						}{}
                                                $}
                                          \hskip-6em
					\end{equation*}
					
				\end{enumerate}

			\end{enumerate}

		\end{enumerate}

		\item\label{split:c}
		This is  the case where the $\ssdr$-rule is applied and $\dD$ is of shape
		\begin{equation}
                  \label{eq:split:c}
			\odn{\ods{\unit}{
					\dD'}{
					\gG\coons{\gP\connn{\gM_1,\ldots,\gM_n}}{S}}{}}{
				\ssdr}{
				\gG\lpar\gP\connn{\gM_1,\ldots,\gM_n}}{}
		\end{equation}
		We apply the induction hypthesis in the form of Lemma~\ref{lem:conred} 
		and get a graph $\gK$ and a
		context $\gC'\coonso R$ such that there are derivations
		\begin{equation*}
			\ods{\unit}{\dD^\ast}{\gK\lpar\gP\connn{\gM_1,\ldots,\gM_n}}{\GS}
			\qomma
			\ods{\gC'\coons{\gK\lpar\XXX}R}{\dD^{\XXX}_{\gG}}{\gG\coons{\XXX}S}{\GS}
			\quand
			\ods{\unit}{\dD_C'}{\gC'}{\GS}
		\end{equation*}
		for any graph $\XXX$. Since the application of $\ssdr$ in~\eqref{eq:split:c} is not trivial, the context $\gC'$ cannot be empty, and therefore we have $\gless{\gK}{\gG}$. 
		Hence we can apply the induction hypothesis to the proof $\dD^\ast$ of $\gK\lpar\gP\connn{\gM_1,\ldots,\gM_n}$, and obtain one of the following two cases:
		\begin{itemize}
			\item
			either there is a context $\gC''\coonso R$ 
			and graphs $\gK_1$, \dots, $\gK_n$, 
			such that 
			there are derivations
			\begin{equation*}
				\ods{\gC''\coons{\cneg\gP\connn{\gK_1,\ldots,\gK_n}}R}{\dD_\gK}{\gK}{\GS}
				\qomma
				\ods{\unit}{\dD_i}{\gK_i\lpar\gM_i}{\GS}
				\quand
				\ods{\unit}{\dD_C''}{\gC''}{\GS}
				\hskip-1em
			\end{equation*}
			for all $i\in\set{1,\ldots,n}$.
			Then we can let $\gC\coonso{R} = \gC'\coons{\gC''\coonso{R''}}{R'}$ and the derivation $\dD_\gG$ is: 
			\begin{equation*}
				\ods{\gC'\Coons{\ods{\gC''\coons{\cneg\gP\connn{\gK_1,\ldots,\gK_n}}{R''}}{\dD_\gK}{\gK}{}}{R'}}{\dD_\gG^\unit}{\gG}{}
			\end{equation*}

			\item
			or there is a context $\gC''\coonso R$ and graphs $\gK_X$ and $\gK_Y$ 
			and derivations
			\begin{equation*}
				\ods{\gC''\coons{\gK_X\lpar\gK_Y}{R''}}{\dD_\gK}{\gK}{\GS}
				\qomma
				\ods{\unit}{\dD_X}{\gK_X\lpar\gM_i}{\GS}
				\qomma
				\ods{\unit}{\dD_Y}{\gK_Y\lpar\gP\connn{\gM_1,\ldots,\gM_{i-1},\unit,\gM_{i+1},\ldots,\gM_n}}{\GS}
				\quand
				\ods{\unit}{\dD_C''}{\gC''}{\GS}
				\hskip-1em
			\end{equation*}
			for some $i\in\set{1,\ldots,n}$.
			Then we let $\gC\coonso{R} = \gC'\coons{\gC''\coonso{R''}}{R'}$ and $\dD_\gG$ is 
			\begin{equation*}
				\ods{\gC'\Coons{\ods{\gC''\coons{\gK_X \lpar \gK_Y}{R''}}{\dD'_\gK}{\gK}{}}{R'}}{\dD_\gG^\unit}{\gG}{}
			\end{equation*}

		\end{itemize}

		

		\item\label{split:d}
		Here we have
		$\gG=\gG'\lpar\cneg\gP\connn{\gN_1,\ldots,\gN_n}$ and $\dD$ is
		of shape
		\begin{equation*}
			\odn{\ods{\unit}{\dD'}{
					[\gG';([N_1;M_1];\cdots;[N_n;M_n])]}{}}{
				\pdr}{
				\gG'\lpar\cneg\gP\connn{\gN_1,\ldots,\gN_n}\lpar\gP\connn{\gM_1,\ldots,\gM_n}}{}
			\;.\hskip-3em
		\end{equation*}
		We apply the induction hypothesis in the form of Lemma \ref{lem:splitting:multitens}, which gives us 
		a context $\gC\coonso R$ and
		graphs $\gL_1,\ldots,\gL_n$ such that
		\begin{equation*}
			\ods{\gC\coons{\gL_1\lpar\cdots\lpar\gL_n}{R}}{\dD_{G}'}{\gG'}{}
			\qomma
			\ods{\unit}{\dD_i}{\gL_i\lpar\gN_i\lpar\gM_i}{}
			\quand
			\ods{\unit}{\dD_C}{\gC}{}
			\hskip-2em
		\end{equation*}
		for all $i\in\set{1,\ldots,n}$. 
		We let $\gK_i=\gL_i\lpar\gN_i$. 
		Then there is a derivation $\dD_G$ defined as
		\begin{equation*}
			\small
			\odn{\gC\Coons{
					\ods{\cneg\gP\connn{\gL_1\lpar\gN_1,\ldots,\gL_n\lpar\gN_n}}{
					}{
						\gL_1\lpar\cdots\lpar\gL_n\lpar\cneg\gP\connn{\gN_1,\ldots,\gN_n}}{\set{\ssdr}}}{R}}{
				\ssdr}{
				\ods{\gC\coons{\gL_1\lpar\cdots\lpar\gL_n}R}{\dD_{G}'}{\gG'}{}\lpar\cneg\gP\connn{\gN_1,\ldots,\gN_n}}{}
			\;.\hskip-3em
		\end{equation*}

		
		\item\label{split:e}
		In this case, the $\pdr$ rule is applied to a larger prime graph $\gQ$ in the context of the principal prime graph $\gP$.
		We have in this case that 
		$\dD$ is of shape
		\begin{equation*}
			\small
			\odn{\ods{\unit}{
					\dD'}{
					[\gG'';(\gN_1 ;[\gN_2;\gL_2];\cdots;[\gN_k;\gL_k])]
				}{}}{
				\pdr}{
				[\gG'';\gQ\connn{\gN_1,\ldots,\gN_k}; \gP\connn{\gM_1,\ldots,\gM_n}]
			}{}
			\hskip-2em       
		\end{equation*}
		An essential assumption is that all modules of $\gQ$ must be non-empty. 
		Hence, in this case we assume
		$\gG=\gG''\lpar\gQ\connn{\gN_1,\ldots,\gN_k}$ where $\gN_1, \dots, \gN_k$ are non-empty graphs, 
		$\gQ$ is a prime graph with $k=\sizeof\vQ > \sizeof\vP$ 
		such that w.l.o.g.
		$\gP\connn{\gM_1,\ldots,\gM_n} \isom \cneg{\gQ}\connn{ \unit, \gL_2, \ldots \gL_k }$ 
		for some possibly empty graphs $\gL_2, \ldots \gL_k$. 
		Observe at least one module of the prime connective $\cneg{\gQ}$ must be empty, and we let that w.l.o.g.\ to be the first module, otherwise $\cneg\gP$ and $\gQ$ are isomorphic, contradicting $\sizeof\vQ > \sizeof\vP$.
		
		We can apply the induction hypothesis in the form of Lemma~\ref{lem:splitting:multitens}, and we get a context $\gC'\coonso{R'}$ and graphs $\gK'_1, \dots, \gK'_k$, such that
		\begin{equation*}
			\ods{\gC'\coons{{\gK_1'\lpar \cdots\lpar \gK'_k}}R}{\dD_\gG''}{\gG''}{\GS}
			\qomma
			\ods{\unit}{\dD_1'}{\gK'_1\lpar \gN_1}{\GS}
			\qomma
			\ods{\unit}{\dD_j'}{\gK'_j\lpar(\gN_j\lpar \gL_j)}{\GS}
			\quand
			\ods{\unit}{\dD_C'}{\gC'}{\GS}
			\hskip-1em
		\end{equation*}
		for all $j\in\intset{2}{k}$.
		Now observe, that there is a graph $\gH$ (that is not necessarily prime), such that $\gQ\connn{\unit, \gK'_2 \lpar \gN_2,\ldots,\gK'_k \lpar \gN_k} \isom {\gH}\connn{ \gK'_2 \lpar \gN_2,\ldots,\gK'_k \lpar \gN_k }$ and $\cneg{\gQ}\connn{ \unit, \gL_2, \ldots \gL_k } \isom \cneg{\gH}\connn{ \gL_2, \ldots \gL_k }$. Hence, since for all $i \in \intset{2}{k}$ we have $\gN_i \neq \unit$, we can apply Lemma~\ref{lem:g}  to construct the following proof.
		\begin{equation}
                  \label{eq:split:e}
			\ods{
				\left(
				\ods{\unit}{\dD'_2}{\gK'_2 \lpar \gN_2 \lpar \gL_2}{}
				\ltens \ldots \ltens 
				\ods{\unit}{\dD'_k}{\gK'_k \lpar \gN_k \lpar \gL_k}{}\right)
			}{\dD_{\gH}}{
				{\gH}\connn{\gK'_2 \lpar \gN_2,\ldots,\gK'_k \lpar \gN_k} \lpar \cneg{\gH}\connn{\gL_2,\ldots,\gL_k}
			}{\mbox{\small Lemma \ref{lem:g}}}
		\end{equation}
		Now, observe that we have $\cneg{\gH}\connn{\gL_2,\ldots,\gL_k} \isom \gP\connn{\gM_1, \ldots , \gM_n }$.
		Furthermore, $\gsize{ \gK'_1 \lpar \cdots \lpar \gK'_k } \leq \gsize{ \gG''}$, 
		and hence we have $\gsize{ {\gH}\connn{ \gK'_1 \lpar \gN_1 , \dots , \gK'_k  \lpar \gN_k } } \leq \gsize{ \gG }$.
		Also
		$\gN_1 \neq \unit$ and $\gN_1$ does not appear in the conclusion of the proof~\eqref{eq:split:e} above; hence we have:
		\[
		\gsize{ {\gH}\connn{\gK'_2 \lpar \gN_2,\ldots,\gK'_k \lpar \gN_k} \lpar \gP\connn{\gM_1, \ldots , \gM_n } } < \gsize{ \gG \lpar \gP\connn{\gM_1,\ldots,\gM_n} }
		\]
		Therefore, we can apply the induction hypothesis to
                the proof in~\eqref{eq:split:e}, giving us one of the
                following two cases.
		\begin{itemize}
			\item 
			there is a context $\gC''\coonso{R'' }$ and graphs $\gK_1 , \ldots, \gK_n$ such that there are derivations 
			\begin{equation*}
				\ods{
					\gC''\coons{\cneg\gP\connn{\gK_1 \lpar \ldots \gK_n}}{R''}}{\dD_\gG'''}{\gQ\connn{\unit,\gK_2' \lpar \gN_2,\ldots,\gK_k' \lpar \gN_k}}{\GS}
				\qomma
				\ods{\unit}{\dD_i}{\gK_i \lpar \gM_i }{\GS} 
				\quand
				\ods{\unit}{\dD''_C}{\gC''}{\GS}
			\end{equation*}
			for all $i \in\intset1n$.
			We conclude by letting $\gC\coonso{R}= \gC'\coons{\gC''\coonso{R''}}{R'}$ and the derivation $\dD_G $ be defined as
			\begin{equation*}
				\odn{
					\gC'\Coons{
						\odn{
							\ods{\gC''\coons{\cneg\gP\connn{\gK_1 \lpar \ldots \gK_n}}{R''}}{\dD_G'''}{
								\gQ\Connn{\ods{\unit}{\dD'_1}{\gK'_1 \lpar \gN_1}{},
									\gK'_2 \lpar \gN_2,\ldots,\gK'_k \lpar \gN_k}
							}{}
						}{\ssdr}{
							[\gK'_1; \gK'_2; \ldots \gK'_k;
							\gQ\connn{\gN_1,\ldots,\gN_k}]
						}{}
					}{R'}
				}{\ssdr}{
					[\ods{
						\gC'\coons{
							[\gK'_1; \gK'_2; \ldots \gK'_k]
						}{R'}
					}{\dD''}{
						\gG''
					}{}
					;\gQ\connn{\gN_1,\ldots,\gN_k}]
				}{}
			\end{equation*}

			\item or we have a context $\gC''\coonso{R''}$ and graphs $ \gK_X$ and $ \gK_Y$ such that, w.l.o.g. there are derivations
			\begin{equation*}
				\ods{\gC''\coons{\gK_X \lpar \gK_Y}{R''}}{\dD'''_G}{\gQ\connn{\unit,\gK_2' \lpar \gN_2,\ldots,\gK_k' \lpar \gN_k}}{\GS}
				\qomma
				\ods{\unit}{\dD_X}{\gK_X\lpar\gM_1}{\GS}
				\qomma
				\ods{\unit}{\dD_Y}{\gK_Y\lpar\gP\connn{\unit, \gM_2,\dots,\gM_n}}{\GS}
				\quand
				\ods{\unit}{\dD''_C}{\gC''}{\GS}
			\end{equation*}
			We conclude by letting $\gC\coonso{R}= \gC'\coons{\gC''\coonso{R''}}{R'}$ and the derivation $\dD_G $ be defined as
			\begin{equation*}
				\odn{
					\gC'\Coons{
						\odn{
							\ods{\gC''\coons{\gK_X \lpar \gK_Y}{R''}}{\dD_G'''}{
								\gQ\Connn{\ods{\unit}{\dD'_1}{\gK'_1 \lpar \gN_1}{},
									\gK'_2 \lpar \gN_2,\ldots,\gK'_k \lpar \gN_k}
							}{}
						}{\ssdr}{
							[\gK'_1; \gK'_2; \ldots \gK'_k;
							\gQ\connn{\gN_1,\ldots,\gN_k}]
						}{}
					}{R'}
				}{\ssdr}{
					[\ods{
						\gC'\coons{
							[\gK'_1; \gK'_2; \ldots \gK'_k]
						}{R'}
					}{\dD''}{
						\gG''
					}{}
					;\gQ\connn{\gN_1,\ldots,\gN_k}]
				}{}
			\end{equation*}
			
		\end{itemize}

	\end{enumerate}

\end{proof}


\lemSplitMulti*
\begin{proof} 
	By induction on $n$.
	If $n=2$, then we conclude by Lemma \ref{lem:splitting:tens}.
	If $n\geq3$, then $A_1\ltens \cdots \ltens  A_n=A_1\ltens (A_2 \ltens A_3 \ltens \cdots \ltens A_n)$.
	Then, by Lemma \ref{lem:splitting:tens}, there are
	a context $\gC'\coonso{R'}$, and graphs $\gK_1$ and $\gK'$ 
	such that
	there are derivations
	\begin{equation*}
		\ods{\gC'\coons{\gK_1 \lpar \gK_Y}{R'}}{\dD_\gG'}{\gG}{\GS}
		\qomma		
		\ods{\unit}{\dD_1}{\gK_1 \lpar M_1}{\GS}
		\qomma
		\ods{\unit}{\dD'}{\gK' \lpar (M_2 \ltens M_3\ltens \cdots \ltens M_n)}{\GS}
		\quand		
		\ods{\unit}{\dD_C'}{\gC' }{\GS}
		\quadfs
	\end{equation*}	 
	By applying the same lemma inductively again to the proof $\dD'$ of $\gK' \lpar (A_2 \ltens A_3\ltens \cdots \ltens A_n)$, 
	we obtain a context $\gC''\coonso{R''}$ and graphs $K_2, \dots, K_n$ 
	such that 
	there are derivations
	\begin{equation*}
		\ods{\gC''\coons{{\gK_2\lpar \cdots\lpar \gK_n}}{R''}}{\dD_K'}{\gK'}{\GS}
		\qomma
		\ods{\unit}{\dD_i}{\gK_i\lpar\gA_i}{\GS}
		\quand
		\ods{\unit}{\dD_\gC''}{\gC''}{\GS}
		\hskip-1em
	\end{equation*}
	for all $i\in \intset2n$.
	We conclude by letting $\gC\coonso{R}=\gC'\coons{\gC''\coonso{R''}}{R'}$ and $\dD_G$ be the derivation
	\begin{equation*}
		\ods{\gC'\Coons{
				\odn{\gC''\coons{\gK_1 \lpar \cdots\lpar \gK_n}{R''}}
				{\ssdr}{{\gK_1 \lpar \ods{\gC''\coons{{\gK_2\lpar \cdots\lpar \gK_n}}{R''}}{\dD_K'}{\gK'}{}}}{}
			}{R'}}{\dD_\gG'}{\gG}{}
	\end{equation*}
\end{proof}

\lemConRed*
\begin{proof}
	
	The case when $\gA = \unit$ is trivial, since we can take $\gC = \gG$ and $R = S$ and $\gK = \unit$.

	Otherwise, without loss of generality 
	$\gG\coons{\gA}S\isom\gG''\lpar\gG'\coons{\gA}S$ for a graph $\gG'\coons{\gA}S$ which is neither a par nor empty.
	The base case is where $\gG' = \unit$ and $S = \emptyset$, in which case we can set $\gK = \gG''$ and $\gC = \unit$, and the derivations $\dD_C$ and  $\dD_G$ to be trivial. 
	The derivation $\dD_A$ is given by the proof of $\gG\coons{\gA}S$, since under these assumptions  $\gK \lpar A\isom\gG''\lpar A\isom \gG\coons{\gA}S$.
	
	If $\gG'$ is non-empty, we proceed by induction on the size of $\gG'\coons{\gA}S$ as follows.
	$\gG'\coons{\gA}S$ must be composed via a prime graph
	$\gP$ with $\gP \neq \lpar$.  Then, without loss of generality
	we can assume that
	$\gG\coons{\gA}S=G''\lpar\gP\connn{\gM_1\coons{A}{S'},\gM_2,\ldots,\gM_n}$. 
	Applying Lemma~\ref{lem:splitting:prime} gives us one of the following
	three cases:
	\begin{enumerate}[(A)]
		\item We have $\gC'\coonso{R'}$ and $\gK_1,\ldots,\gK_n$, such that
		\begin{equation*}
			\ods{\gC'\coons{\cneg\gP\connn{\gK_1,\ldots,\gK_n}}{R'}}{ \dD''_{G}}{\gG''}{\GS}
			\qomma
			\ods{\unit}{\dD_1}{\gK_1\lpar\gM_1\coons{A}{S'}}{\GS}
			\qomma
			\ods{\unit}{\dD_i}{\gK_i\lpar\gM_i}{\GS}
			\quand
			\ods{\unit}{\dD'_{C}}{\gC'}{\GS}
		\end{equation*}
		for $2\le i\le n$. We can apply the induction hypothesis to $\gK_1\lpar\gM_1\coons{A}{S'}$ and obtain
		$\gK$ and
		$\gC''\coonso{R''}$, such that
		\begin{equation*}
			\ods{\unit}{\dD_{C}''}{\gC''}{\GS}
			\qomma
			\ods{\unit}{\dD_A}{\gK\lpar\gA}{\GS}
			\quand          
			\ods{\gC''\coons{\gK\lpar\XXX}{R''}}{\dD'_{G}}{\gK_1\lpar\gM_1\coons{\XXX}{S'}}{\GS}
			\hskip-3em
		\end{equation*}
		for any $\XXX$. 
		We let $\gC\coonso R=\gC'\coons{\gC''\coonso{R''}}{R'}$, 
		the derivation $\dD_C$ is defined by Lemma~\ref{lem:context}, and the derivation $\dD_G$ is defined as 
		\begin{equation*}
			\small\footnotesize
			\odn{
				\gC'\Coons{
					\odn{
						\ods{\gC''\coons{\gK\lpar\XXX}{R''}}{ \dD'_{G}}{\gK_1\lpar\gM_1\coons{\XXX}{S'}}{}
						\ltens
						\ods{\unit}{\dD_2}{\gK_2\lpar\gM_2}{}
						\ltens
						\cdots
						\ltens
						\ods{\unit}{\dD_n}{\gK_n\lpar\gM_m}{}
					}{\pdr}{
						\cneg\gP\connn{\gK_1,\ldots,\gK_n}\lpar\gP\connn{\gM_1\coons{\XXX}{S'},\gM_2,\ldots,\gM_n}}{}}{R'}}{
				\ssdr}{
				\ods{\gC'\coons{\cneg\gP\connn{\gK_1,\ldots,\gK_n}}{R'}}{\dD''_{G}}{\gG''}{}
				\lpar\gP\connn{\gM_1\coons{\XXX}{S'},\gM_2,\ldots,\gM_n}}{}
			\qomma
			\hskip-2em
		\end{equation*}
		where $\gG\coons{\XXX}{S}=\gG''\lpar\gP\connn{\gM_1\coons{\XXX}{S'},\gM_2,\ldots,\gM_n}$.

		\item  We have $\gC'\coonso{R'}$ and $\gK_X$ and $\gK_Y$, such that
		\begin{equation*}
			\ods{\gC'\coons{\gK_X\lpar\gK_Y}{R'}}{\dD''_{G}}{\gG''}{\GS}
			\qomma
			\ods{\unit}{\dD_X}{\gK_X\lpar\gM_1\coons{A}{S'}}{\GS}
			\qomma
			\ods{\unit}{\dD_Y}{\gK_Y\lpar\gP\connn{\unit,\gM_2,\ldots,\gM_n}}{\GS}
			\quand
			\ods{\unit}{\dD'_{C}}{\gC'}{}
			\;.
			\hskip-3em
		\end{equation*}
		We apply the induction hypothesis to $\gK_X\lpar\gM_1\coons{A}{S'}$ and get $\gK$ and
		$\gC''\coonso{R''}$, such that
		\begin{equation*}
			\ods{\unit}{\dD''_{C}}{\gC''}{\GS}
			\qomma
			\ods{\gC''\coons{\gK\lpar\XXX}{R''}}{\dD'_{G}}{\gK_X\lpar\gM_1\coons{\XXX}{S'}}{\GS}
			\quand          
			\ods{\unit}{\dD_A}{\gK\lpar\gA}{\GS}
			\hskip-3em
		\end{equation*}
		for any $\XXX$. 
		We let $\gC\coonso R=\gC'\coons{\gK_Y\lpar\gP\connn{\gC''\coonso{R''},\gM_2,\ldots,\gM_n}}{R'}$
		and obtain $\dD_C$ from Lemma~\ref{lem:context}, and $\dD_G$ is as follows:
		\begin{equation*}
			\small
			\odn{
				\gC'\Coons{
					\odn{\gK_Y\lpar
						\gP\Connn{
							\ods{\gC''\coons{\gK\lpar\XXX}{R''}}{\dD'_{G}}{\gK_X\lpar\gM_1\coons{\XXX}{S'}}{}
							,\gM_2,\ldots,\gM_n}}{
						\ssdr}{
						\gK_X\lpar\gK_Y \lpar\gP\connn{\gM_1\coons{\XXX}{S'},\gM_2,\ldots,\gM_n}}{}
				}{R'}}{
				\ssdr}{
				\ods{\gC'\coons{\gK_X\lpar\gK_Y}{R'}}{\dD''_{G}}{\gG''}{}
				\lpar\gP\connn{\gM_1\coons{\XXX}{S'},\gM_2,\ldots,\gM_n}}{}
			\;.
			\hskip-5em
		\end{equation*}

		\item We have $\gC'\coonso{R'}$ and $\gK_X$ and $\gK_Y$, such that
		\begin{equation*}
			\ods{\gC'\coons{\gK_X\lpar\gK_Y}{R'}}{\dD''_{G}}{\gG''}{\GS}
			\qomma
			\ods{\unit}{\dD_X}{\gK_X\lpar\gM_2}{\GS}
			\qomma
			\ods{\unit}{\dD_Y}{\gK_Y\lpar\gP\connn{\gM_1\coons{A}{S'},\unit,\gM_3,\ldots,\gM_n}}{\GS}
			\quand
			\ods{\unit}{\dD_{C}'}{\gC'}{\GS}
			\;.
			\hskip-3em
		\end{equation*}
		We apply the induction hypothesis to $\gK_Y\lpar\gP\connn{\gM_1\coons{A}{S'},\unit,\gM_3,\ldots,\gM_n}$ and get $\gK$ and
		$\gC''\coonso{R''}$, such that
		\begin{equation*}
			\ods{\unit}{\dD''_{C}}{\gC''}{\GS}
			\qomma
			\ods{\unit}{\dD_A}{\gK\lpar\gA}{\GS}
			\quand
			\ods{\gC''\coons{\gK\lpar\XXX}{R''}}{\dD_{G}''}{\gK_Y\lpar\gP\connn{\gM_1\coons{\XXX}{S'},\unit,\gM_3,\ldots,\gM_n}}{\GS}
			\hskip-3em
		\end{equation*}
		for any $\XXX$. We let $\gC\coonso R=\gC'\coons{\gC''\coonso{R''}}{R'}$, the derivation $\dD_C$ be defined by Lemma~\ref{lem:context} and the derivation $\dD_G$ defined as follows
		\begin{equation*}
			\small\footnotesize
			\odn{
				\gC'\Coons{
					\odn{
						\ods{\gC''\coons{\gK\lpar\gX}{R''}}{
							\dD'_{G}}{
							\gK_Y\lpar
							\gP\Connn{
								\gM_1\coons{\XXX}{S'},
								\ods{\unit}{\dD_X}{\gK_X\lpar\gM_2}{},
								\gM_3,\ldots,\gM_n}}{}}{
						\ssdr}{
						\gK_X\lpar\gK_Y \lpar\gP\connn{\gM_1\coons{\XXX}{S'},\gM_2,\gM_3,\ldots,\gM_n}}{}
				}{R'}}{
				\ssdr}{
				\ods{\gC'\coons{\gK_X\lpar\gK_Y}{R'}}{\dD''_{G}}{\gG''}{}
				\lpar\gP\connn{\gM_1\coons{\XXX}{S'},\gM_2,\gM_3,\ldots,\gM_n}}{}
			\;.
			\hskip-5em
		\end{equation*}        
	\end{enumerate}
\end{proof}

Finally, it remains to give the proof of Atomic Splitting, which follows from context reduction (Lemma~\ref{lem:conred}) and splitting for prime graphs (Lemma~\ref{lem:splitting:prime}), as indicated in Figure~\ref{fig:complete-roadmap}.

\lemAtomSplit*
\begin{proof}
	The proof is similar to the one of Lemma~\ref{lem:splitting:prime}.
	%
	We assume  $\proves[\GS]{G\lpar a}$
	and aim to construct 
	$\gC\coonso{R}$, $\dD_G$, $\dD_C$
	as in the statement of the lemma. 
	Observe that $\gG\neq\unit$, otherwise $\gG\lpar a$ would not
	be provable in $\GS$. 
	We make a case analysis on the bottommost rule instance $\rr$ in~$\dD$, and follow the same pattern as in the case analysis in the proof of Lemma~\ref{lem:splitting:prime}. 
	\begin{enumerate}[(a)]

		\item 
		\label{proof:atomCase:inside}
		
		If rule $\rr$ acts inside $\gG$, then the derivation $\dD$ is of shape 
		\begin{equation*}
			\small
			\odv{\unit}
			{\dD'}{[\odn{\gG'}{\rr}{\gG}{};a]}{\GS}
			\hskip-2em
		\end{equation*}
		for some $\dD'$. 
		By Observation~\ref{obs:size} we know that  $\gless{\gG'}{\gG}$ ,
		then we apply the induction hypothesis on $\gG'\lpar a$.
		This gives us a context $\gC'\coonso {R'}$ and two derivations
		\begin{equation*}
			\ods{\gC'\coons{\cneg a}R}{\dD_{G}'}{\gG'}{\GS}
			\quand
			\ods{\unit}{\dD_{\gC}'}{\gC' }{\GS}
			\quadfs
		\end{equation*}
		We conclude by applying the rule 
		$\rr$ to the conclusion $\gG'$ of $\dD_{G'}$.

	      \item The last rule in $\dD$ is a $\ssdr$, such that $\dD$ is of the following shape (where $\gG=\gG''\lpar\gG'$ and $\gG'\neq\unit$):
		\begin{equation*}
			\odn{\ods{\unit}{
					\dD'}{[\gG'';a\coons{\gG'}S]}{\GS}}{
				\ssdr}{[\gG'';\gG';a]}{}
		\end{equation*}
		If the $\ssdr$ instance is not trivial, then $\gG''\lpar a\coons{\gG'}S\isom \gG''\lpar(a\ltens \gG')\isom\gG\coons{a}{V_{\gG'}}$, which makes this a case of Case~(c) below.
		
		\item 
		The last rule in $\dD$ is a $\ssdr$, such that $\dD$ is of the shape
		\begin{equation*}
			\odn{\ods{\unit}{
					\dD'}{[\gG\coons{a}S]}{\GS}}{
				\ssdr}{[\gG;a]}{}
		\end{equation*}
		By Lemma~\ref{lem:conred}
		there is a graph $\gK$ and a
		context $\gC\coonso R$ such that there are derivations
		\begin{equation*}
			\ods{\unit}{\dD_C'}{\gC'}{\GS}
			\quand
			\ods{\unit}{\dD^\ast}{\gK \lpar a}{\GS}
			\quand
			\ods{\gC'\coons{\gK\lpar\XXX}R}{\dD_G^X}{\gG\coons{\XXX}S}{\GS}
		\end{equation*}
		for any graph $\XXX$. Now we apply the induction hypothesis to $\dD^\ast$ to obtain
                a context $\gC''\coonso{R''}$ such that we have
                \begin{equation*}
                  \ods{\gC''\coons{\cneg a}{R''}}{\dD_\gK}{\gK}{\GS}
                  \quand
                  \ods{\unit}{\dD_C''}{\gC''}{\GS}
                \end{equation*}
		We let $\gC\coonso{R}= \gC'\coons{\gC''\coonso{R''}}{R'}$  and construct $\dD_\gG$ as
		\begin{equation*}
		  \ods{\gC'\Coons{
                      \ods{\gC''\coons{\cneg a}{R''}}{\dD_\gK}{\gK}{}
                    }{R'}}{\dD_G^\unit}{\gG\coons{\unit}S}{}
		\end{equation*}
		(Note that $\gK\isom\gK\lpar\unit$ and $\gG\isom\gG\coons\unit S$.)

		\item The last rule in $\dD$ is an $\aidr$ such that $\dD$ is of shape
		\begin{equation*}
		  \ods{\unit}{\dD_{G}'}{\gG'\lpar\odn{\unit}{\aidr}{\cneg a\lpar a}{}}{\GS}
		\end{equation*}
                i.e., $\gG\isom\gG'\lpar\cneg a$.
		We can conclude by letting $\gC=\unit$ (hence $\dD_C$ is trivial) and $\dD_\gG=\dD_\gG'\lpar\cneg a$.


		\item  \label{splitting:multitensor}
		If the last rule is a $\pdr$ which does not act inside $\gG$, such that $\gG\isom\gG'\lpar\gP\connn{\gN_1, \dots, \gN_m}$, then w.l.o.g. $\dD$ is of the shape
		\begin{equation*}
			\odv{\unit}
			    {\dD'}{[\gG';\odn{[([\gN_1 ; a];\gN_2 ; \cdots ;\gN_m)]}{\pdr}{[\gP\connn{\gN_1, \dots, \gN_m};a]}{}]}{\GS}
                            \quad.
			\hskip-2em
		\end{equation*}
		By Lemma~\ref{lem:splitting:multitens} there is a context $\gC'\coonso{R'}$ and  graphs $\gL_1, \dots , \gL_m$ such that there are derivations
		\begin{equation*}
			\ods{\gC'\coons{\gL_1\lpar\cdots \lpar \gL_m}{R'}}{\dD'_{G}}{\gG'}{\GS}
			\qomma
			\ods{\unit}{\dD_1}{\gL_1\lpar (a \lpar \gN_1) }{\GS}
			\qomma
			\ods{\unit}{\dD_i}{\gL_i\lpar\gN_i }{\GS}
			\quand
			\ods{\unit}{\dD'_{C}}{\gC'}{\GS}
			\;.
			\hskip-5em
		\end{equation*}
		for all $i\in \intset2m$.
		By inductive hypothesis on $\gL_1\lpar (a \lpar \gN_1)$ we have
		\begin{equation*}
			\ods{\gC''\coons{\cneg a}{R''}}{\dD'_{1}}{\gL_1\lpar \gN_1}{\GS}
			\quand
			\ods{\unit}{\dD''_{C}}{\gC''}{\GS}
			\;.
			\hskip-5em
		\end{equation*}
                for some context $\gC''\coonso{R''}$.
		We conclude by setting $\gC=\gC'\coons{ \gC''\coonso{R''}}{R'}$  since
		\begin{equation*}
			\odn{\gC'\Coons{
					\odn
					{(
						\ods{\gC''\coons{\cneg a}{R''}}{\dD'_1}{[\gL_1;\gN_1]}{};
						\ods{\unit}{\dD'_2}{[\gL_2;\gN_2]}{};
						\cdots;
						\ods{\unit}{\dD'_m}{[\gL_m;\gN_m]}{}
						)}
					{\ssdr}
					{\gL_1\lpar \cdots \lpar \gL_m\lpar \odn{(\gN_1; \gN_2;\cdots;\gN_m)}{\pdr}{\gP\connn{\gN_1, \dots, \gN_m}}{}}
					{}
				}{R'}}{
				\ssdr}{
				[\ods{\gC'\coons{\gL_1\lpar \cdots \lpar \gL_m}{R'}}{\dD'_{G}}{\gG'}{};
				\gP\connn{\gN_1, \dots, \gN_m}]
			}{}
			\hskip-4em
		\end{equation*}

		\qedhere

	\end{enumerate}

\end{proof}


\section{Requirements of an Analytic Proof System on Graphs}\label{sec:phil}

\newcommand{\requirement}[2]{\begin{center}\textbf{Requirement: #1}\\ \textit{#2}\end{center}}
\newcommand{\design}[2]{\begin{center}\textbf{Design choice: #1}\\ \textit{#2}\end{center}}
\newcommand{\consequence}[2]{\begin{center}\textbf{Consequence: #1}\\ \textit{#2}\end{center}}
\def\mp{modus ponens\xspace}

In this section, we reflect on design decisions in order to support our claim that we have defined a logic where we reason about graphs rather than formulas.
We present here an argument that is independent from the proof system that we developed, thus it is not necessary for the reader to accept deep inference a priori in order to accept the set of theorems proven by the logic $\GS$.
We reinforce the message that, for this logic without precedent (that is without a pre-existing semantics or proof system), we had success designing and justifying our system when we started from logical principles. Our approach of designing from principles is credible, since, as highlighted in the discussion on related work in Section~\ref{sec:relatedWork}, there does not appear to be an immediate generalisation of the semantics for some established logic on formulas to graph that yields a system with the logical properties we desire.

From the title, we desire:
\begin{quote}
	\textit{An {analytic} {propositional proof system} on {graphs}. }
\end{quote}
To understand fully the above statement we explain its two aspects. 
Firstly, we make precise what we mean by a \emph{propositional proof system}, by means of logical principles that such a proof system should conform to.
Secondly, we must explain what it means for such a propositional proof system to be \emph{analytic}, 
particularly since traditional definitions of analyticity do not lift immediately to our setting.
In the discussion that follows we show that all design decisions are widely accepted in logic, making is difficult to argue that $\GS$ is not a logic.

\subsection*{Graph isomorphism for propositional proofs}

In this study, we have restricted ourselves to simple undirected graphs (see Definition~\ref{def:graph}),
where vertices are labelled with positive or negative propositional atoms, such as $a$ and $\cneg{a}$.
This allows us to align with existing graphical representations of formulas if we restrict to \emph{cographs}, which are exactly those graphs generated by propositional formulas using the operations illustrated in (\ref{eq:TensPar}).

The first logical assumption we make is that isomorphic graphs (see Definition~\ref{def:iso}) are logically equivalent, which we expect to hold for all graphical logics.   
\begin{quote}
	\textbf{Extensionality requirement:} For pairs of isomorphic graphs $\gG$ and $\gH$, we have that $\vdash \gG \multimap \gH$ holds.
\end{quote}
Graph isomorphisms allows us to rename the underlying vertices of a graph, while preserving labels. Indeed in our diagrams of graphs, extensionality is implicitly appealed to, since we only show the labels of vertices, i.e., we quotient graphs by label-preserving isomorphisms.
The term ``{extensionality}'' is consistent with the idea that objects are equivalent if their externally visible properties, i.e., their labels and edges shown in diagrams, are the same.

This is one of the few principles we expect to be common to all logics on graphs.
To reinforce this belief, we acknowledge there are schools of philosophy that maintain that it is possible that $A\neq A$~\cite{Mates1968}. The essence of such arguments is that if $A$ is not well typed or does not exist then $A \neq A$ is a reasonable conclusion. But notice that, firstly, this is different from saying that $A = A$ does not hold, and, secondly, a metaphysical discussion on types or existence is perpendicular to the logic in this work.

Even if we assume extensionality as our sole logical principle,
since we aim to define a propositional proof system for our logic on graphs, we must take additional care to ensure that we satisfy the following principle.
\begin{quote}
	\textbf{Cook-Reckhow requirement:} Every rule is checkable in polynomial time.
\end{quote}
The above is a fundamental property of proof systems for propositional logic~\cite{Cook1979}.
As mentioned in Observation~\ref{obs:cook-reckhow}, since checking an explicit isomorphism is in $\textsf{P}$ but
currently 
there are no algorithms for finding graph isomorphisms in $\textsf{P}$, a formulation of a propositional proof system on graphs must make the isomorphisms explicit in the proof system.

\subsection*{Involutive negation and consistency}
The first proper design decision we make is that we insist on having a logic featuring an involutive negation, as found in most classical and linear logics (but not intuitionistic logic of course).
\begin{quote}
	\textbf{De Morgan requirement:} Negation should be involutive.
\end{quote}
Formally, an involution on graphs is a unary operator $\cneg{\left(\cdot\right)}$ that satisfies the property $\cneg{(\cneg{\gG})} = \gG$ for all graphs $\gG$. 
The assumption that we have an involutive negation means that De Morgan dualities hold for pairs of connectives on graphs we define (Observation~\ref{obs:deMorgan}).
For a proof system on  graphs, there are only two possible choices for an involution, namely the identity function and the graph complement function.

The use of the identity function to define an involutive negation is ruled out by the logical principle of consistency, which is that not all graphs are provable.
The formulation of consistency that we achieve for $\GS$ is as follows.
\begin{quote}
	\textbf{Consistency requirement:} For non-empty graphs $\gG$, if $\vdash \gG$ then $\not\vdash \cneg{\gG}$.
\end{quote}
To see why, the above principle rules out the identity function as negation, observe that extensionality ensures that there are some provable graphs say $\vdash \gG$. 
If the negation is defined by the identity, then $\vdash \cneg\gG$ holds, which violates the above consistency requirement.
Thereby, the assumptions thus far fix negation as \textit{graph complement} (Definition~\ref{def:dual}).
Observe also that the consistency of $\GS$ is indeed a corollary of cut elimination (Corollary~\ref{cor:consistency}).

\subsection*{Implication.}

Our next proper design decision is to materialise implication in our logic. We materialise implication ``$\gG$ implies $\gH$''
as ``not $\gG$ or $\gH$'', for some notion of disjunction, as in logics such as classical and linear logics.
We know already that negation is graph complement, but disjunction we have not yet defined.
In order for there to be a single implication, the disjunction used must be commutative (implications materialised using non-commutative  disjunction lead to a distinct left and right implication~\cite{Lambek1961}).
\begin{quote}
	\textbf{AC requirement:} Disjunction is commutative and associative.
\end{quote}

While there may be several elaborate choices for defining disjunction,\footnote{For an example, consider the way multiplicative disjunction is defined on coherent spaces \cite{girard:etal:89}.} we make the design decision that disjunction is defined as the disjoint union of graphs.
This design decision aligns with an established culture of using cographs to represent formulas~\cite{duffin:65,retore:03}.
Note that, by symmetry, we could have selected graph join as disjunction, which would simply interchange edges and non-edges throughout this paper without changing the meaning of the logic.

We are now able to materialise implication by means of negation and disjunction, as in classical and linear logic. More precisely, we have the following principle.
\begin{quote}
	\textbf{Material requirement:}
	Implication $\gG \multimap \gH$ is defined as 
	$\cneg \gG \lpar \gH$.
\end{quote}
This assumption should not be taken for granted, since various logics such as intuitionistic logic cannot materialise implication in this way. 
Also, this is not a guaranteed property of a logic satisfying De Morgan properties, since even if a disjunction is present it may not be the right disjunction to internalise negation. 
For example, linear implication cannot be materialised using additive disjunction and linear negation, hence the additive fragment of linear logic~\cite{girard:87} has no material implication.

We can now, by using only the principals we have laid out above, prove important examples independently of the proof system developed in the body of this paper.
By extensionality we have implication $\vdash a \multimap a$, where $a$ is the singleton graph.
Hence we know that $\vdash \cneg{a} \lpar a$ should hold in our logic.
This of course, agrees with the majority of proof systems.
Moving beyond graphs that correspond to traditional formulas, we also have theorems such as the following, by the same reasoning, i.e.,~the two disjoint sub-graphs are dual.
$$
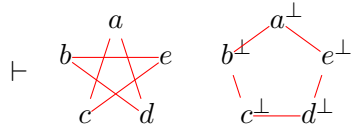
\def\diameter{20pt}
\vdash
\begin{array}{c}
	\begin{tikzpicture}
		\draw (90+0*360/5:\diameter) node {$\va1$};
		\draw (90+1*360/5:\diameter) node {$\vb1$};
		\draw (90+2*360/5:\diameter) node {$\vc1$};
		\draw (90+3*360/5:\diameter) node {$\vd1$};
		\draw (90+4*360/5:\diameter) node {$\ve1$};
	\end{tikzpicture}
	\quad
	\begin{tikzpicture}
		\draw (90+0*360/5:\diameter) node {$\vna1$};
		\draw (90+1*360/5:\diameter) node {$\vnb1$};
		\draw (90+2*360/5:\diameter) node {$\vnc1$};
		\draw (90+3*360/5:\diameter) node {$\vnd1$};
		\draw (90+4*360/5:\diameter) node {$\vne1$};
	\end{tikzpicture}
\end{array}
\edges{a1/c1,a1/d1,b1/e1,b1/d1,c1/e1}
\edges{na1/nb1,nb1/nc1,nc1/nd1,nd1/ne1,ne1/na1}
$$

For almost all logical systems, we require that implication is transitive,
since transitivity, known since antiquity as the \textit{hypothetical syllogism},  enables us to perform deduction.
\begin{quote}
	\textbf{Transitivity requirement:} If $\vdash F \multimap G$ and $\vdash G \multimap H$, then $\vdash F \multimap H$.
\end{quote}
Notice that transitivity allows us to apply 
\textit{modus ponens}. 
To be explicit, observe that the empty graph (denoted $\unit$) is the unit of the disjoint union operation on graphs. Furthermore $\unit$ is self-dual.
By the definition of implication, we have that $\vdash \gG$ holds whenever $\vdash \unit \multimap \gG$ does. 
Thereby, if we assume $\vdash \gG$ and $\vdash \gG \multimap \gH$ hold,
then by the definition of implication $\vdash \unit \multimap \gG$ holds and by transitivity 
$\vdash \unit \multimap \gH$ holds. 
By the same argument we can conclude that $\vdash \gH$ holds. 
This is exactly modus ponens.

Like consistency, transitivity of implication is a default decision, although we acknowledge that transitivity can fail for some logics featuring negation-as-failure~\cite{Bonner1990}.

The final feature we add is the context-free assumption.
For a logic with formulas this is the assumption that implication is preserved in all contexts, where a context is a formula with a hole in which another formula can be plugged, such as $\phi \ltens \coonso{}$.
For example, if $\psi \multimap \chi$, then $\phi \ltens \psi \multimap \phi \ltens \chi$.
More specifically, we mean all positive contexts, since negations reverse the direction of implication. However, in a logic satisfying the assumptions we have made, with De Morgan dualities and material implication, we can always reduce to a negation-normal-form, where all negations are pushed to the atoms, i.e., the labels on the vertices; thereby allowing us to range over all contexts.
Similarly, in modal logics extending $\mathsf K$ we have that if $\psi \multimap \chi$ then $\Box \psi \multimap \Box \chi$. The same holds for the diamond modality. Indeed, this holds even for non-classical modal logics~\cite{marin:str:aiml14}.

When we move from formulas to graphs the notion of a context must be generalised, where the obvious notion is introduced in Notation~\ref{not:context}.
In the broader philosophy of language (programming languages and natural languages), context freedom is not guaranteed, since the meaning of something may change in different contexts. 
However, logics are usually designed such that if there is some context in which an implication does not hold, then the implication itself does not hold; situations in knowledge representation where this fails is usually due to using moving between systems~\cite{Gabbay1998}. 
\begin{quote}
	\textbf{Context-free requirement}: For all positive contexts $\gC\coonso{R}$, if $\vdash G \multimap H$,
		then we have $\vdash \gC\coons{ G }{R} \multimap \gC\coons{ H}{R}$.
\end{quote}
The above property ensures that if we prove a theorem then it holds in any context. It is exploited explicitly in some proof systems to allow rules to be applied deep inside any module of a graph, which is the technique called \textit{deep inference} employed in this paper. However, the above principle is not specific to deep inference.

Thus we expect the following to be a theorem of a logic satisfying the above principles.
\begin{equation}
  \label{eq:context}
\vdash \left( G \multimap H \right) \multimap \left( \context{ G } \multimap \context{ H} \right)
\end{equation}
From this, we can establish non-trivial facts that we expect to hold for a logic on graphs, that are beyond the scope of formulas.
Consider the following example where we instantiate~\eqref{eq:context} with
\[
\gG \triangleq a 
\qquad
\gH \triangleq \emptyset
\qquad
\gC\coonso{\set{a}}
\triangleq
\begin{array}{c@{\quad\;\;}c}
	\va1 & \vb1 
	\\\\[-1ex]
	\vmodule1{\quad} & \vd1 
\end{array}
\edges{M1/a1,a1/d1,b1/d1}
\]
From the above instantiation we obtain the following theorem:
\[
\vdash
\left(
c  \limp \emptyset
\right)
\limp
\left(
\begin{array}{c@{\quad\;\;}c}
	\va1 & \vb1 
	\\\\[-1ex]
	\vc1 & \vd1 
\end{array}
\edges{a1/d1,b1/d1,c1/a1}
\multimap
\begin{array}{c@{\quad\;\;}c}
	\va1 & \vb1 
	\\\\[-1ex]
	\vmodule1{\unit}& \vd1 
\end{array}
\edges{b1/d1,M1/a1,a1/d1}
\right)
\]
Reorganising, according to the definitions of implication and negation that we have fixed, we obtain one of our running examples (see~\eqref{eq:exup}, \eqref{eg:proofA} and \eqref{eg:proof2}).
\[
\vdash
\begin{array}{c@{\quad\;\;}c}
	\va1 & \vb1 
	\\\\[-1ex]
	\vc1 & \vd1 
\end{array}
\edges{a1/c1,b1/d1,a1/d1}
\multimap
\begin{array}{c@{\quad\;\;}c}
	\va1 & \vb1 
	\\\\[-1ex]
	\vc1 & \vd1 
\end{array}
\edges{b1/d1,a1/d1}
\]
The above proposition would hold for any logic satisfying the principles laid out in this section.
There may be more propositions that hold other than those that are enforced by these principles; for instance, if we induce more principles from classical logic, we could prove strictly more theorems.
However, in this work, we make the design decision of aiming for the minimal system, which we state as a concluding principle that fixes our target logic.
\begin{quote}
	\textbf{Minimality requirement:}
	No further propositions hold, other than those forced by the other requirements. 
\end{quote}
We have remarked throughout this section that $\GS$ satisfies all the principles laid out, mainly as a consequence of cut elimination.
Minimality then follows from observing that all rules of $\GS$ are sound with respect to the same principles.

\subsection*{On analyticity}
Throughout this paper we provide evidence that we can design \textit{propositional proof systems} that achieve the above stated requirements.
However, 
designing a proof system is not the main challenge.
The main challenge is to design an \textit{analytic} propositional proof system.

Analyticity, the idea that propositions contains all the information required to in order to judge their validity with respect to some logical system, has long been debated by philosophers.
The design of modern analytic proof calculi is widely considered to begin with the \textit{sequent calculus}, as developed in 1935~\cite{gentzen:35:I,gentzen:35:II},
with improvements incorporated by Girard~\cite{girard:87}.

The rules of the sequent calculus, such as those for $\MLLm$ 
presented in Figure~\ref{fig:mll}, are considered to be analytic since they satisfy the \textit{subformula property}. 
The subformula property states that 
in each rule, every formula that occurs in one premise sequent, occurs as a subformula in the conclusion sequent.
The subformula property facilitates proof search, since there are only finitely many subformulas to consider during proof search.
The rules in Figure~\ref{fig:mll} clearly satisfy the subformula property; whereas the $\cutr$ rule, below, which generalises transitivity, and hence also \textit{modus ponens}, does not in general satisfy the subformula property.
\[
\vliiinf{\cutr}{}{\Gamma, \Delta}{\Gamma, \phi}\quad{\cneg{\phi}, \Delta}
\]
If we start with some conclusion $\Gamma, \Delta$ and try to apply the $\cutr$ rule above, there are infinitely many formulas $\phi$ to choose from.
Similarly, the transitivity assumption requires insight external to the system to know which formulas to introduce.
Thus the subformula property is effectively avoiding infinite branching in the proof search space.
This is a fundamental reason for proving cut elimination when designing a proof system: on one hand, we aim to ensure that the basic principles of deduction such as modus ponens may be applied; on the other hand, we also want to show that deductive rules such as $\cutr$ and modus ponens, which are not well-behaved with respect to proof search are \textit{admissible}.

The sub-formula property does not lift immediately to the setting of deep inference.
Since deep inference is necessary for $\GS$, we make use of an alternative definition of analyticity~\cite{Bruscoli2009,Bruscoli2016}.
\begin{quote}
	\textbf{Analyticity requirement:} For graph $\gG$ and rule $\rr$, there is an $n$ such that, for all contexts $\gC\coonso{R}$, there are at most $n$ graphs $\gH$ such that $\odn{\gC\coons{\gH}{R}}{\rr}{\gC\coons{\gG}{R}}{}$.
\end{quote}
This ensures a rule is guaranteed to be finitely branching regardless of the context in which the rule is applied.
This property follows immediately for $\GS$ from inspecting the rules.

Thus, in this work, we have shown that we can design a proof system on graphs satisfying a notion of analyticity.
The fact that traditional methods (such as the sequent calculus) for developing analytic proof calculi fail here is a surprise, particularly because we have targeted the minimal logic on graphs satisfying some widely accepted logical principles.
The logical principles we have based our design decisions on we consider to be difficult to argue against --- which does not prevent one from exploring alternative design decisions that may lead to further logics on graphs.
The only assumption that may be strongly argued about is the use of the empty graph as a self-dual unit; which can be regarded as a simplifying assumption in our initial design.
That design decision can be justified by \textit{conservativity} (Theorem~\ref{thm:conservativity}), in the sense that there are established formula-based logics featuring such a self-dual unit.

\end{document}